\documentclass[onecolumn, draftclsnofoot,11pt]{IEEEtran}

\usepackage[letterpaper,top=2.2cm,bottom=2.2cm,left=2.2cm,right=2.2cm,marginparwidth=1.75cm]{geometry}
\usepackage[table]{xcolor}
\usepackage{booktabs}
\usepackage{pgf} 
\newcommand{\heatcell}[1]{%
  \pgfmathsetmacro{\gamma}{1.8}%
  \pgfmathsetmacro{\hue}{0.33*pow(#1/100, \gamma)}
  \edef\temp{\noexpand\cellcolor[hsb]{\hue,0.55,1.0}}%
  \temp #1\%%
}
\usepackage{cite} 
\usepackage{amssymb}
\usepackage{amsthm}
\usepackage[english]{babel}
\usepackage[most]{tcolorbox}
\usepackage{amsmath}
\usepackage{amsfonts}
\usepackage{url}
\usepackage{graphicx} 
\usepackage{bm}
\usepackage{algpseudocode}
\usepackage[linesnumbered,ruled,vlined,procnumbered]{algorithm2e}
\usepackage{setspace}
\usepackage{hyperref}
\usepackage{caption}
\usepackage{subcaption}
\usepackage{lineno}
\usepackage{bibunits}
\newtheorem{definition}{Definition}
\newtheorem{lemma}{Lemma}

\newtheorem{theorem}{Theorem}
\newtheorem*{theorem*}{Theorem}
\newtheorem{remark}{Remark}
\newtheorem{prop}{Proposition}

\global\long\def\RR{\mathbb{R}}

\global\long\def\EE{\mathbb{E}}

\newcommand{\bfm}{\mathbf{m}}

\newcommand{\bfw}{\mathbf{w}}

\newcommand{\rmd}{\mathrm{d}}

\newcommand{\rmA}{\mathrm{A}}
\newcommand{\rmB}{\mathrm{B}}
\newcommand{\rmC}{\mathrm{C}}
\newcommand{\rmD}{\mathrm{D}}
\newcommand{\rmH}{\mathrm{H}}

\newcommand{\rmF}{\mathrm{F}}
\newcommand{\rmI}{\mathrm{I}}

\newcommand{\rmT}{\mathrm{T}}

\newcommand{\rmX}{\mathrm{X}}
\newcommand{\rmY}{\mathrm{Y}}
\newcommand{\rmZ}{\mathrm{Z}}

\def \tensor {\otimes}

\def \ds {\rmd s}

\def\<{\langle}
\def\>{\rangle}

\def\tr{\mathrm{tr}}
\def\Pr{\mathrm{Pr}}
\newcommand{\ket}[1]{|#1\>}
\newcommand{\bra}[1]{\<#1|}

\def \calD {\mathcal{D}}
\def \calE {\mathcal{E}}

\def \calG {\mathcal{G}}

\def \calL {\mathcal{L}}
\def \calM {\mathcal{M}}
\def \calN {\mathcal{N}}
\def \calO {\mathcal{O}}

\def \calS {\mathcal{S}}
\def \calT {\mathcal{T}}

\def \calW {\mathcal{W}}
\def \calX {\mathcal{X}}

\def \sfB {\mathsf{B}}

\def \sfD {\mathsf{D}}

\def \sfN {\mathsf{N}}

\def \sfd {\mathsf{d}}

\usepackage{stackengine}

\def \eqand {\text{ and }}
\newcommand{\addon}[1]{\textcolor{black}{#1}}

\def \genei {\mathsf{g}_i}
\def \genej {\mathsf{g}_j}
\def \dt {\mathrm dt}
\def \imag {\mathrm i}

\title{\huge Quantum Hamiltonian-Based Generative Modeling of Single-Cell Transcriptomics for Gene Regulatory Network Inference}

\author{
\IEEEauthorblockN{Mohammad Aamir Sohail\IEEEauthorrefmark{1},  Ranga R. Sudharshan\IEEEauthorrefmark{2}, S. Sandeep Pradhan\IEEEauthorrefmark{1}, Arvind Rao\IEEEauthorrefmark{2}}\\
      
      \vspace{5pt}\IEEEauthorrefmark{1}\normalsize Department of EECS, University of Michigan, Ann Arbor, USA\\
 \IEEEauthorrefmark{2}\normalsize Department of Computational Medicine and Bioinformatics, University of Michigan, Ann Arbor, USA }
\date{August 2025}

\begin{document}

\maketitle
\pagenumbering{arabic}
\begin{abstract}
We introduce a novel quantum Hamiltonian-based gene expression model (QHGM), a generative framework for modeling pseudotime-ordered single-cell gene expression data. In QHGM, gene interactions are encoded via a parameterized Hamiltonian, and the outcomes of quantum measurements provide a discrete representation of the gene expression profile. To learn the Hamiltonian parameters and infer gene regulatory networks (GRNs), we develop a scalable variational quantum algorithm for network inference (VQ-Net) based on empirical risk minimization. We derive finite-sample recovery guarantees for accurate parameter estimation, demonstrating polynomial scaling with the number of genes. Experiments on synthetic data demonstrate accurate GRN recovery, with VQ-NET achieving over 25\% improvement in edge recovery and over 50\% improvement in parameter-sign recovery compared to state-of-the-art classical methods. We further apply the framework to glioblastoma scRNA-seq data, where it identifies biologically plausible regulatory interactions associated with cancer progression, highlighting the potential of \textit{quantum-like} modeling beyond classical probabilistic frameworks. 
\end{abstract}

\section{Introduction}
The transcriptional state of a cell is highly dynamic and tightly regulated by complex interactions between genes and transcription factors \cite{kitano2002systems,barabasi2004network,alon2019introduction}. Understanding these gene regulatory relationships is critical for deciphering cellular behavior and function, particularly because disruptions in gene interactions can contribute to the development and progression of diseases such as cancer \cite{ray2026gene, hanahan2011hallmarks, koumadorakis2022consensus}.  
Recent advances in single-cell RNA sequencing (scRNA-seq) \cite{tang2009mrna,svensson2020curated} have greatly facilitated the development of computational tools for inferring gene regulatory networks (GRNs). The scRNA-seq data captures gene expression levels across cells, transient cellular states, and heterogeneous developmental trajectories. This provides an enhanced view of transcriptional regulation and cellular dynamics, enabling more accurate modeling of gene interactions across diverse cellular states and conditions \cite{stuart2019comprehensive,trapnell2015defining,gupta2020biological,van2020trajectory,ding2022temporal}. As a result, inferring GRNs from high-dimensional single-cell data has become one of the central challenges in systems biology.

Numerous classical statistical methods have been developed for GRN inference, including correlation-based methods \cite{kim2015ppcor,specht2017leap}, information-theoretic approaches \cite{margolin2006aracne}, regression and tree-based ensemble methods \cite{huynh2010inferring,papili2018sincerities}, and Bayesian network models \cite{sanchez2018bayesian}. These methods effectively model certain statistical dependencies of gene regulation, with each classical approach possessing its own strengths and limitations relative to others \cite{ali2025review,akers2021gene}. However, they are fundamentally rooted in classical probability theory, which assumes that interactions are governed by commutative and context-independent stochastic processes. In particular, classical probability assumes that the order in which events or measurements occur does not affect the resulting joint probability distribution. 
Furthermore, it assumes that the probability of an event under a combination of mutually exclusive contexts is given by a weighted linear combination of the conditional probabilities associated with the individual contexts, according to the law of total probability \cite{ross1998first}.

However, increasing experimental evidence suggests that these assumptions may not always hold in a cellular ecosystem. In particular, many biological processes exhibit order-dependent behavior, closely related to the notion of noncommutativity, in which different sequences of regulatory operations can yield distinct outcomes \cite{letsou2016noncommutative,chan2013peripheral,sayegh2005visualization}. For instance, the order in which transcription factors act or signaling pathways are activated can significantly alter downstream gene-expression responses and cellular behavior, as observed in hematopoietic stem-cell differentiation \cite{iwasaki2006order}, neuronal firing dynamics \cite{moro2010neuronal}, and signaling pathways such as ERK \cite{albeck2013frequency}, NF-$\kappa$B \cite{nelson2004oscillations,kellogg2015digital}, and p53 \cite{lahav2004dynamics}.
In addition, violations of the classical law of total probability have been observed in biological systems. For example, experimental studies of \textit{E. coli} metabolism under different sugar environments have revealed strong context-dependent gene-expression responses that cannot be fully explained using classical probability theory \cite{basieva2011quantum,asano2012quantum}. Furthermore, increasing evidence suggests that transcriptional regulators can simultaneously exhibit activating and repressing behavior \cite{guharajan2021quantifying,delrosso2023large,martinez2025emergence}. In cancer progression, certain cell populations have also been observed to exist in hybrid phenotypic states, simultaneously expressing features of multiple cellular programs \cite{alvarez2024quantum,neftel2019integrative}. Both examples suggest characteristics analogous to quantum superposition.
These observations highlight the necessity for more expressive probabilistic frameworks capable of modeling order effects, contextual regulation, and superposition-like behavior, thereby motivating the use of \textit{quantum-like} modeling.

The \textit{quantum-like} modeling paradigm \cite{khrennikov2004quantum,khrennikov2006quantum,khrennikov2015quantum} applies the mathematical formalism of quantum probability theory to systems that may not be inherently quantum-scale, but exhibit behavior that cannot be adequately characterized using classical probabilistic models. Importantly, \textit{quantum-like} modeling does not assume that genes or cells obey quantum mechanics at the physical level. Rather, it uses mathematical tools such as non-commutative observables, superposition states, and generalized quantum measurements modeled by positive-operator-valued measures (POVMs) to represent contextual, order-dependent, and interference-like stochastic phenomena. As put forth in \cite{plotnitsky2023quantum}, the quantum formalism can be applied ``not only to physical systems, but to systems of any origin --- whether biological, social, or financial --- as long as their behavior exhibits some distinguishing features of quantum systems.'' Consequently, the \textit{quantum-like} paradigm has been explored across diverse fields beyond physics, including cognitive science, neuroscience, finance, and decision theory \cite{khrennikov2010ubiquitous,melkikh2015nontrivial,haven2017palgrave,pothos2022quantum,plotnitsky2023quantum}.

Recent research has started to explore \textit{quantum-like} models for biological systems, including quantum random walks for disease-gene ranking and biomolecular network analysis \cite{Dubovitskii2025}, quantum regression networks for studying gene interactions in Alzheimer’s disease \cite{konar2025alz}, and quantum kernel methods for capturing cell-cell interactions \cite{romero2025quantum}.
More closely related to GRN inference from scRNA-seq data, a quantum circuit model has been proposed in \cite{RomanVicharra2023}, in which gene interactions are modeled via C-R$_{\mathrm{Y}}$ gates \cite{nielsen2010quantum}, yielding a circuit depth that scales polynomially with the network size.
However, the circuit construction is sensitive to gene ordering because it prioritizes genes with higher activation ratios in the gene-to-qubit mapping. The method also uses KL divergence as a loss function \cite{Cover2006}, which requires evaluating joint probability distributions over genes, resulting in computational costs that scale exponentially with network size.

Motivated by these limitations, we propose a quantum Hamiltonian-based framework for modeling GRNs. In this framework, each gene is represented as a qubit, and the regulatory interactions among genes are encoded through a parameterized interaction Hamiltonian. The transcriptional state of a cell is modeled as a many-qubit \textit{quantum-like} state evolving under the parameterized Hamiltonian. The Hamiltonian parameters encode both the direction and strength of regulatory interactions, so that learning the GRN becomes equivalent to learning the corresponding Hamiltonian parameters.

This connection links our framework to the broader area of quantum Hamiltonian learning (QHL), which studies the problem of inferring the parameters of a Hamiltonian governing the dynamics of a many-body system from measurement data \cite{burgarth2017evolution,wang2017experimental,kwon2020magnetic,wang2020machine}. A broad range of efficient QHL methods have been developed using local observables, Gibbs states, short-time dynamics, and time-resolved measurements \cite{anshu2021sample,haah2022optimal,gu2024practical,chen2025learning,Bakshi2024LearningHamiltonians,qi2019determining,dupont2019eigenstate,bairey2020learning,evans2019scalable,bairey2019learning,cao2020supervised,stilck2024efficient,zubida2021optimal,hangleiter2024robustly,gupta2023hamiltonian,weibe2014QHL,Wiebe2015QuantumBootstrapping,wilde2022scalably,liu2025optimal}. Recent advances further include noise-resilient and Heisenberg-limited Hamiltonian learning methods capable of highly sample-efficient parameter estimation \cite{yu2023robust,Huang2023LearningManyBody,Hu2025AnsatzFreeHamiltonian,Sinha2025ImprovedHamiltonianLearning,Baran2025HeisenbergLimitedQuantumHamiltonian}.

Observe that existing QHL approaches are primarily designed for physical quantum many-body systems and often rely on assumptions such as access to entangled initial states, eigenstates of the Hamiltonian, and randomized multi-qubit measurements. In contrast, single-cell transcriptomic data only provide single-qubit measurement outcomes corresponding to individual gene expression readouts. Consequently, these existing QHL frameworks cannot be directly applied in our setting. 
To address this gap, we develop a biologically motivated Hamiltonian learning framework using fixed local informationally complete POVM (IC-POVM) measurements, and analyze it within a statistical learning-theoretic framework. Below, we summarize the main contributions of this work.

\textbf{$1.$ Quantum Hamiltonian-Based Generative Model:}
We introduce a quantum Hamiltonian-based single-cell gene-expression model (QHGM) (see Fig.~\ref{fig:QHGM}), where genes are represented as qubits and regulatory interactions are encoded through biologically interpretable Hamiltonian terms constructed using tensor products of computational basis states. The transcriptional state of a cell evolves according to Hamiltonian dynamics. Following the evolution, a fixed local IC-POVM measurement is performed on each qubit, with each outcome corresponding to the gene's expression level. Together, these outcomes define the cell's gene expression profile. Additionally, we utilize pseudotime \cite{stassen2021generalized}, which organizes cells along inferred developmental trajectories derived from scRNA-seq data, serving as an approximate analog for physical evolution time.

\textbf{$2.$ Scalable Quantum Algorithm for Network Inference:}
We develop a scalable variational quantum algorithm for network inference, namely VQ-Net, for learning Hamiltonian parameters from scRNA-seq data (see Fig.~\ref{fig:algorithm} and Alg.~\ref{alg:VQNet}). The algorithm uses a quantum circuit to perform Hamiltonian evolution, thereby avoiding the exponential computational overhead of classical simulation methods, followed by fixed local IC-POVM measurements. In addition to the quantum circuit, it uses a classical optimizer to minimize the negative log-likelihood loss over mini-batches of scRNA-seq data collected across multiple pseudotime bins, enabling scalable network inference without explicitly evaluating the exponentially large joint probability distributions over genes, as required by divergence-based loss functions.

\textbf{$3.$ Application to Glioblastoma scRNA-seq Data:}
We apply our framework to scRNA-seq data from glioblastoma (GBM) patients \cite{ruiz2025charting}, focusing on the gene regulatory programs governing differentiation of OPC-like cells (see Fig.~\ref{fig:finalgrn}). GBM is the most common malignant brain tumor, with a poor prognosis and complex cellular heterogeneity \cite{grochans2022epidemiology, rodriguez2022glioblastoma}. To our knowledge, this is the first application of quantum Hamiltonian-based modeling to infer biologically relevant regulatory networks in cancer research. Our results reveal potential GRN structures and interaction patterns that reflect the cellular plasticity within malignant OPC-like populations, opening new avenues for quantum-driven exploration of information flows in biological systems.

\textbf{$4.$ Sample-Efficient Hamiltonian Learning with Statistical Guarantees:}
We derive finite-sample statistical guarantees characterizing the number of pseudotime bins and measurement samples required to accurately recover the Hamiltonian parameters with high probability (see Theorems \ref{thm:emp_strongconvexity} and \ref{thm:unif}). We show that if the number of pseudotime bins and the number of measurement samples per pseudotime bin grow polynomially with the number of genes, then the parameter-recovery error scales inversely with the square root of the number of measurement samples and grows only logarithmically with the number of pseudotime bins.

\textbf{$5.$ Empirical Evaluations on Synthetic Data and Comparison with Classical GRN Inference Methods:} We provide numerical evaluations of VQNet on synthetic gene-expression data generated using QHGM (see Fig.~\ref{fig:synthetica_data_summary}). For synthetic GRNs with \(12\) genes, VQ-Net achieves over \(95\%\) weight-recovery accuracy using only \(25\) pseudotime bins and \(1000\) measurement samples per bin, providing empirical support for the polynomial sample-complexity scaling predicted by our theoretical analysis.
 We compare VQNet against state-of-the-art classical GRN inference methods on QHGM-generated data across varying sample sizes and network sparsities (see Fig.~\ref{fig:classical_performance}). VQNet achieves over \(25\%\) improvement in edge-recovery accuracy and over \(50\%\) improvement in parameter-sign recovery compared to state-of-the-art classical GRN inference methods, consistent with the premise that classical methods struggle to accurately recover interaction structures arising from contextual and order-dependent dynamics.

\section{Main Results}
We have structured our findings into four main subsections. The first subsection introduces a quantum Hamiltonian-based generative model for single-cell gene expression. The second subsection discusses the scalable VQ-Net algorithm. In the third subsection, we analyze the sample complexity of VQ-Net. Finally, the fourth subsection presents numerical results from synthetic data, comparisons with traditional network inference methods, and experiments on real scRNA-seq data.

\captionsetup[figure]{font=small, labelfont=bf, width=\textwidth}
\begin{figure}[!htb]
    \centering
    \includegraphics[width=\textwidth]{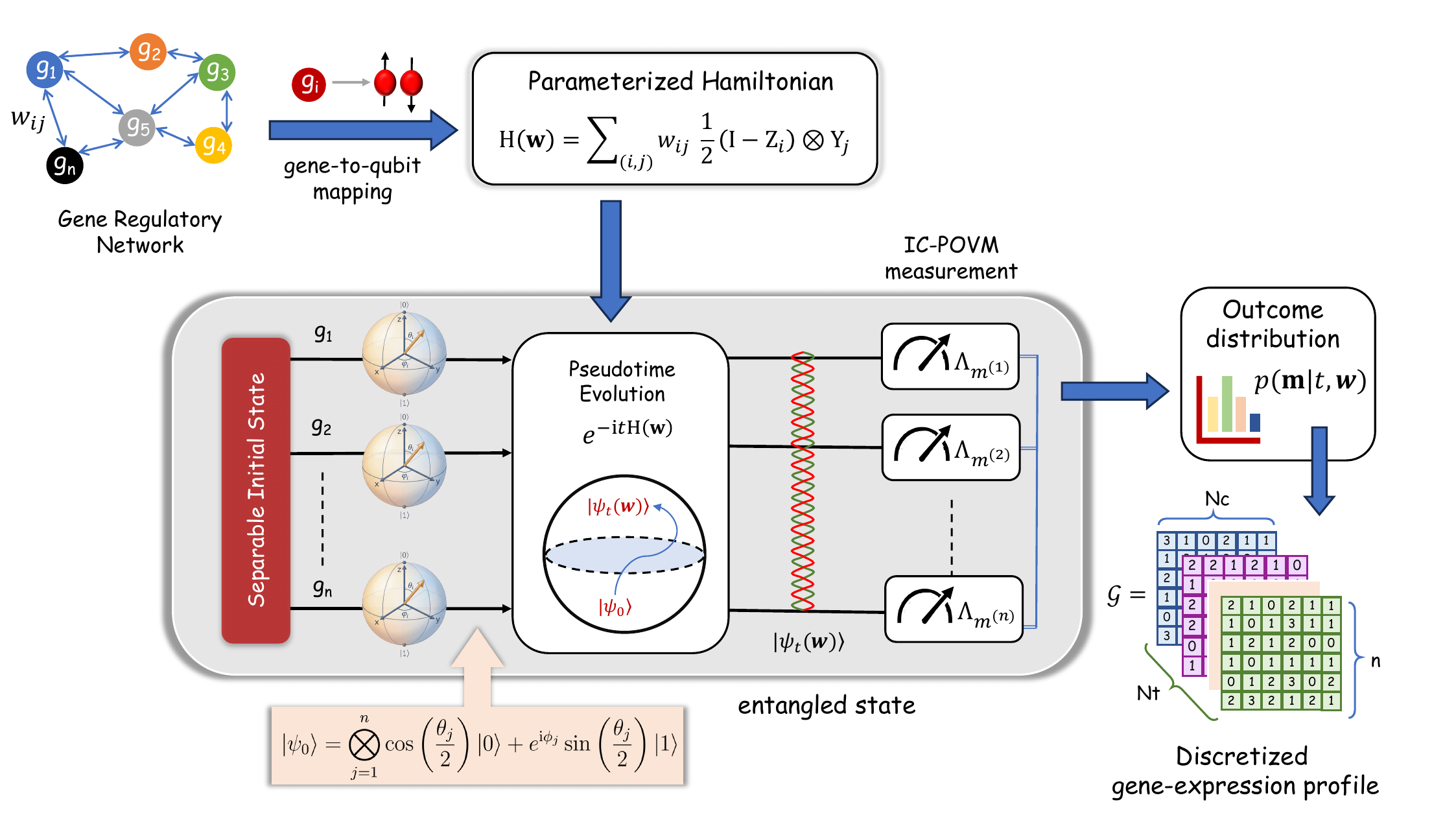}
    \caption{Overview of the quantum Hamiltonian-based gene-expression model (QHGM). A gene regulatory network (GRN) is mapped
to a parameterized Hamiltonian $\rmH(\bfw)$, where the presence of gene 
$\genei$ induces the action of the Pauli-
$\rmY$ operator on gene 
$\genej$ with weight $w_{ij}$. The model begins from an initial separable state, representing independent gene states.
As the system evolves along pseudotime, correlations between genes are gradually introduced, resulting
in an entangled quantum state. At each pseudotime point, this state is measured using a
fixed single-qubit IC-POVM, producing a probability distribution
$p(\bfm|t,\bfw)$ over measurement outcomes. Collecting repeated measurement outcomes at each pseudotime point yields a discretized
gene-expression profile, denoted as $\calG$, with dimension $(\sfN_t,\sfN_c,n)$ that serves as the observable data for inference. Here, $\sfN_t$ is the number of pseudotime bins, $\sfN_c$ is the number of cells per
bin, and $n$ is the number of genes in the network.
}
    \label{fig:QHGM}
\end{figure}

\subsection{Quantum Hamiltonian-Based Generative Model for Single-Cell Gene Expression}
In this subsection, we present the \emph{quantum Hamiltonian-based gene expression model} (QHGM) (see Fig.~\ref{fig:QHGM}), which serves as a statistical generative model for gene expression data.
In QHGM, for simplicity of exposition, we represent each gene as a qubit\footnote{The framework can be naturally generalized to higher-dimensional quantum systems (qudits), enabling multi-level representations of gene expression states and richer encoding of cellular information.}. The computational basis states $\ket{1}$ and $\ket{0}$ correspond to the transcriptional state of the gene, indicating whether it is expressed or unexpressed in a cell, respectively.
The model consists of three main components: $(i)$ a Hamiltonian that encodes the regulatory structure of the GRN, $(ii)$ initial state preparation and pseudotime state evolution, and $(iii)$ an IC-POVM measurement, each described in detail below. 

\textit{$1.$ Hamiltonian for GRNs:}  
The regulatory interactions between genes are encoded in a parameterized Hamiltonian as
\begin{equation}
    \rmH(\bfw)
= \sum_{(i,j)\in \calE}
 w_{ij} \; \tfrac{1}{2}(\rmI  -  \rmZ_i) \tensor \rmY_j, \label{eqn:grn_hamiltonian}
\end{equation}
where $\calE := \{(i,j) : i,j \in \{1,2,\cdots,n\}, \, i \neq j\}$ and $n$ is the number of genes in the network. 
The weights $\{w_{ij}\}$ capture both the strength and direction of regulatory influence, with each $w_{ij}$ quantifying the effect of gene $\genei$ on gene $\genej$. 
A positive $w_{ij}$ indicates activation, meaning that the presence of gene $\genei$ promotes the expression of $\genej$, whereas a negative $w_{ij}$ implies that $\genei$ suppresses $\genej$. 
The magnitude $|w_{ij}|$ reflects the strength of regulation, with larger absolute values corresponding to stronger activating or repressing effects. Each weight is bounded to maintain biologically meaningful interaction strengths \cite{szklarczyk2025string,muller2023expanding,wingender1996transfac}.
We exclude self-interaction links, as they correspond to intrinsic gene dynamics rather than inter-gene regulation. 
We provide additional details on the construction of the GRN Hamiltonian in the Methods section. 

 \textit{$2.$ Initial State Preparation and Pseudotime Evolution:}  In QHGM, the system is initialized in a separable quantum state given by
\begin{equation}
    \ket{\psi_0} := \bigotimes_{i=1}^{n} \ket{\psi^{(i)}_0}, 
\quad 
\text{where} \quad 
\ket{\psi^{(i)}_0} = \cos\theta_i\,\ket{0} + e^{\mathrm{i}\phi_i}\sin\theta_i\,\ket{1}. \label{eqn:initial_state}
\end{equation}
Each qubit-state $\ket{\psi^{(i)}_0}$ represents the transcriptional state of $\genei$ as a point on the Bloch sphere \cite{nielsen2010quantum}, 
where 
\(\sin^2\theta_i\) and \(\cos^2\theta_i\) denote the prior probabilities of $\genei$ being expressed or unexpressed, respectively. 
\addon{The angle $\phi_i$ encode its initial directional or kinetic phase information \cite{la2018rna}.} This initialization reflects an assumption that genes exist independently and uncorrelated, with correlations emerging dynamically through the regulatory interactions 
encoded in the Hamiltonian\footnote{More generally, the framework can incorporate entangled initial states, by using prior biological information inferred from complementary multi-omics or multimodal data.}. The system then evolves following the Schrödinger equation 
\(
{\mathrm{d}\ket{\psi_t}}/{\dt} = -\imag \rmH(\bfw) \ket{\psi_t},
\)
with the time-independent Hamiltonian \(\rmH(\bfw)\), yielding a final state $\ket{\psi_t(\bfw)} = \exp(-\imag t\,\rmH(\bfw))\,\ket{\psi_0},$
where $t$ represents the pseudotime, which corresponds to the temporal progression of cell-state transition. Pseudotime is the approximate position of cells along a trajectory that quantifies the relative progression of the underlying biological process \cite{trapnell2014dynamics}. 
  
\textit{$3.$ Measurement and Gene Expression Readout:} Following pseudotime evolution, 
each gene (qubit) is measured using a single-qubit IC-POVM 
$\{\Lambda_m\}_{m=0}^3$ given as  
\begin{align}\label{eqn:ic-povm_qgrn}
\Lambda_0&=\tfrac14\!\big(\rmI+\tfrac12 \rmX+\tfrac{\sqrt3}{2} \rmZ\big),
\hspace{49pt}\Lambda_1=\tfrac14\!\big(\rmI-\tfrac12 \rmX-\tfrac{1}{\sqrt2} \rmY+\tfrac12 \rmZ\big),\nonumber\\
\Lambda_2&=\tfrac14\!\big(\rmI-\tfrac12 \rmX+\tfrac{1}{\sqrt2} \rmY-\tfrac12 \rmZ\big),
\hspace{20pt}\Lambda_3=\tfrac14\!\big(\rmI+\tfrac12 \rmX-\tfrac{\sqrt3}{2} \rmZ\big).
\end{align}
The IC-POVM elements are designed such that for each measurement outcome $m \in \{0,1,2,3\}$, there is a corresponding \emph{expression score} $\tau_m \in [0,1]$, which offers a fixed, geometry-based encoding of gene expression levels into four discretized, biologically interpretable expression states along the pseudotime trajectory. Further details are provided in Methods. The measurement label for the $j$-th qubit, denoted as $m^{(j)} \in \{0,1,2,3\}$, represents a discretized expression level of gene $\genej$: \( m^{(j)} = 0 \) indicates that gene \( \genej \) is not expressed, while \( m^{(j)} = 3 \) signifies that gene \( \genej \) is highly expressed in a cell.
The overall measurement outcomes vector is denoted as $\bfm = (m^{(1)}, m^{(2)}, \ldots, m^{(n)})$ and spans $4^n$ possible joint outcomes. The probability distribution over the measurement outcomes vector at pseudotime $t$ is given by  \[p(\bfm|t,\bfw) = \<\psi_{t}(\bfw)|\Lambda_{\bfm}|\psi_{t}(\bfw)\>,\] where $\Lambda_{\bfm} = (\Lambda_{m^{(1)}} \otimes \cdots \otimes \Lambda_{m^{(n)}})$.

Collecting repeated measurement outcomes at each pseudotime $t_i\in\{t_1,\cdots,t_{\sfN_t}\}$ yields a discretized gene-expression dataset
\[\calG := [\bfm_{(i,k)}]_{i=1,k=1}^{\sfN_t,\sfN_c} \in \{0,1,2,3\}^{\sfN_t \times \sfN_c \times n},\] which serves as the observable data
for inference. Here, $\sfN_t$ denotes the number of pseudotime bins (discrete time points),
$\sfN_c$ denotes the number of independently measured cells (measurement samples) per bin,
and $n$ is the number of genes (qubits) in the network. This concludes the description of QHGM.

\captionsetup[figure]{font=small, labelfont=bf, width=1\textwidth}
\begin{figure}[]
    \centering
    \includegraphics[width=0.95\linewidth]{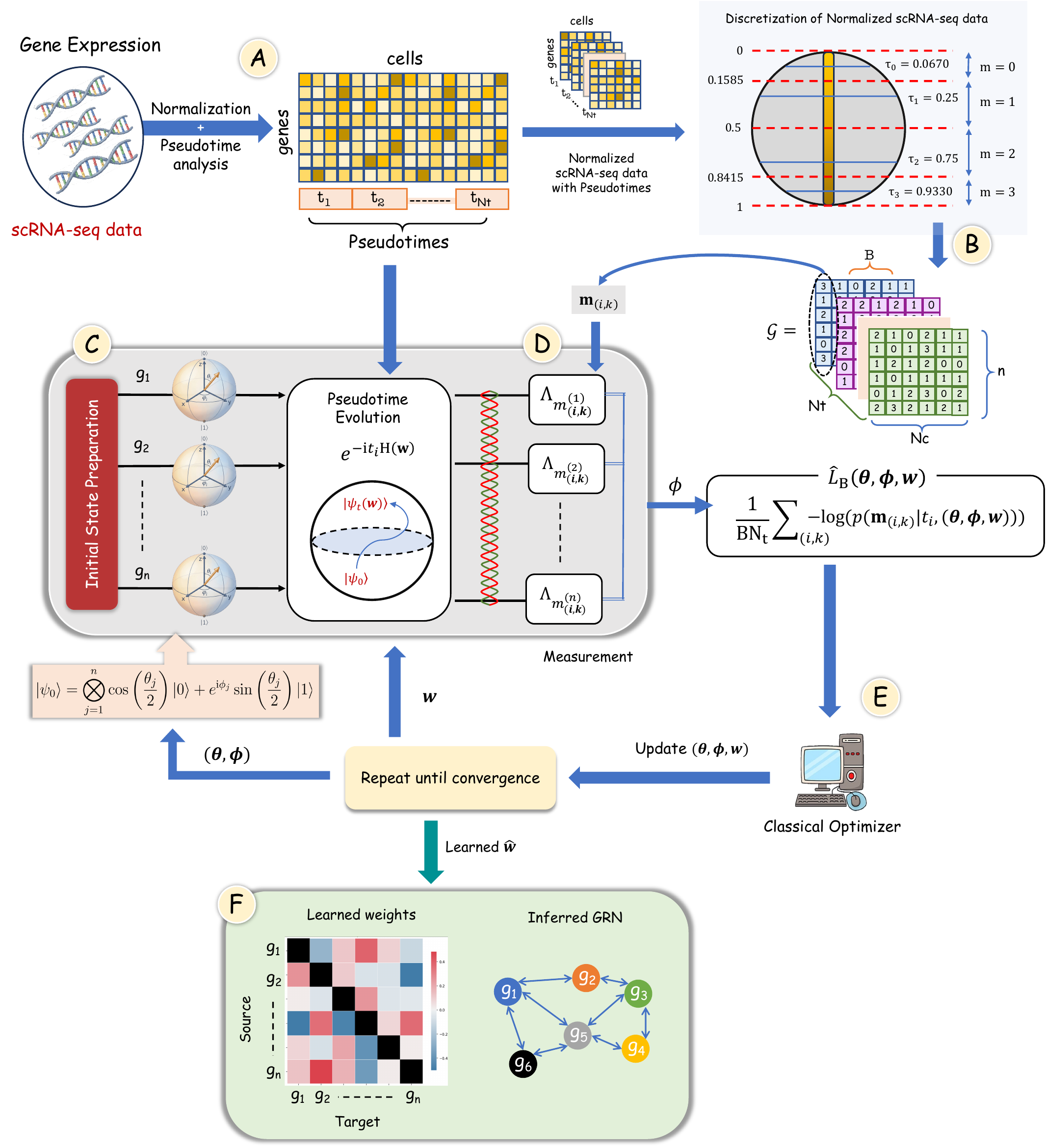}
    \caption{VQ-Net. \textbf{(A)} Raw scRNA-seq data are preprocessed, normalized, and assigned pseudotime values,
providing a temporal ordering of cells along a developmental trajectory.
\textbf{(B)} The normalized pseudotime-ordered scRNA-seq data are converted into four discrete values, denoted as $\calG$, where $\sfN_t$ denotes the number of pseudotime bins,
$\sfN_c$ denotes the number of cells per bin,
and $n$ is the number of genes in the network.
\textbf{(C)} Prepares a separable initial state and evolves
under the parameterized Hamiltonian
$\rmH(\bfw)
\,=\, \sum_{(i,j)\in\calE}
w_{ij}\,\tfrac{1}{2}\bigl(\rmI - \rmZ_i\bigr)\,\otimes\,\rmY_j,$
which encodes directed regulatory interactions.
\textbf{(D)} For each pseudotime bin $t_i$, the single-qubit IC-POVM is applied as measurement
observables on entangled evolved states conditioned on the discretized scRNA-seq data $\bfm_{(i,k)}$ as input. Here, $k$ is the index of the cell in the corresponding pseudotime bin $i$.
\textbf{(E)} The parameters $(\boldsymbol{\theta},\boldsymbol{\phi},\bfw)$ are optimized
by minimizing the mini-batch empirical loss using a classical optimizer.
\textbf{(F)} The learned weights $\widehat{\bfw}$ are visualized as a signed, asymmetric weight
matrix, from which the GRN is inferred.
} 
    \label{fig:algorithm}
\end{figure}

\subsection{Scalable Quantum Algorithm for Network Inference}

In this subsection, we describe the variational quantum network inference (VQ-Net) algorithm for learning the Hamiltonian parameters using discretized pseudotime-ordered scRNA-seq data. The raw scRNA-seq data is first preprocessed using tools such as Scanpy \cite{wolf2018scanpy} or the Seurat (R package) toolkit \cite{hao2024dictionary}. Next, the expression profile of each cell is normalized to the range $[0, 1]$ using methods such as Min-Max normalization. 
The normalized scRNA-seq data matrix is provided as input to the VQ-Net training pipeline, which comprises pseudotime inference, pseudotime binning, scRNA-seq data discretization, measurement probability estimation, and Hamiltonian parameter optimization. Further details are provided below, and summarized in Fig.~\ref{fig:algorithm} and Alg.~\ref{alg:VQNet}.

\begin{algorithm}[!htb]
\caption{VQ-Net}
\label{alg:VQNet}
\DontPrintSemicolon
\setstretch{1.25}
\LinesNumbered

\KwIn{Raw scRNA-seq data $\mathcal{D}$  (genes $\times$ cells),  number of genes $n$, number of pseudotime bins $\sfN_t$, number of cells per bin $\sfN_c$, mini-batch size $\textsf{B}$, max iterations $K$, and weight bound $B$}
\KwOut{inferred Hamiltonian parameters $\hat{\bfw}$ and inferred GRN}

\tcc{Preprocessing and Pseudotime Analysis}
Implement the Scanpy preprocessing pipeline on raw data $\mathcal{D}_{n\times \sfN_t \sfN_c}$\;

Learn latent representation using scVI:
$\mathbf{X}_{\mathrm{PCA}} \leftarrow \mathrm{scVI}(\mathcal{D})$\;

Initialize VIA using the latent embedding, user-defined root cell, and terminal states\;

Infer pseudotime values using VIA:
$\tilde{\mathbf{t}} \leftarrow \mathrm{RunVIA}(\mathbf{X}_{\mathrm{PCA}})$\;

Reorder the cells in $\calD$ in increasing order of $\tilde{\mathbf{t}}$ and group $\sfN_c$ cells sequentially into $\sfN_t$ bins $\{\mathcal{I}_i\}_{i=1}^{\sfN_t}$\;

\For{$i=1$ \KwTo $\sfN_t$}
{
Assign representative pseudotime
${t}_i \leftarrow \mathrm{Median}(\{\tilde{t}_c:c\in\mathcal{I}_i\})$\;
}

\tcc{Discretization}
Normalize $\calD$ to the range $[0,1]$ using min-max normalization:
$\widetilde{\mathcal{D}}\leftarrow \mathrm{MinMaxNormalize}(\mathcal{D})$\;

Define boundaries: $b_0=0, b_1 = 0.1585, b_2=0.5, b_3=0.8415, b_4=1$\;

\ForEach{normalized gene-expression value $g\in \tilde{\calD}$}
{
Assign $m=k$ if $b_k\le g<b_{k+1}$, where $k\in\{0,1,2,3\}$\;
}
Construct
$\calG \leftarrow \{\bfm_{(i,k)}\}\in\{0,1,2,3\}^{\sfN_t\times\sfN_c\times n}$\;

\tcc{VQ-Net training}
Initialize $(\boldsymbol{\theta}^{(0)},\boldsymbol{\phi}^{(0)},\bfw^{(0)})$\;

\For{$\normalfont r=0$ \KwTo $K-1$}
{
Prepare
$\ket{\psi_0(\boldsymbol{\theta}^{(r)},\boldsymbol{\phi}^{(r)})}
\leftarrow
\bigotimes_{j=1}^{n}
\cos({\theta_j^{(r)}}/{2})\ket{0}
+
e^{\mathrm{i}\phi_j^{(r)}}\sin({\theta_j^{(r)}}/{2})\ket{1}$\;

\%\% Use sequential sampling during the first pass over $\calG$, and random sampling thereafter\;

Sample a mini-batch
$\mathsf{B}\leftarrow\{\bfm_{(i,k)}\}_{i=1,k=1}^{\sfN_t,\textsf{B}}\subseteq\calG$\;

    Construct
    $\rmH(\bfw^{(r)})\leftarrow
    \sum_{(a,b)}w_{ab}^{(r)}
    \frac{1}{2}(\rmI-\rmZ_a)\otimes\rmY_b$\; 

    \ForEach{$\bfm_{(i,k)}\in\mathsf{B}$}
    {
        Evolve
        $\ket{\psi_t}\leftarrow
        e^{-\mathrm{i}{t}_i\rmH(\bfw^{(r)})}
        \ket{\psi_0(\boldsymbol{\theta}^{(r)},\boldsymbol{\phi}^{(r)})}$\;

        Compute
        $p(\bfm_{(i,k)}| t_i, (\boldsymbol{\theta}^{(r)},\boldsymbol{\phi}^{(r)},\bfw^{(r)}))
        \leftarrow
        \bra{\psi_t}\Lambda_{\bfm_{(i,k)}}\ket{\psi_t}$\;
    }

    \tcc{Mini-batch empirical loss}
    Compute
    $\widehat{L}_{\textsf{B}}\leftarrow
    -\frac{1}{\textsf{B}\sfN_t}
    \sum_{i=1}^{\sfN_t}\sum_{k=1}^{\textsf{B}}
    \log
    p(\bfm_{(i,k)}| t_i, (\boldsymbol{\theta}^{(r)},\boldsymbol{\phi}^{(r)},\bfw^{(r)}))$\;

    \tcc{Classical parameter optimization}
    Update
    $(\boldsymbol{\theta}^{(r+1)},\boldsymbol{\phi}^{(r+1)},\bfw^{(r+1)})
    \leftarrow
    \texttt{Opt}(\widehat{L}_{\textsf{B}})$
    subject to $\bfw\in\calW_B$\;
}

Set $\hat{\bfw}\leftarrow\bfw^{(K)}$\;

\KwRet{$\hat{\bfw}$ and inferred GRN}
\end{algorithm}

\textit{$1.$ Pseudotime analysis}: Following preprocessing and normalization, a pseudotime value is assigned to
each cell to capture its position along an inferred developmental trajectory. This step can
be performed using established pseudotime inference methods such as VIA \cite{stassen2021generalized},
which outputs a scalar pseudotime value for each cell. The resulting pseudotime $t$ represents the progression of cell-state transitions inferred from the transcriptomic data. In practice, however, pseudotime estimates are affected by noise in scRNA-seq data \cite{brennecke2013accounting,grun2014validation}, leading to small variations between consecutive cell pseudotime values. Such fine-scale variability can introduce numerical instability and increase statistical noise during learning. To obtain more stable pseudotime values, the algorithm groups cells into pseudotime bins. Each bin is then assigned a representative pseudotime equal to the median pseudotime of the cells it contains.

\vspace{5pt}
 \textit{$2.$ Discretization of scRNA-seq data:} 
The algorithm uses the expression score \(\tau_m\) to determine the bins for discretizing the continuous and normalized scRNA-seq data into four discrete levels, corresponding to the measurement outcomes \(m \in \{0, 1, 2, 3\}\). 
The boundaries of the discretization bins are defined by the midpoints between consecutive expression scores.
Let \( b_l = ({\tau_{l-1} + \tau_l})/{2} \) for \( l = 1, 2, 3 \). Additionally, define \( b_0 = 0 \) and \( b_4 = 1 \).
Then, for each normalized expression value \(g \in [0,1]\) assign a discrete expression level \(m\) as follows
\[
m = l \quad \text{if} \quad b_{l} \le g < b_{l+1}, \qquad l\in\{0,1,2,3\}.
\]
The discretization process is summarized in Fig.~\ref{fig:algorithm}B.  Thus, after normalization and discretization of scRNA-seq data, we get the discretized gene expression matrix $\calG = [\bfm_{(i,k)}]_{i=1, k=1}^{\sfN_t, \sfN_c}\in \{0,1,2,3\}^{\sfN_t\times \sfN_c\times n}$ in the same representation format as the data generated by QHGM. 

\vspace{5pt}
\textit{$3.$ Initial state preparation:}
The system is initialized in a product state as given in~\eqref{eqn:initial_state},
using available biological priors to estimate the amplitude angles $\theta_i$ and phase angles $\phi_i$ of individual qubits. 
For example, gene-specific  
$\theta_i$ can be estimated by mapping normalized expression levels at pseudotime $t=0$ to empirical activation frequencies \cite{RomanVicharra2023}. The phases $\phi_i$ can be computed using kinetic information derived from RNA velocity analyses \cite{la2018rna}\cite{lange2022cellrank}. 
In the absence of such prior information, the system can be initialized in a uniform superposition with $\theta_i=\pi/2$ and $\phi_i=0$ for all $i$, providing an unbiased starting configuration. 
More generally, when priors are unavailable or uncertain, the amplitude and phase parameters, $\boldsymbol{\theta}$ and $\boldsymbol{\phi}$, can be treated as trainable variables and jointly learned along with the weights $\bfw$.

\vspace{5pt}
\textit{$4.$ Quantum Hamiltonian evolution and IC-POVM measurements:}
Starting from the separable initial state \eqref{eqn:initial_state}, the system evolves for time $t$ under the parameterized Hamiltonian via the unitary operator $\exp(-\mathrm{i}t\rmH(\bfw))$. Classical simulation of this evolution generally requires manipulating a $2^n \times 2^n$-dimensional Hamiltonian matrix, resulting in exponential computational overhead with the number of qubits. To avoid this bottleneck, the Hamiltonian evolution is implemented using a quantum Hamiltonian simulation circuit, where the unitary evolution is decomposed into efficiently implementable local operations acting on a small number of qubits. Quantum simulation techniques such as Trotterization \cite{childs2021theory} and randomized simulation methods, including qDRIFT \cite{qDRIFT}, enable simulation complexity to scale polynomially with the number of Hamiltonian terms and the evolution time, rather than exponentially with the number of qubits \cite{berry2007efficient,childs2021theory,qDRIFT}. The evolved quantum state is subsequently measured using the fixed local IC-POVM, yielding the model-predicted likelihood \(p(\bfm|t,(\boldsymbol{\theta},\boldsymbol{\phi},\bfw)) = \<\psi_t(\boldsymbol{\theta},\boldsymbol{\phi},\bfw)|\Lambda_{\bfm}|\psi_{t}(\boldsymbol{\theta},\boldsymbol{\phi},\bfw)\>\) over measurement outcome $\bfm$. 

\vspace{5pt}
\textit{$5.$ Optimization:}
The algorithm processes the dataset $\calG$ in mini-batches of size $\sfB$ to improve computational efficiency and introduce stochasticity, which helps avoid suboptimal stationary points during optimization. During the initial pass over the dataset, mini-batches are sampled sequentially to ensure complete coverage across all $\sfN_t$ pseudotime bins. After the full dataset has been scanned, subsequent mini-batches are sampled uniformly at random to improve generalization.
For each pseudotime $t_i$ and cell index $k$ within a mini-batch, the corresponding measurement outcome $\mathbf{m}_{(i,k)}$ is used to evaluate the model-predicted measurement likelihood
$p(\bfm_{(i,k)}| t_i,(\boldsymbol{\theta},\boldsymbol{\phi},\bfw))$.
The empirical loss over a mini-batch is defined as
\begin{align*}
    \widehat{L}_{\sfB}(\boldsymbol{\theta},\boldsymbol{\phi},\bfw)
&:= -\frac{1}{\sfB\sfN_t}
\sum_{i=1}^{\sfN_t}\sum_{k=1}^{\sfB}
\log p(\mathbf{m}_{(i,k)}|t_i,(\boldsymbol{\theta},\boldsymbol{\phi},\bfw)).
\end{align*}
The estimated parameters $(\hat{\boldsymbol{\theta}},\hat{\boldsymbol{\phi}},\hat{\bfw})$ are then obtained by minimizing $\widehat{L}_{\sfB}$ using a classical optimizer, subject to the constraint
$\bfw \in \calW_B := \{\bfw \in \RR^c : \|\bfw\|_2 \leq B\},$ where $c$ is the number of Hamiltonian parameters. The use of the negative log-likelihood objective enables scalable optimization by avoiding the explicit evaluation of exponentially large joint gene-expression distributions, as required by loss functions such as KL divergence, total variation distance, or squared error over probability distributions. Consequently, the per-iteration computational overhead depends only on the mini-batch size and the number of pseudotime bins, yielding polynomial scaling with the number of genes, as discussed in the next subsection.

\subsection{Sample Complexity for Hamiltonian Parameters Recovery using VQ-Net}
In this subsection, we analyze the sample complexity of VQ-Net for reliable recovery of Hamiltonian parameters.
In particular, we seek bounds on $\sfN_t$ and $\sfN_c$ that do not scale exponentially with the number of genes (qubits), while ensuring that the learned weights remain sufficiently close to ground-truth weights. Since our framework is not restricted to qubit systems or separable initial states. For the sample-complexity analysis, we consider quantum systems composed of $n$ qudits of local dimension $\sfd$\footnote{A qudit is an $\sfd$-dimensional quantum system and is a natural generalization of a qubit; the case $\sfd = 2$ corresponds to a qubit.}, and a fixed entangled initial quantum state. Moreover, the sample-complexity guarantees are established with respect to the full empirical loss $\widehat{L}(\bfw):= -\frac{1}{\sfN_t}
\sum_{i=1}^{\sfN_t}\frac{1}{\sfN_c}\sum_{k=1}^{\sfN_c}
\log p(\mathbf{m}_{(i,k)}|t_i,\bfw),$ which is constructed using all $\sfN_c$ measurement samples for all pseudotime bins, whereas the mini-batch objective used during VQ-Net training enables scalable optimization by sampling a batch of size $\sfB$ from the available $\sfN_c$ measurement samples at each optimization step. The precise statistical model is formally introduced in the Methods section (see Definition~\ref{def:model}).

Let $\bfw^*$ be the ground-truth weights. Under the following assumptions, $(i)$ the expected loss, defined as $L(\bfw)
:= \EE_{t}
\EE_{\bfm\sim p(\cdot|t,\bfw^*)}[-\log(p(\bfm|t,\bfw))],$ is $\mu_0$-strongly convex (SC) for some $\mu_0>0$ over $\calW_B$, and $(ii)$ the measurement likelihoods remain sufficiently bounded (see Methods for details), the following theorem establishes sufficient conditions on $\sfN_t$ and $\sfN_c$ under which the learned weights concentrate near the ground-truth weights. Further details and proofs are provided in the Methods section.

\begin{theorem}\label{thm:emp_strongconvexity} For the Hamiltonian learning problem described above,
fix a confidence $\delta>0$ and an empirical SC tolerance $\varepsilon>0$. Under the assumptions stated above, if the number of sampled times $\sfN_t$ and the number of measurement outcomes per time $\sfN_c$ are chosen such that
\begin{align*}
    \sfN_t &= \tilde{\calO}\Big(\frac{c^3}{\varepsilon^2} \log\Big(\frac{1}{\delta}\Big)\Big)
\eqand 
\sfN_c = \tilde{\calO}\Big(\frac{c^3}{\varepsilon^2}\,
\log\Big(\frac{\sfN_t}{\delta}\Big)\Big),
\end{align*}
where $\tilde{O}$ hides logarithmic factors in $c$ and $1/\varepsilon$, as well as fixed model-dependent constants.
Then, with a probability at least $(1-2\delta)$, the empirical loss $\widehat{L}$ is $(1-2\varepsilon)\mu_0$-strongly convex over $\calW_B$. 

Furthermore, suppose $\widehat{\bfw} \in \mathrm{interior} (\calW_B)$ with probability 1. Then, on the event that the empirical loss is $(1-2\varepsilon)\mu_0$-strongly convex, the empirical minimizer satisfies the following bound
\begin{equation}
    \|\widehat{\bfw} - \bfw^*\|_2 \leq {\calO}\bigg(\frac{1}{(1\!-\!2\varepsilon)\mu_0}\sqrt{\frac{c}{\sfN_c} \!\log\Big(\frac{\sfN_t}{\delta}\Big)}\bigg),
\end{equation}
with probability at least $(1-3\delta)$.
\end{theorem}
\noindent\textit{Scaling with system size.}
The above bounds show that the required number of pseudotime bins $\sfN_t$ and the number of cells (measurement samples) per bin $\sfN_c$ scale polynomially with the number of Hamiltonian parameters. Since the number of parameters $c =  \calO(n^2)$ for pairwise gene interactions, both $\sfN_t$ and $\sfN_c$ scale polynomially with the number of genes (qubits). 

\vspace{5pt}
\noindent\textit{Insufficient $\sfN_t$ but $\sfN_c$ is large.} 
Note that for a given sampled time $t_i$, it may happen that different parameter values can induce the same measurement distribution. In particular, for a given time $t$, there may exist $\bfw \neq \bfw^*$ such that $p(\cdot | t, \bfw) = p(\cdot | t, \bfw^*).$
If the number of pseudotime bins $\sfN_t$ is too small, such non-identifiability can persist across the entire set of sampled times, so that distinct parameters remain indistinguishable from the observed data. Increasing $\sfN_c$ only reduces the variance in estimating the per-time expected loss. As a result, increasing $\sfN_c$ alone cannot compensate for insufficient identifiability caused by the limited time samples.

\vspace{5pt}
\noindent\textit{Sufficiently large $\sfN_t$ but $\sfN_c=1$.} If the number of time samples \(\sfN_t\) becomes very large, taking only a single measurement at each time point is insufficient for accurate parameter recovery. When the number of measurements per time is small, the empirical estimate of the measurement distribution at each time is dominated by sampling noise. Consequently, although increasing \(\sfN_t\) provides sufficient identifiability of the parameters, the information available at each time remains too noisy to be useful. 

We next establish a finite-sample uniform convergence guarantee for the empirical loss $\widehat{L}$ around the expected loss $L$ over $\calW_B$. The following theorem shows that as the number of time samples $\sfN_t$ and the number of measurements per time sample $\sfN_c$ increase, $\widehat{L}$ uniformly concentrates around $L$, thus justifying the use of empirical risk minimization as a faithful approximation to the expected risk minimization problem.

\begin{theorem}\label{thm:unif}
Under the assumption $(ii)$ stated above, for any $ \delta >0$, the following non-asymptotic uniform deviation bound holds with probability at least $\big(1-4\delta)$:
\begin{align}
\!\!\sup_{\bfw\in\calW_B}\!\big|\widehat{L}(\bfw)-L(\bfw)\big|
\le\;
&\tilde{\calO}\bigg(\frac{c+\sqrt{\log(1/\delta)}}{\sqrt{\sfN_t}}
+\frac{\sqrt{n\log(d)} + \sqrt{\log(\sfN_t/\delta)}}{\sqrt{\sfN_c}}\bigg),
\label{eq:unif_bound_final}
\end{align}
where $\tilde{O}$ hides fixed model-dependent constants.
\end{theorem}
\noindent An important implication of this uniform convergence result is that it provides a key ingredient for establishing asymptotic strong consistency of the empirical minimizer as $\sfN_t, \sfN_c \to \infty$ (see \cite[Theorem~4.10]{wainwright2019high}).
\subsection{Numerical Experiments and Results}
We demonstrate the usage of QHGM on the GBMap scRNA-seq data to infer GRNs with both known and novel regulatory interactions.
Moreover, we evaluate the proposed VQ-Net through numerical experiments on synthetic data generated with QHGM and benchmark it against classical GRN inference methods.

\textbf{GBMap scRNA-seq data}: We implemented the VQ-Net to build QHGM using the {core GBmap}  scRNA-seq data comprising 16 independent studies and 109 patients, post-standardization and batch normalization \cite{ruiz2025charting}. We focused on Differentiated-like and Stem-like cells in the {annotation level 2} classification, yielding approximately 127{,}000 cells in total. At {annotation level 3}, these cells were further classified as Astrocyte-like (AC-like) (50{,}847), Mesenchymal-like (MES-like) (33{,}167), Neural Progenitor Cell-like (NPC-like) (22{,}117), and Oligodendrocyte Progenitor Cell-like (OPC-like) (21{,}390). 
 Pseudotime is inferred using the VIA algorithm \cite{stassen2021generalized}, with an OPC-like cell as the root representing the putative progenitor state and AC-like and MES-like cells as terminal endpoints (see Fig.~\ref{fig:finalgrn}A). Cells are embedded in UMAP space based on the 5,000 most variable genes, and a streamplot is overlaid on the embedding to visualize the dominant probabilistic flow along pseudotime. Higher pseudotime values correspond to more differentiated states, and the trajectory highlights branching points and intermediate transitional states along the differentiation continuum.  To examine the gene regulatory dynamics along this trajectory, we considered a set of $14$ genes: \textsf{BCAN}, \textsf{STMN1}, \textsf{HES6}, \textsf{ETV1}, \textsf{CADM2}, \textsf{MMP16}, \textsf{CKB}, \textsf{LIMA1}, \textsf{VCAN}, \textsf{JPT1}, \textsf{ASCL1}, \textsf{CDK4}, \textsf{TUBB2B}, and \textsf{NCAM1}. A detailed description of the data preparation and training details is provided in Methods. The learned weights are shown in Fig.~\ref{fig:finalgrn}B, which illustrates the median learned weights across independent training instances. To quantify interaction selectivity, we computed the coefficient of variation (CV), defined as the ratio of the standard deviation to the absolute value of the mean, separately for positive and negative learned weights (Fig.~\ref{fig:finalgrn}C). A high CV indicates selective regulation, in which a gene strongly influences only a subset of the gene set, whereas a low CV reflects broad and more uniform regulatory influence across the gene set. Fig.~\ref{fig:finalgrn}D shows the inferred GRN, consisting of the strongest regulatory interactions, defined by row-wise selection of the top 15th percentile of positive and negative weights.

\captionsetup[figure]{font=small, labelfont=bf, width=1\textwidth}

\begin{figure}[p]
    \centering
    \includegraphics[width=0.95\linewidth]{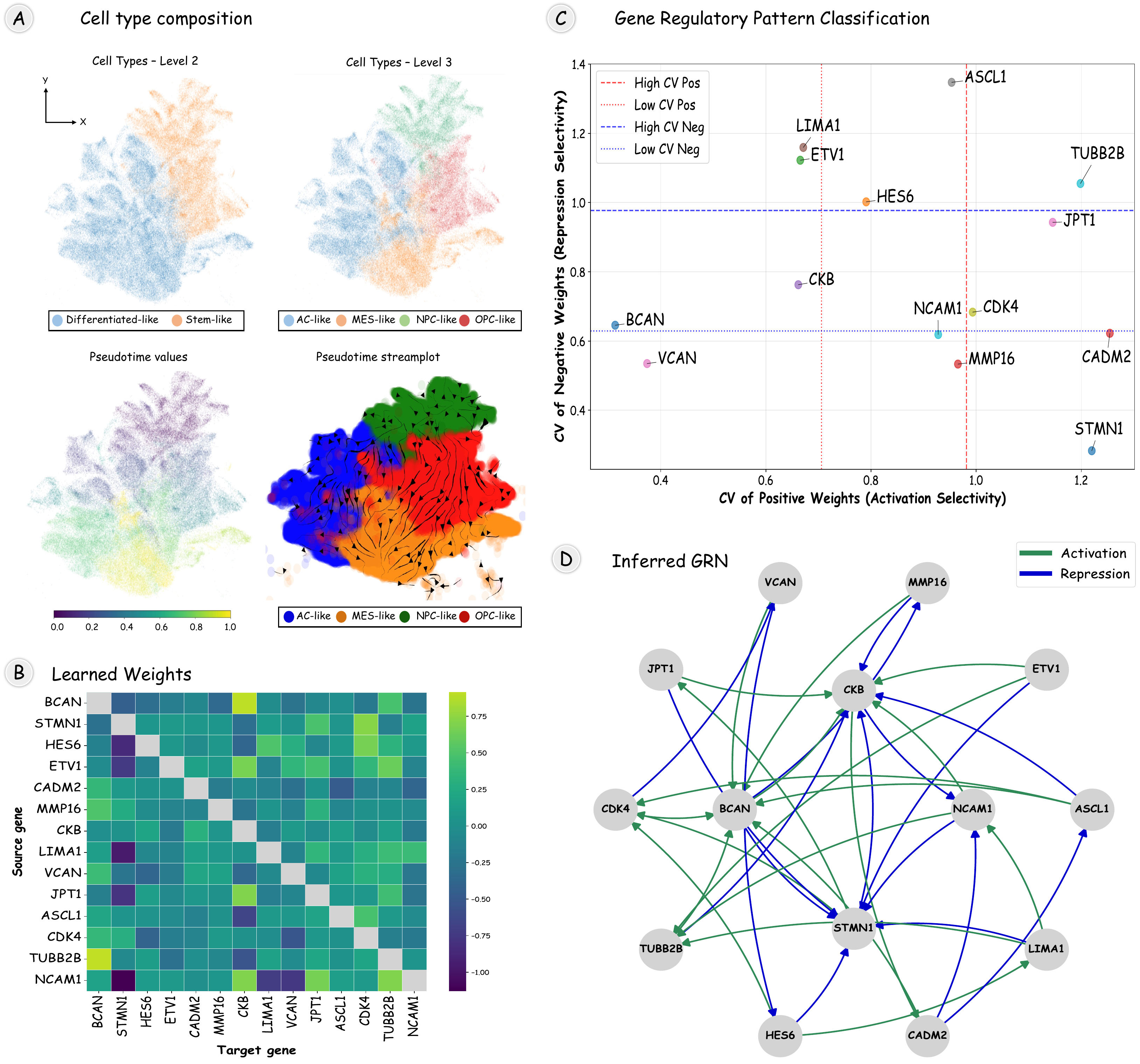}
    \caption{
    Inferred GRN from QHGM on GBMap scRNA-seq data. 
    \textbf{(A)} Cells are shown at increasing annotation granularity (Level 2 and Level 3) at UMAP coordinates. The trajectory inference is set with an OPC-like cell as the root node and progresses toward AC-like and MES-like cell types specified as terminal endpoints. The streamplot overlaid on the UMAP embedding visualizes this progression, showing the dominant probabilistic flow of cells along pseudotime, where higher pseudotime values correspond to more differentiated states.
    \textbf{(B)} Heatmap of median learned weights across 10 simulations, 
    which summarizes central tendencies across simulations.
    \textbf{(C)} Gene-wise classification of regulatory behavior based on the coefficient of variation (CV) of positive and negative weights across simulations. The joint CV analysis distinguishes genes with selective, stable activation or repression from those exhibiting high variability, providing a quantitative measure of regulatory consistency and context dependence.
    \textbf{(D)} GRN is inferred from the learned weights by selecting the top 15th percentile of positive (activation) and negative (repression) Hamiltonian weights separately, visualized as a directed network. Green and blue edges denote activating and repressing interactions, respectively.
    }
    \label{fig:finalgrn}
\end{figure}

To validate the accuracy of these inferred interactions, we compared the weights learned via VQ-Net against known regulatory pairs reported in the literature. As seen in Fig.~\ref{fig:finalgrn}D, QHGM correctly identifies ASCL1’s broad regulatory influence and recovers strong interactions with targets such as \textsf{BCAN}, \textsf{CDK4}, and {\textsf{CKB}}, confirming its established role in driving GBM heterogeneity \cite{myers2024transcription}. Nodes like \textsf{STMN1}, integrate various regulatory inputs that have opposing effects. This is consistent with the known sensitivity of microtubule dynamics and cell cycle regulators to competing upstream signals \cite{bao2017high}. The cyclin-dependent kinase axis, particularly \textsf{CDK4}/\textsf{6}, is well-established as a driver of GBM cell proliferation and cell-cycle progression, with inhibition altering subtype programs in proneural stem-like cells supporting dynamic reprogramming under perturbation \cite{li2017CDK4} \cite{michaud2010pharmacologic}. Proliferation-associated genes (\textsf{CDK4}, \textsf{STMN1}, \textsf{TUBB2B}) are coupled to extracellular matrix (ECM) and cell-adhesion components (\textsf{BCAN}, \textsf{\textsf{VCAN}}, \textsf{NCAM1}, \textsf{MMP16}). This aligns with the literature on ECM proteoglycans (e.g., versican/\textsf{VCAN}), which promotes glioma proliferation and invasion \cite{wei2024vcan}, and adhesion molecules that mediate GBM cell migration in a context-dependent manner \cite{turaga2016adhering}. Beyond the central hub structure, the model captures specific pairwise regulatory dynamics that are well supported by functional biology studies. 
Also, QHGM recovers biologically meaningful feedback regulatory loops. Notably, the model identifies a positive association between \textsf{VCAN} and \textsf{BCAN}. This finding aligns with analyses comparing high versus low \textsf{VCAN} expression groups, where increased VCAN levels were accompanied by upregulation of \textsf{BCAN} \cite{wei2024vcan}. These results suggest that our \textit{quantum-like} model demonstrates high precision in recovering interactions.

Overall, QHGM reveals extensive feedback loops and context-dependent sign switching, particularly involving lineage regulators such as \textsf{ASCL1}, which has been implicated in multiple biological contexts, such as proneural transcriptional programs \cite{parkinson2022proneural} and the neural stem cell-like features of glioma stem cells \cite{cheng2019inhibition}. Such context dependency indicates that regulatory effects are not fixed, instead vary with the global state of the cell, consistent with the heterogeneity observed in GBM stem-like populations and dynamic state transitions \cite{suva2014reconstructing}. Moreover, the order and combination of regulatory events, such as lineage specification versus microenvironmental engagement, can yield distinct downstream outcomes, supporting the notion that OPC-like GBM cells occupy a continuum of regulatory states rather than discrete phenotypes. The inferred GRN reveals that OPC-like GBM states are governed by a densely interconnected regulatory architecture rather than by linear or modular signaling pathways. Collectively, these results suggest that regulatory information in OPC-like cells propagates through interfering, context-dependent, and non-separable pathways, producing dynamic cellular states that are not adequately captured by classical additive or modular gene regulatory models.

\vspace{10pt}
\captionsetup[figure]{font=small, labelfont=bf, width=\textwidth}
\begin{figure}[p]
    \centering
    \includegraphics[width=0.95\textwidth]{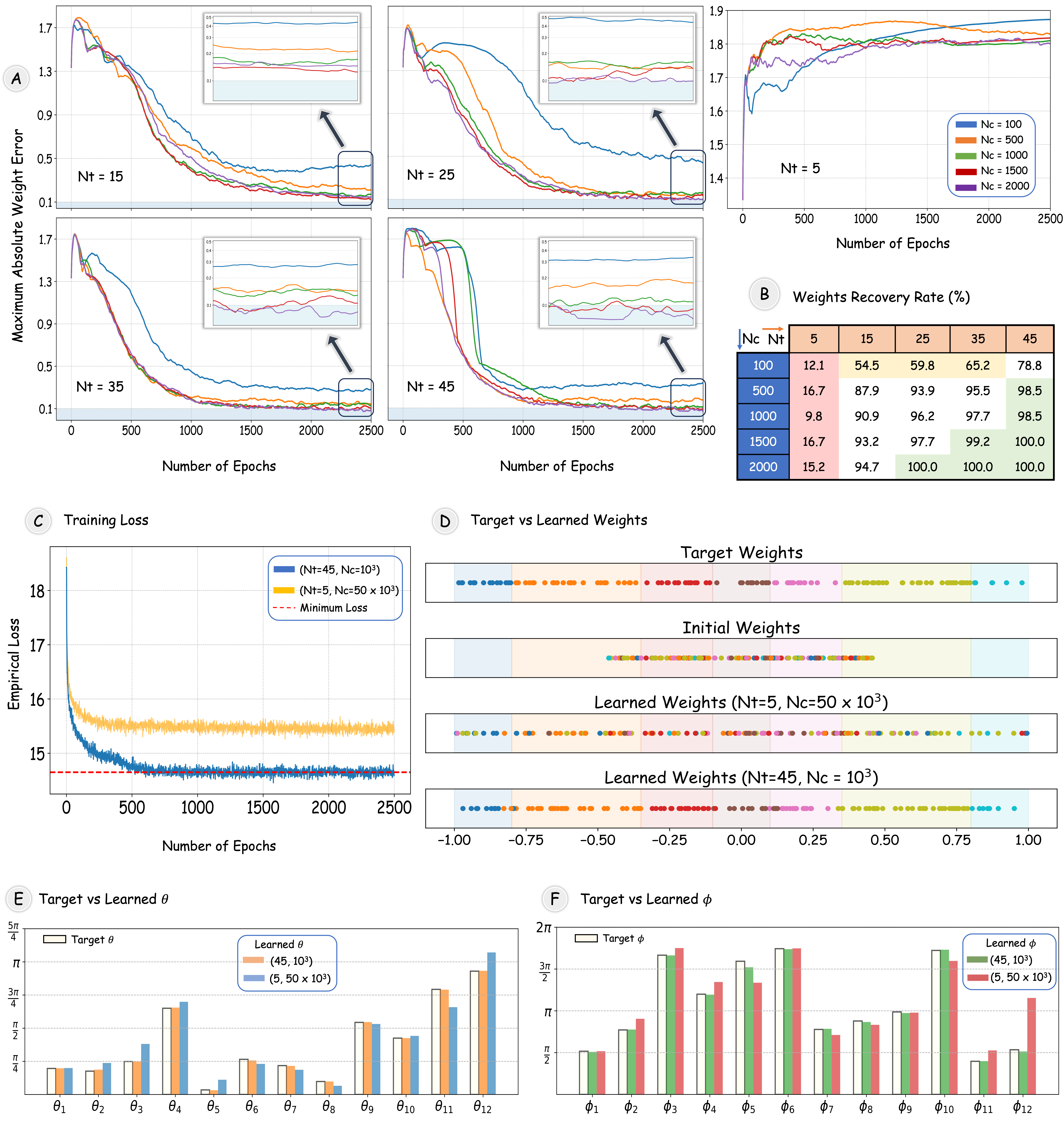}
    \caption{Performance of VQ-Net on the synthetic data generated using QHGM. 
\textbf{(A)} Maximum absolute weight error versus epochs for different numbers of
sampled times $\sfN_t$ and measurements per time  $\sfN_c$.  
\textbf{(B)} Percentage of recovered weights within $10\%$ error, illustrating
the tradeoff between empirical identifiability (controlled by $\sfN_t$) and sampling
variance (controlled by $\sfN_c$).  
\textbf{(C)} Batch empirical loss during training is computed using a mini-batch of size $20$ and $200$ for $\sfN_t=45 \eqand 5$, respectively. The results indicate convergence to
the theoretical optimum for $\sfN_t = 45$.  
\textbf{(D)} Learned weights compared to ground truth, demonstrating strong
recovery for $\sfN_t = 45$ and dispersion for $\sfN_t = 5$.  
\textbf{(E--F)} Recovered initial-state parameters
$(\boldsymbol{\theta},\boldsymbol{\phi})$ with relative errors of $(0.0098, 0.0194) $  for $(\sfN_t,\sfN_c) = (45, 10^3)$ and $(0.1638,0.180)$ for $(\sfN_t,\sfN_c) = (5,50\times 10^3)$, respectively.
}
    \label{fig:synthetica_data_summary}
\end{figure}

 \textbf{Synthetic Data}: To evaluate the performance of {VQ-Net}, we generate synthetic time-resolved measurement data using QHGM as the ground-truth generative model.  We consider a $12$-qubit system with
$132$ randomly chosen weights $\bfw\in [-1,1]^{132}$.  
Fig.~\ref{fig:synthetica_data_summary} summarizes the empirical tradeoff between $\sfN_t$ and $\sfN_c$. 

In Fig.~\ref{fig:synthetica_data_summary}A, we test different $\sfN_t \in \{5,15,25,35,45\}$ and
$\sfN_c \in \{100,500,1000,1500,2000\}$ and report the maximum absolute weight error between learned weights and true weights ($\max_{(i,j)}|w_{ij}-w^{\text{true}}_{ij}| \in [0,2]$) against training epochs. 
For $\sfN_t=5$, the error diverges and eventually saturates near the maximum error of 2 for all values of $\sfN_c$. This suggests potential issues with insufficient empirical identifiability due to fewer time samples. 
In contrast, when $\sfN_t$ exceeds 15, the error decreases monotonically and converges towards zero. Additionally, more stable convergence and reduced error are observed as $\sfN_c$ increases.
In Fig.~\ref{fig:synthetica_data_summary}B, we report the percentage of recovered
weights within a tolerance of $0.1$ (i.e., less than $10\%$ maximum absolute error). For all $\sfN_t $ greater than $ 5$,
the recovery rate improves monotonically as $\sfN_c$ increases. However, when the number of
measurements per time is small ($\sfN_c = 100$), we observe a sharp degradation in
performance, up to a $30\%$ drop in recovery rate (highlighted in yellow), relative to
$\sfN_c > 100$. The sudden decrease can be attributed to high sample variance at low $\sfN_c$, which results in noisy empirical estimates of the true likelihood function, even when there are enough time samples. In this scenario, the optimizer is significantly influenced by sampling noise rather than structural information, causing many weights to exceed the 0.1 tolerance level.
For $\sfN_t=5$, the recovery rate decreases significantly, by as much as \(80\%\) compared to larger values of \(\sfN_t\) (highlighted in red). This decline is likely due to parameter indistinguishability when the number of sampled time points is small. Specifically, multiple distinct weight configurations produce nearly identical measurement statistics across the observed times, resulting in a drop in weight recovery accuracy that does not reliably improve even with increasing \(\sfN_c\). For example, the relative error, defined as ${\|\bfw^{\text{learned}}-\bfw^{\text{true}}\|}/{\bfw^{\text{true}}}$, for \((\sfN_t,\sfN_c) = (5,50\times 10^3)\) is 1.0055, whereas it is 1.251 for \((\sfN_t,\sfN_c) = (5,10^2)\). 

Fig.~\ref{fig:synthetica_data_summary}C shows the batch empirical loss, computed as the
average negative log-likelihood over a mini-batch, versus training epochs. 
In this figure, $(\sfN_t,\sfN_c)=(45,10^3)$
(blue) converges to the theoretical minimum (dashed line), which corresponds to the expected negative log-likelihood loss averaged over the sampled times. In contrast, $(\sfN_t,\sfN_c)=(5, 50\times 10^3)$
(yellow) remains trapped above the optimum. The learned weights (color-coded by target bins) corresponding to these two cases are shown in Fig.~\ref{fig:synthetica_data_summary}D.  For $\sfN_t=45$, the learned weights
cluster tightly around their true values. For $\sfN_t=5$, the learned weights remain dispersed
despite receiving the same initialization. Finally, in Fig.~\ref{fig:synthetica_data_summary}E and F, we inspect the reconstructed initial-state parameters $(\boldsymbol{\theta},\boldsymbol{\phi})$.
For $(45,10^3)$, the relative errors for $\boldsymbol{\theta}$ and $\boldsymbol{\phi}$ are
$0.0098$ and $0.0194$, respectively. Interestingly, even when the weights are not accurately
estimated at $(5,50\times10^3)$, the initial-state parameters are learned with
a lower error of $0.1638$ for $\boldsymbol{\theta}$ and $0.180$ for $\boldsymbol{\phi}$. This indicates that
{VQ-Net} can simultaneously infer both the weights and the initial
state, and the estimation of initial-state parameters is comparatively more robust to
changes in $\sfN_t$ and $\sfN_c$ than the recovery of weights. These observations align closely with our theoretical analysis: the indistinguishability of the parameters is governed by the number of times sampled $\sfN_t$, while
the sampling variance of the empirical loss is controlled by $\sfN_c$. Therefore, simply increasing \( \sfN_t \) or \( \sfN_c \) on its own is not enough to drive the weight estimation error to zero. To achieve consistent recovery, it is necessary to have both a sufficient number of time samples and a sufficient number of measurements at each time point.

\textbf{Benchmarking Classical GRN Inference Methods on QHGM-Generated Data:}
We benchmark state-of-the-art classical inference methods, including ARACNE, GeneNet, GENIE3, and SINCERITIES, on the QHGM-generated data. Recall, QHGM generates discretized gene expression data. However, the classical methods need continuous gene expression data. Therefore, we convert the synthetic QHGM-generated data to continuous values between $0$ and $1$ via Beta distributions with expression scores $\{\tau_m\}_{m=0}^3$ as the mode of the distributions (see Appendix \ref{experimental classical}). We empirically evaluate the performance of these classical inference methods 
using the F1 score and accuracy. 
In Figs.\ref{fig:classical_performance}A and B, we evaluate F1 score and accuracy for network edge recovery and weights sign recovery, respectively, against different total sample sizes ($\sfN_t\cdot \sfN_c$). Network edge recovery is the correct identification of weights up to a threshold, regardless of the sign of the weights. Weights sign recovery is the correct identification of both the presence of an edge and the sign of the weights (up-regulation or down-regulation). 
We keep the ground-truth network's sparsity at roughly 15\% for both edge and weights sign recovery. Error bars indicate variability across different values of $\sfN_t$. For a fixed total sample size, samples are randomly subsampled for each $\sfN_t$, and the resulting error bars for ARACNE, GeneNet, and GENIE3 reflect the corresponding performance variability.

The average accuracy of VQ-Net for both edge recovery and weight sign recovery exceeds $0.95$, and the error bars decrease with increasing sample size. In contrast, for edge recovery, GENIE3 and GeneNet exhibit slightly better performance than ARACNE and SINCERITIES; however, their average performance remains below $25\%$ than that of VQ-NET.  Neither GeneNet nor SINCERITIES achieves an F$_1$ score or accuracy of $0.5$ for signed edge recovery. Moreover, GENIE3 and ARACNE cannot distinguish the sign of regulatory interactions, limiting their ability to infer directed signed GRNs.
Across all classical methods, we observe that increasing sample size does not yield significant improvements in accuracy or F1 score for both edge and sign recovery.

In Fig.~\ref{fig:classical_performance}C, we assess the edge recovery performance of classical methods at various levels of ground-truth network sparsity, defined as the fraction of edges with non-zero weights in the ground-truth network.
Here, VQ-Net again exhibits consistent performance across different levels of sparsity. In contrast, classical methods show a considerable sensitivity to sparsity. For instance, GENIE3 performs well at lower sparsity levels, but its performance declines sharply as sparsity increases. SINCERITIES achieves an average accuracy of only about 50\% across all levels of sparsity. GeneNet and ARCANE performance increase as sparsity increases. However, their F1 score and accuracy are significantly lower than VQ-Net's at lower sparsity levels. 

\begin{figure}[!htbp]
    \centering
\includegraphics[width=\linewidth]{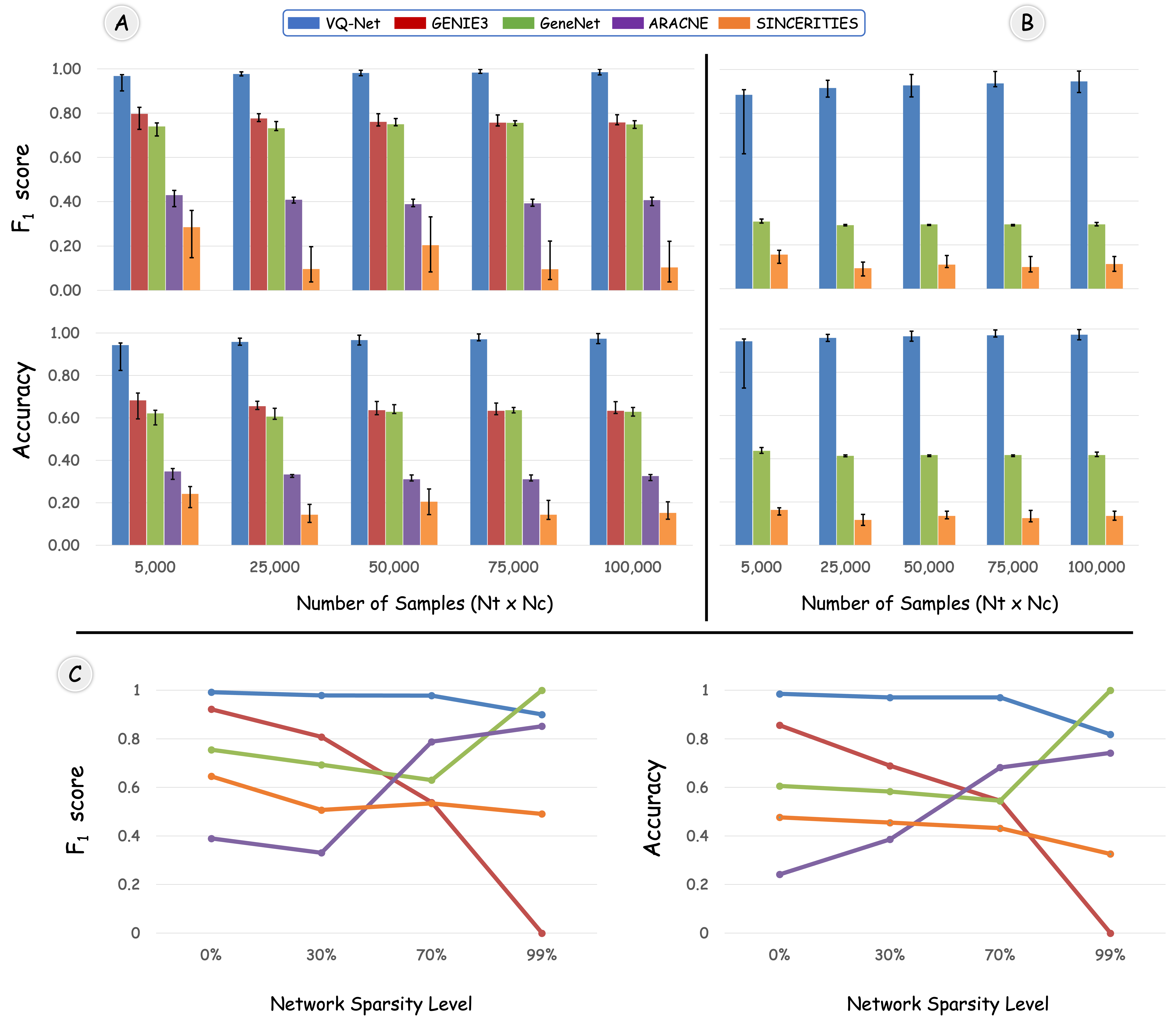}
    \caption{Performance of state-of-the-art classical inference methods on QHGM-generated data. \textbf{(A)} Network Edge recovery
    \textbf{(B)} Weights Sign recovery
    \textbf{(C)} Sparsity-Edge recovery Tradeoff
    }
    \label{fig:classical_performance}
\end{figure}

\section{Discussion}
In this work, we introduced a quantum Hamiltonian learning framework for modeling and inferring gene regulatory networks from scRNA-seq data. We first developed the quantum Hamiltonian-based generative model (QHGM), in which regulatory interactions are encoded via a parameterized Hamiltonian, and gene expression profiles arise from fixed local IC-POVM measurements performed on time-evolved \textit{quantum-like} states. Building on this model, we proposed the VQ-Net algorithm, a scalable variational quantum learning framework that combines Hamiltonian simulation circuits with classical optimization to infer regulatory interactions from pseudotime-resolved scRNA-seq measurements.

We further established finite-sample statistical guarantees for the proposed learning framework. In particular, we analyzed Hamiltonian learning using time-resolved IC-POVM measurements under a fixed-measurement setting and derived non-asymptotic sample-complexity bounds for reliable parameter recovery. Our analysis shows that both the number of sampled pseudotime bins and the number of measurement outcomes per pseudotime bin scale polynomially with the number of genes (qubits). We first established high-probability strong convexity of the empirical loss, which enabled rigorous parameter recovery guarantees for the empirical minimizer. In addition, we derived finite-sample uniform convergence guarantees characterizing the deviation between the empirical and expected loss functions. The resulting bounds separate the statistical contributions arising from pseudotime sampling and finite measurement sampling.

Having established the theoretical guarantees of QHL, we next applied the framework to GRN inference. In contrast to quantum circuit-based approaches \cite{RomanVicharra2023,konar2025alz}, our Hamiltonian formulation is agnostic to gene ordering, and the computational complexity of VQ-Net scales polynomially with the number of genes, making it suitable for large GRNs. It is important to note that QHGM is not meant to serve as a biochemical or biophysical model of gene regulation, nor does it assume that genes function like macroscopic quantum systems. Instead, it uses the mathematical framework of quantum probability theory to construct a richer probabilistic model for understanding complex, context-dependent, and order-sensitive biological interactions that classical probabilistic inference methods may not adequately address.

A practical advantage of the Hamiltonian formulation is its extensibility and expressivity. The QHL framework can naturally incorporate multi-omics data, including chromatin accessibility \cite{buenrostro2015single} and transcription factor binding \cite{kaya2019cut}, by introducing additional Hamiltonian terms or measurement operators. Higher-order regulatory interactions can also be represented by constructing hypergraph-inspired Hamiltonians involving $k$-local interaction terms for $k>2$, thereby extending the framework beyond pairwise interactions \cite{pickard2025scalable}. Moreover, time-dependent Hamiltonians may provide a principled mechanism for modeling dynamically evolving regulatory programs across cellular states and disease progression. The framework may also be extended toward pan-cancer GRN inference by integrating transcriptomic measurements across multiple cancer types within a unified Hamiltonian representation, enabling the identification of shared and cancer-specific regulatory mechanisms. Furthermore, introducing ancillary systems corresponding to external biological or environmental factors may provide a principled route toward modeling open-system effects in regulatory dynamics.

Beyond genomics, the proposed Hamiltonian perspective may also provide useful modeling tools for other complex networked systems. One compelling direction is the study of mind--body interactions using time-series measurements collected from IoT wearable devices, such as heart-rate variability, stress indicators, and behavioral signals, as well as applications to evolving social systems. Such systems often exhibit contextual and interference-like effects that are difficult to characterize using classical Markovian models, suggesting that \textit{quantum-like} representations may reveal richer underlying information structures in human physiology and collective behavior. More broadly, although our experiments focus on GRN inference, the proposed QHL framework is not restricted to genomics and may be extended to quantum many-body learning problems and other \textit{quantum-like} systems, including social, cognitive, and economic networks exhibiting contextual or non-classical probabilistic behavior.
\section{Methods}
In this section, we present the proof of Theorems \ref{thm:emp_strongconvexity} and \ref{thm:unif}, and outline the technical details of QHGM. We first introduce the statistical model underlying the sample-complexity analysis, with the required preliminaries (Lemmas \ref{lem:matrix_sampling}, \ref{lem:khintchine}, and \ref{lem:grad-hess-exp}) summarized in Appendix \ref{app:primer}. We then describe the construction of the parameterized Hamiltonian used to model GRNs (Eq.~\eqref{eqn:grn_hamiltonian}), followed by the construction of the IC-POVM measurements employed in our framework (see Eq.~\eqref{eqn:ic-povm_qgrn}).

\vspace{2pt}
\noindent \textbf{Notations}: Let $\mathcal{H}:= \mathbb{C}^{\sf D}$ denote the $\sf D$-dimensional Hilbert space, which serves as the configuration space for quantum states. For the rest of the paper, $\|\cdot\|$ denotes the Euclidean norm for vectors and the spectral (or Schatten-$\infty$) norm for matrices unless otherwise stated. We use $\|\cdot\|_p$ to denote Schatten $p$-norm for matrices for $p\in [1,\infty)$.

\begin{definition}[IC-POVM \cite{medlock2020informationally}]
An IC-POVM $\Lambda:= \{\Lambda_\bfm \}_{m\in\calM}$ is a finite collection of positive semidefinite operators that sum to the identity, 
and span the space of Hermitian operators $\mathrm{Herm}(\mathcal{H})$, i.e, 
\[
\Lambda_\bfm \geq 0,\ 
\sum_{m\in\calM}\Lambda_\bfm = I,\eqand 
\operatorname{span}\{\Lambda_\bfm\}_{m\in\calM} = \mathrm{Herm}(\mathcal{H}).
\]
\end{definition}
\noindent The minimum number of elements required for informational completeness is \(\sfD^2\).

\begin{definition}[Statistical Model]\label{def:model} Consider a Hamiltonian $\rmH(\bfw) = \sum_{j=1}^c w_j \rmH_j$ that acts on a quantum system of $n$ qudits, each of dimension $\sfd$. It consists of $c$ local terms $\{\rmH_j\}$ acting on a subset of qudits and a parameter vector $\bfw \in \calW_B = \{\bfw \in \RR^c: \|\bfw\|_2 \leq B\}$. For an evolution time $t$, define the corresponding evolved quantum state as follows:
$\rho_t(\bfw) := U_t(\bfw)\,\rho_0\,U_t^\dagger(\bfw),$ 
where $U_t(\bfw) = \exp\{-\imag\,t\,\rmH(\bfw)\}$, and \(\rho_0\) is the initial state. After the evolution, an IC-POVM measurement, denoted as $\Lambda:= \{\Lambda_{m}: m\in \calM\}$, is performed on each qudit, producing an outcome in a finite set $\calM$. This results in a tensored-product measurement acting
on the entire system of dimension $\sfD = \sfd^n$. 
This measurement generates an outcome vector \(\bfm:=(m^{(1)},\cdots, m^{(n)}) \in \calM^{n}\) at time $t$ according to the probability given as
\[p(\bfm|t,\bfw) := \tr(\Lambda_{\bfm} \rho_{t}(\bfw)),\]
where $\Lambda_{\bfm} = (\Lambda_{m^{(1)}} \otimes \cdots \otimes \Lambda_{m^{(n)}})$.

\end{definition}

Consider a set of $\sfN_t$ evolution times, denoted as $\calT:=\{t_1, t_2, \dots, t_{\sfN_t}\}$, independently drawn from a design distribution \(\pi(t)\) over \((0, t_{\max}]\). 
For each $t_i \in \calT$, we are given $\sfN_c$ measurement outcomes, denoted as $\{\bfm_{(i,1)}, \cdots, \bfm_{(i,\sfN_c)}\}$. Here, $\bfm_{(i,k)}$ is a $n$-length vector that denotes the $k$-th outcome collected at time $t_i$. These outcomes are independently and identically generated according to the probability distribution $p(\bfm|t_i,\bfw^*)$ determined by  (unknown) ground-truth weights $\bfw^* \in \calW_B$, i.e., each measurement arises from an independent preparation of 
$\rho_0$, followed by evolution for time 
$t_i$, and performing the IC-POVM $\Lambda$. Note that the outcomes $\{\bfm_{(1,k)},\bfm_{(2,k)}, \cdots,\bfm_{(\sfN_t,k)}\}$ are independent but not identically distrbuted. 

\vspace{5pt}
\textit{Objective:} Given the collection of measurement outcomes across all times, our goal is to obtain an estimate $\widehat{\bfw}:= \widehat{\bfw}\big(\{\bfm_{(i,k)}\}_{i=1,k=1}^{\sfN_t,\sfN_c}\big)
 \in \calW_B$ of parameter $\bfw^*$ such that for any $\delta\in (0,1)$, a bound of the following form holds with probability at least $(1-\delta)$: 
\[\|\widehat{\bfw}-\bfw^*\|_2\leq g(\delta, \sfN_c,\sfN_t).\]
To obtain an estimator $\widehat{\bfw}$ that meets the high-probability error bound stated above, 
we consider the empirical risk minimization framework. Given $\sfN_t$ time samples and $\sfN_c$ independent and identically distributed measurement outcomes for each time $t$, define the empirical and expected loss as follows
\[
\widehat{L}(\bfw)
= \frac{1}{\sfN_t} \sum_{i=1}^{\sfN_t} 
\frac{1}{\sfN_c} \sum_{k=1}^{\sfN_c} 
\ell\!\left(p(\bfm_{(i,k)}|t_i, \bfw)\right)\qquad\eqand
\qquad
L(\bfw)
= \EE_{t\sim\pi}\left[
\EE_{\bfm\sim p(\cdot|t,\bfw^*)}\!\left[\ell\!\left(p(\bfm|t,\bfw)\right)\right]\right],
\]
where $\ell(\cdot):= -\log(\cdot)$.
Then, we obtain the minimizer of empirical loss $\widehat{L}(\bfw)$ as 
$\widehat{\bfw}
:= \arg\min_{\substack{\bfw\in\calW_B}}
\widehat{L}(\bfw).$

\textit{Assumptions:} We state the following assumptions. 
$\textsf{A}(i)$ The expected loss function $L(\bfw)$ is $\mu_0$- strongly convex (SC) over $\calW_B$. $\textsf{A}(ii)$ The likelihood function $\phi$ is bounded,
$p(\bfm|t,\bfw) \;\ge\; p_{\min} > 0$ for all $\bfm\in\calM^n,t \in (0,t_{\max}],$ and $ \bfw\in\calW_B.$
\subsection{Proof of Theorem \ref{thm:emp_strongconvexity}}

\textit{Proof Outline}: 
The proof establishes a finite-sample error bound between $\hat{\bfw}$ and $\bfw^*$ by first showing that the empirical loss inherits the strong convexity of the expected loss with high probability. The proof begins by observing that the complete set of measurement outcomes $\{\bfm_{(i,k)}\}_{i=1,k=1}^{\sfN_t,\sfN_c}$ is jointly independent but not identically distributed. However, for each time index $i$, outcomes are independent and identically distributed. This structure motivates decomposing the difference between the empirical and expected Hessian into two terms, first measuring deviation within the conditional samples $\{\bfm_{(i,k)}\}_{k=1}^{\sfN_c}$ for each time index $i$, and the second term measuring deviation across the time. By applying the matrix Bernstein inequality (Lemma \ref{lem:matrix_sampling} in Appendix \ref{app:primer}) to both terms in this decomposition, we show that the empirical Hessian concentrates near its expectation pointwise for each $\bfw \in \mathcal{W}_B$ with high probability. We then employ a covering argument \cite[Chapter 5]{wainwright2019high} to extend this concentration uniformly over $\calW_B$. By applying Weyl’s inequality \cite[Equation 8.9]{wainwright2019high} and invoking assumption \textsf{A(i)}, we establish the empirical strong convexity of the loss, ensuring a unique empirical minimizer $\hat{\bfw}$ with high probability. However, even $\widehat{L}$ has a unique minimizer, $\hat{\bfw}$ can still be distant from $\bfw^*$. Therefore, under the high-probability event of empirical strong convexity, we bound the distance $\|\hat{\bfw} - \bfw^*\|_2$ by the norm of the gradient of the empirical loss at the true parameter, $\|\nabla \hat{L}(\bfw^*)\|$. Using Lemma \ref{lem:matrix_sampling} again, we bound this gradient norm, which in turn provides a finite-sample bound that vanishes as $\sfN_t, \sfN_c \to \infty$, ensuring that the unique empirical minimizer concentrates around $\mathbf{w}^*$.

Let $\hat{H}(\bfw):=\nabla^2 \widehat{L}(\bfw)$ denote the Hessian of the empirical loss. We begin the proof by introducing the following definitions: $\widehat{H}_{(i,k)}(\bfw):=-\nabla^2_\bfw\log(p(\bfm_{(i,k)}|t_i,\bfw))$ and $ \bar{H}_{t_i}(\bfw) := \EE_{p(\cdot|t_i,\bfw^*)}[-\nabla^2_\bfw \log(p(\bfm|t_i,\bfw))].$
Next, consider the following inequalities:
\begin{align*}
    \|\widehat{H}(\bfw)-\bar{H}(\bfw)\|
&\leq \Big\|\widehat{H}(\bfw) - \frac{1}{\sfN_t}\sum_{i=1}^{\sfN_t} \bar{H}_{t_i}(\bfw)
\Big\| + \Big\|\frac{1}{\sfN_t}\sum_{i=1}^{\sfN_t} \bar{H}_{t_i}(\bfw) - \bar{H}(\bfw)
\Big\|\\
&\leq \frac{1}{\sfN_t}\sum_{i=1}^{\sfN_t} \Big\|\frac{1}{\sfN_c}\sum_{k=1}^{\sfN_c}\widehat{H}_{(i,k)}(\bfw) - \bar{H}_{t_i}(\bfw)
\Big\| + \Big\|\frac{1}{\sfN_t}\sum_{i=1}^{\sfN_t} \bar{H}_{t_i}(\bfw) - \bar{H}(\bfw)
\Big\|\leq \rmT_1 + \rmT_2, 
\end{align*}
where $\rmT_1:=\max_i\|\frac{1}{\sfN_c}\sum_{k=1}^{\sfN_c} \widehat{H}_{(i,k)}(\bfw) - \bar{H}_{t_i}(\bfw)
\|$ and $\rmT_2:= \|\tfrac{1}{\sfN_t}\sum_{i=1}^{\sfN_t} \bar{H}_{t_i}(\bfw) - \bar{H}(\bfw)
\|.$ Note by using Jensen’s inequality and Lemma~\ref{lem:grad-hess-exp} (see Appendix \ref{app:primer}), we establish a bound on $\|\widehat{H}_{(i,k)}(\bfw)\|$ and $\|\bar{H}_{t_i}(\bfw)\|$ for a given time $t_i$. For every $\bfm\in \calM^n, \bfw \in \calW_B$, and $t\in(0,t_{\max}]$, we have 
\begin{equation}\label{eqn:bound_Fbar_ti}
     \|-\nabla^2_\bfw \log(p(\bfm|t,\bfw))\|\;\le\; 2({L_\phi}/{p_{\min}})^2
 \quad \eqand \quad  \|\bar{H}_t(\bfw)\|\;\le\;
2({L_\phi}/{p_{\min}})^2,
\end{equation}
where $L_\phi =  2t_{\max}\sqrt{c}\|\rmH\|_{\infty,\max}.$ The proof is provided in Appendix \ref{app:proof:eqn:bound_Fbar_ti}.
Next, we bound the terms $\rmT_1 \eqand \rmT_2$ using the following proposition. 
\begin{prop}\label{prop:FisherBound}
    For each $\bfw\in\calW_B$, and for all $\varepsilon_1,\varepsilon_2 \geq 0$, the following inequality holds:
    \begin{align}
    \text{(Within-time concentration.)}\quad 
           \Pr\{\rmT_1 \geq \varepsilon_1\} &\leq 2c\;\sfN_t \exp\bigg\{-\frac{\sfN_c \;\varepsilon_1^2/2}{\sigma^2 +  4\varepsilon_1(L_\phi/p_{\min})^2 /3}\bigg\}, \label{eqn:T1_bound}\\
    \text{(Across-time concentration.)}\quad
     \Pr\{\rmT_2\geq \varepsilon_2\} &\leq 2c\exp\bigg\{-\frac{\sfN_t \;\varepsilon_2^2/2}{\sigma^2 +  4\varepsilon_2(L_\phi/p_{\min})^2 /3}\bigg\}
    \label{eqn:T2_bound},
    \end{align}
    where $\sigma^2:= 4(L_\phi/p_{\min})^4$.    
\end{prop}
The proof can be found in Appendix \ref{app:proof:prop:FisherBound}. Finally, combining the bounds on 
$\rmT_1$ \eqref{eqn:T1_bound} and $\rmT_2$ \eqref{eqn:T2_bound}, and letting $\varepsilon_1=\varepsilon_2 = \varepsilon \mu_0/2$, we conclude that for each $\bfw\in\calW_B$, and for all $\varepsilon>0$,
\begin{align}
      \Pr\{ \|\widehat{H}(\bfw)-&\bar{H}(\bfw)\| \geq \varepsilon \mu_0\} \leq   \Pr\{\rmT_1 \geq \varepsilon \mu_0/2\} +  \Pr\{\rmT_2 \geq \varepsilon \mu_0/2\}\nonumber\\
      &\leq 2c\bigg(\sfN_t\exp\bigg\{\frac{-\sfN_c \;\mu_0^2 \;\varepsilon^2/8}{\sigma^2\! +  2\varepsilon\mu_0(L_\phi/p_{\min})^2 /3}\bigg\} + \exp\bigg\{\frac{-\sfN_t\; \mu_0^2 \;\varepsilon^2/8}{\sigma^2 +  2\varepsilon\mu_0(L_\phi/p_{\min})^2 /3}\bigg\}\bigg).
    \label{eqn:finalbound}
\end{align}
We now extend the pointwise bound to hold uniformly over $\calW_B$. To this end, we first establish the Lipschitz continuity of 
$\widehat{H}$ and 
$\bar{H}$. The proof is provided in Appendix \ref{app:proof:prop:FisherLipschitz}. 
\begin{prop} \label{prop:FisherLipschitz} For every $\bfm \in \calM^n$ and $t \in (0, t_{\max}]$, both the expected Hessian and the empirical Hessian are Lipschitz continuous.
\[\|\widehat{H}(\bfw)-\widehat{H}(\bfw')\| \leq 7(L_\phi/p_{\min})^3 \|\bfw-\bfw'\| \quad \eqand \quad \|\bar{H}(\bfw)-\bar{H}(\bfw')\| \leq 7(L_\phi/p_{\min})^3 \|\bfw-\bfw'\|,\]
for each $\bfw,\bfw'\in\calW_B$.
\end{prop}
Fix $\eta > 0$, to be specified later. To control the supremum over $\calW_B$, we cover $\calW_B$ with an $\eta$-net. 
That is, there exists a finite set $\calN_\eta \subset \calW_B$ such that for every 
$\bfw \in \calW_B$, there exists $\bfw' \in \calN_\eta$ with 
$\|\bfw - \bfw'\|_2 \le \eta$. The cardinality of such a net can be bounded as
$ N_\eta := |\calN_\eta|\;\le\; (\frac{3B}{\eta})^c,$
where the inequality follows from the standard bound on the covering number of a $c$-dimensional Euclidean ball \cite[Chapter 5]{wainwright2019high}.
 Let $\calN_\eta:= \{\bfw_1, \bfw_2, \cdots, \bfw_{N_\eta}\} $ be an $\eta-$net points of $\calW_B$. 
Now, fix any $\bfw\in\calW_B$, and let $\bfw_k$ be its closest $\eta-$net point. Then, using Proposition \ref{prop:FisherLipschitz}, we obtain
\begin{align}
    \|\widehat{H}(\bfw)-\bar{H}(\bfw)\| &\leq \|\widehat{H}(\bfw)-\widehat{H}(\bfw_k)\| + \|\widehat{H}(\bfw_k)-\bar{H}(\bfw_k)\| + \|\bar{H}(\bfw_k)-\bar{H}(\bfw)\|\nonumber\\
    &\leq\|\widehat{H}({\bfw_k})-\bar{H}({\bfw_k})\|
+ 14(L_\phi/p_{\min})^3 \|\bfw-\bfw_k\|.\nonumber
\end{align}
Thus, $ \sup_{\bfw\in\calW_B} \|\widehat{H}(\bfw)-\bar{H}(\bfw)\| \leq  \max_{k \in \{1,\cdots,N_\eta\}} \|\widehat{H}({\bfw_k})-\bar{H}({\bfw_k})\| + 14\eta(L_\phi/p_{\min})^3.$ Next, using \eqref{eqn:finalbound} and union bound over the finite $\eta$-net, we obtain, for all $\varepsilon>0$
\begin{align}
    \Pr\{\max&\;\!_{k\in \{1,\cdots,N_\eta\}} \;\|\widehat{H}(\bfw_k)-\bar{H}(\bfw_k)\| \geq \varepsilon\mu_0\} \nonumber\\
    &=  \Pr\Big\{\bigcup\;\!\!_{k\in \{1,\cdots,N_\eta\}} \; \{\|\widehat{H}(\bfw_k)-\bar{H}(\bfw_k)\| \geq \varepsilon\mu_0\}\Big\} \nonumber\\
    &\leq 2c\;N_{\eta}\bigg(\sfN_t\exp\bigg\{\frac{-\sfN_c \;\mu_0^2 \;\varepsilon^2/8}{\sigma^2\! +  2\varepsilon\mu_0(L_\phi/p_{\min})^2 /3}\bigg\} + \exp\bigg\{\frac{-\sfN_t\; \mu_0^2 \;\varepsilon^2/8}{\sigma^2 +  \varepsilon\mu_0(L_\phi/p_{\min})^2 /3}\bigg\}\bigg).
\end{align}
Therefore, after setting $\eta = \varepsilon\mu_0 / \big(14(L_\phi/p_{\min})^3\big)$, we get
\begin{align}
    \Pr\Big\{&\sup_{\bfw\in\calW_B}  \|\widehat{H}(\bfw)-\bar{H}(\bfw)\| \geq 2\varepsilon\mu_0\Big\} \nonumber\\
    &\leq  \Pr\{\max\;\!\!_{k\in \{1,\cdots,N_\eta\}} \;\|\widehat{H}(\bfw_k)-\bar{H}(\bfw_k)\| \geq \varepsilon\mu_0\}\nonumber\\
    &\leq 2c\bigg(\frac{42 B\;L_\phi^3}{\varepsilon\mu_0\;p_{\min}^3}\bigg)^c \bigg(\sfN_t\exp\bigg\{\frac{-\sfN_c \;\mu_0^2 \;\varepsilon^2/8}{\sigma^2\! +  2\varepsilon\mu_0(L_\phi/p_{\min})^2 /3} \bigg\} + \exp\bigg\{\frac{-\sfN_t\; \mu_0^2 \;\varepsilon^2/8}{\sigma^2 +  \varepsilon\mu_0(L_\phi/p_{\min})^2 /3}\bigg\}\bigg)\nonumber.
\end{align}
Finally, for all $\varepsilon>0$ and $\delta \in (0,1/2)$, if we choose 
\begin{align*}
    \sfN_t &= {\calO}\!\bigg(\frac{L_\phi^4}{\mu_0^2\; p_{\min}^4\,\varepsilon^2}\,
\bigg(c\log\!\bigg(\frac{BL_\phi^3}{\mu_0\;p_{\min}^3\,\varepsilon}\bigg) + \log\!\left(\frac{c}{\delta}\right)\!\!\bigg)\\
\eqand 
\sfN_c& = {\calO}\!\bigg(\frac{L_\phi^4}{\mu_0^2\; p_{\min}^4\,\varepsilon^2}\,
\bigg(c\log\!\bigg(\frac{BL_\phi^3}{\mu_0\;p_{\min}^3\,\varepsilon}\bigg) + \log\!\left(\frac{c\sfN_t}{\delta}\right)\!\!\bigg),
\end{align*}
then, with probability at least $(1-2\delta)$,
$\sup_{\bfw\in\calW_B}\|\widehat{H}(\bfw)-\bar{H}(\bfw)\|\le 2\varepsilon\mu_0.$ This establishes uniform convergence of the empirical Hessian
over $\calW_B$. Next, we show the empirical strong convexity using Weyl’s inequality \cite[Equation 8.9]{wainwright2019high}. For each
$\bfw\in\calW_B$,
\[
\big|\lambda_{\min}(\widehat{H}(\bfw))-\lambda_{\min}(\bar{H}(\bfw))\big|
\le \|\widehat{H}(\bfw)-\bar{H}(\bfw)\|.
\]
Taking the supremum over $\bfw\in\calW_B$ yields,
\[
\sup_{\bfw\in\calW_B}
\big|\lambda_{\min}(\widehat{H}(\bfw))-\lambda_{\min}(\bar{H}(\bfw))\big|
\le
\sup_{\bfw\in\calW_B}
\|\widehat{H}(\bfw)-\bar{H}(\bfw)\|.
\]
We now use the inequality that for bounded functions
$g_1,g_2:\calX\to\RR$,
$|\inf_{\calX} g_1-\inf_{\calX} g_2|\le \sup_{\calX}|g_1-g_2|$.
Applying this with
$g_1(\bfw)\!=\!\lambda_{\min}(\widehat{H}(\bfw))$ and
$g_2(\bfw)\!=\!\lambda_{\min}(\bar{H}(\bfw))$, we obtain
\[
\Big|
\inf_{\bfw\in\calW_B}\lambda_{\min}(\widehat{H}(\bfw))
-
\inf_{\bfw\in\calW_B}\lambda_{\min}(\bar{H}(\bfw))
\Big|
\le
\sup_{\bfw\in\calW_B}
\|\widehat{H}(\bfw)-\bar{H}(\bfw)\|.
\]
Recalling the strong convexity assumption, $\inf_{\bfw\in\calW_B}\lambda_{\min}(\bar{H}(\bfw))\ge \mu_0.$
Therefore,
$\inf_{\bfw\in\calW_B}\lambda_{\min}\big(\widehat{H}(\bfw)\big)
\ge (1-2\varepsilon)\mu_0$
with probability at least $(1-2\delta)$. 
This completes the first part of Theorem \ref{thm:emp_strongconvexity}.

Next, we derive the finite-sample bound on the parameter error. Define the event
\[\mathcal{E}_{\rm sc}
:=\{
\inf_{\bfw\in\calW_B}
\lambda_{\min}\big(\widehat{H}(\bfw)\big)
\ge (1-2\varepsilon)\mu_0\}.\] On the event $\mathcal{E}_{\mathrm{esc}}$, the empirical loss $\widehat{L}$ is
$(1-2\varepsilon)\mu_0$-strongly convex over $\calW_B$.
Therefore, from the  equivalent conditions of strong convexity (see Appendix \ref{app:strong_convexity}), for all $\bfw\in\calW_B$,
\begin{equation}
    (1-2\varepsilon)\mu_0\,\|\bfw^*-\bfw\|^2 \leq \big(\nabla_\bfw \widehat{L}(\bfw^*)-\nabla_\bfw \widehat{L}(\bfw)\big)^\intercal(\bfw^*-\bfw)  \leq \big\|\nabla_\bfw \widehat{L}(\bfw^*)-\nabla_\bfw \widehat{L}(\bfw)\big\|\; \|\bfw^*-\bfw\|, \label{eqn:strong_convex1}
\end{equation}
where the second inequality follows from the Cauchy--Schwarz inequality.
Since $\widehat{\bfw} \in \text{interior}(\calW_B)$ with probability 1, the inequality \eqref{eqn:strong_convex1} holds at $\bfw=\widehat{\bfw}$, yielding
\[(1-2\varepsilon)\mu_0\,\|\bfw^*-\widehat{\bfw}\| \leq \big\|\nabla_\bfw \widehat{L}(\bfw^*)-\nabla_\bfw \widehat{L}(\widehat{\bfw})\big\| = \big\|\nabla_\bfw \widehat{L}(\bfw^*)\big\|\]
with probability at least $(1-2\delta)$, where the equality follows from the first-order optimality condition $\nabla_\bfw \widehat{L}(\widehat{\bfw}) = 0$.
Finally, we need to bound the norm of the gradient $\|\nabla_\bfw \widehat{L}(\bfw^*)\|$. To do this, we use Lemma \ref{lem:matrix_sampling} along with the union bound. Consider the following inequalities:
\[\|\nabla_\bfw \widehat{L}(\bfw^*)\|\leq \frac{1}{\sfN_t}\sum_{i=1}^{\sfN_t} \Big\|\frac{1}{\sfN_c}\sum_{k=1}^{\sfN_c}\nabla_\bfw \ell(p(\bfm_{(i,k)}|t_i,\bfw^*))\Big\|
\leq \max_{1\leq i \leq \sfN_t} \Big\|\frac{1}{\sfN_c}\sum_{k=1}^{\sfN_c}\nabla_\bfw \ell(p(\bfm_{(i,k)}|t_i,\bfw^*))\Big\|.\]
Observe that, for a given time $t_i$, the random vectors $\nabla_\bfw \ell(p(\bfm_{(i,k)}|t_i,\bfw^*))$ are independent satisfying $\EE_{p(\cdot|t_i,\bfw^*)}[\nabla_\bfw \ell(p(\bfm_{(i,k)}|t_i,\bfw^*))] = 0$, and by Lemma \ref{lem:grad-hess-exp}, we have $\|\nabla_\bfw \ell(p(\bfm_{(i,k)}|t_i,\bfw^*))\| \leq (L_\phi/p_{\min})$. Therefore, for any $\delta'\in(0,1)$, the following inequality holds
\[\|\nabla_\bfw \widehat{L}(\bfw^*)\| \leq \sqrt{\frac{2\tau^2_{\max}}{\sfN_c} \log\left(\frac{(1+c)\sfN_t}{\delta'}\right)} + \frac{2L_\phi}{3p_{\min} \sfN_c}\log\left(\frac{(1+c)\sfN_t}{\delta'}\right),\]
with probability at least $(1-\delta')$, where $\tau^2_{\max}:= \max_i \tau^2_i$ and \[\tau^2_i:= \|\EE_{p(\cdot|t_i,\bfw^*)}[\nabla_\bfw \ell(p(\bfm_{(i,k)}|t_i,\bfw^*)) \nabla_\bfw \ell(p(\bfm_{(i,k)}|t_i,\bfw^*)) ^\intercal]\| \leq (L_\phi/p_{\min})^2\] for all $1\leq i\leq \sfN_t.$ Let $\delta' = \delta$. This gives us the desired error bound. For any $\delta >0$, the following inequality holds:
 \begin{equation}
    \Pr\bigg\{\calE_{\rm sc} \cap \|\widehat{\bfw} - \bfw^*\| \leq \frac{(L_\phi/p_{\min})}{(1-2\varepsilon)\mu_0}\bigg(\sqrt{\frac{2}{\sfN_c} \log\bigg(\frac{(1+c)\sfN_t}{\delta}\bigg)} + \frac{1}{\sfN_c}\log\bigg(\frac{(1+c)\sfN_t}{\delta}\bigg)\bigg)\bigg\} \geq (1-3\delta),
 \end{equation}
This completes the proof of Theorem \ref{thm:emp_strongconvexity}.

\subsection{Proof of Theorem \ref{thm:unif}}
\textit{Proof Outline:} Similar to proof of Theorem \ref{thm:emp_strongconvexity}, we first decompose $\sup_{\bfw \in\calW_B}|\widehat{L}(\bfw)-L(\bfw)|$ into within time and across time components. To control each part, we use the notion of empirical \emph{Rademacher complexity} \cite{shalev2014understanding,wainwright2019high}. Combining these components yields a uniform convergence bound that vanishes in probability as $\sfN_t,\sfN_c\to\infty$.

To formalize the proof outline above, we define $L_{t}(\bfw) := \EE_{\bfm \sim p(\cdot | t, \bfw^\ast)}[\ell(p(\bfm | t, \bfw))]$. Note by using Jensen’s inequality and Lemma~\ref{lem:grad-hess-exp}, the function $L_t(\bfw)$ is Lipschitz continuous on $\calW_B$ with constant $L_\phi/p_{\min}$: for any $t\in(0,t_{\max}]$ and any $\bfw,\bfw'\in\calW_B$,
\begin{equation}\label{eqn:lipchitz_Lt}
    | L_t(\bfw)-L_t(\bfw') | \le \frac{L_\phi}{p_{\min}} \, \| \bfw-\bfw'\| .
\end{equation}
Consider the following inequalities:
\begin{align*}
\big| \widehat{L}(\bfw)-L(\bfw)\big|
&=
\Big|
\frac{1}{\sfN_t\sfN_c}\sum_{i=1}^{\sfN_t}\sum_{k=1}^{\sfN_c}
\ell(p(\bfm_{(i,k)}|t_i,\bfw)) 
-L(\bfw)
\Big|\\
&\le
{\Big|\frac{1}{\sfN_t\sfN_c}\sum_{i=1}^{\sfN_t}\sum_{k=1}^{\sfN_c}
\ell(p(\bfm_{(i,k)}|t_i,\bfw)) 
-\frac{1}{\sfN_t}\sum_{i=1}^{\sfN_t} L_{t_i}(\bfw)\Big|} + {\rmT_2(\bfw)}\le \frac{1}{\sfN_t}\sum_{i=1}^{\sfN_t} \rmT_1^{(i)}(\bfw) + {\rmT_2(\bfw)}
\end{align*}
where \[{\rmT_2(\bfw)} := \Big|\frac{1}{\sfN_t}\sum_{i=1}^{\sfN_t} L_{t_i}(\bfw)-L(\bfw)\Big| \eqand \rmT_1^{(i)}(\bfw):= {\Big|\frac{1}{\sfN_c}\sum_{k=1}^{\sfN_c}
\ell(p(\bfm_{(i,k)}|t_i,\bfw)) 
-L_{t_i}(\bfw)\Big|}.\]
Hence, \[\sup_{\bfw\in\calW_B}\big| \widehat{L}(\bfw)-L(\bfw)\big|
\le
\frac{1}{\sfN_t}\sum_{i=1}^{\sfN_t}\sup_{\bfw\in\calW_B} \rmT_1^{(i)}(\bfw)
+
\sup_{\bfw\in\calW_B} \rmT_2(\bfw)
\le \max_{1\leq i \leq \sfN_t} \sup_{\bfw\in\calW_B}  \rmT_1^{(i)}(\bfw) + \sup_{\bfw\in\calW_B} {\rmT_2(\bfw)}.\]
Next, define $\calS_i:=\{\bfm_{(i,1)},\cdots,\bfm_{(i,\sfN_c)}\}$ be the set of measurement outcomes at a time $t_i$ and $\calT:= \{t_1,\cdots,t_{\sfN_t}\}$ be the set of times. We now bound the above terms using McDiarmid's inequality \cite[Corollary 2.21]{wainwright2019high} \cite[Lemma 26.4]{shalev2014understanding}.  In particular, for any $\delta_1,\delta_2 > 0$, the following bounds hold:
\begin{align}
    \max_{1\le i\le \sfN_t}\;
\sup_{\bfw\in\calW_B}
\rmT_1^{(i)}(\bfw)
&\le
\max_i 2\;\hat{\mathfrak R}_{S_i}(\calL_{\bfm}^{(i)})
+3(-\log p_{\min})\sqrt{\frac{\log({2\sfN_t/\delta_1})}{2\sfN_c}}\label{eqn:unif_1}\\
\eqand \sup_{\bfw\in\calW_B}\rmT_2(\bfw)
&\le
2\,\hat{\mathfrak R}_\calT(\calL_{t})
\;+\;
3(-\log p_{\min})\sqrt{\frac{2\log(2/\delta_2)}{\sfN_t}},\label{eqn:unif_2}
\end{align}
with probabilities of at least $(1\!-\!2\delta_1)$ and $(1\!-\!2\delta_2)$, respectively.
The derivation of equations \eqref{eqn:unif_1} and \eqref{eqn:unif_2} are provided in Appendices \ref{app:proof:eqn:unif_1} and \ref{app:proof:eqn:unif_2}, respectively.
Here, $\calL_{\bfm}^{(i)}:=\{\bfm\mapsto \ell(p(\bfm|t_i,\bfw)):\;\bfw\in\calW_B\}$ is the loss class defined over outcomes $\bfm$ at time $t_i$, and $\calL_{t}:=\{t\mapsto \EE_{p(\cdot|t,\bfw^*)}[\ell(p(\bfm|t,\bfw))]: \;\bfw\in\calW_B\}$ is the loss class defined over time $t$. The empirical Rademacher complexities are defined as:
\begin{align*}
    \hat{\mathfrak R}_{S_i}(\calL_{\bfm}^{(i)})
&:=
\EE_{\boldsymbol\sigma}\left[
\sup_{\bfw\in\calW_B}\frac{1}{\sfN_c}\sum_{k=1}^{\sfN_c}
\sigma_{(i,k)}\,
\ell(p(\bfm_{(i,k)}|t_i,\bfw)) 
\right]\;\eqand \;
\hat{\mathfrak R}_\calT(\calL_{t}):=
\EE_{\boldsymbol\sigma}\left[
\sup_{\bfw\in\calW_B}\frac{1}{\sfN_t}\sum_{i=1}^{\sfN_t}\sigma_i\,L_{t_i}(\bfw)
\right],
\end{align*}
where $\{\sigma_{(i,k)}\}$ and $\{\sigma_i\}$ are independent Rademacher random variables (i.e., uniformly distributed on $\{\pm 1\}$).  
It remains to derive bounds on the Rademacher complexities of the loss classes.
We bound $\hat{\mathfrak R}_{S_i}(\calL_{\bfm}^{(i)})$ using the contraction lemma \cite[Lemma 26.9]{shalev2014understanding}  together with Lemmas~\ref{lem:khintchine} and \ref{lem:grad-hess-exp} (see Appendix \ref{app:primer}), as stated in Proposition~\ref{prop:rademacher_1}. Additionally, we bound $\hat{\mathfrak R}_{\calT}(\calL_{t})$ using Dudley’s theorem \cite{dudley1967sizes} \cite[Example 5.24]{wainwright2019high}and Eqn.~\ref{eqn:lipchitz_Lt}, as stated in Proposition~\ref{prop:rademacher_2}.

\begin{prop}
\label{prop:rademacher_1}
Given a set of sampled times $\calT$, the empirical Rademacher complexity of the loss class $\calL_{\bfm}^{(i)}$ satisfies the following bound:
\[\max_{i}\hat{\mathfrak R}_{\calS_i}(\calL_{\bfm}^{(i)})\le
\frac{\sqrt{2\log(2\sfD)}}{p_{\min}\sqrt{\sfN_c}},\]
for each $t_i \in \calT$, as well as for any set of measurement outcomes $\calS_i$.
\end{prop}
\begin{proof}
    The proof is provided in Appendix \ref{app:proof:prop:rademacher_1}.
\end{proof}
\begin{prop}
\label{prop:rademacher_2}
The empirical Rademacher complexity of the loss class $\calL_{t}$ satisfies the following bound:
\[\hat{\mathfrak R}_\calT(\calL_{t})\le
\frac{36\; L_\phi B \sqrt{\pi c}}{p_{\min}\sqrt{\sfN_c}},\]
for any set of sampled times $\calT$.
\end{prop}
\begin{proof}
    The proof is provided in Appendix \ref{app:proof:prop:rademacher_2}.
\end{proof}

Finally, by applying Propositions \ref{prop:rademacher_1} and \ref{prop:rademacher_2}, along with equations \eqref{eqn:unif_1} and \eqref{eqn:unif_2}, and setting $\delta_1 = \delta_2 = \delta$, we derive the following: for any $\delta \in (0,1/4)$, with probability  at least $(1-4\delta)$, we conclude that
\begin{equation*}\label{eqn:uc}
\sup_{\bfw\in\calW_B}\big|\hat{L}(\bfw)-L(\bfw)\big|
\le\;\frac{72 L_\phi B \sqrt{\pi c}}{p_{\min}\sqrt{\sfN_t}}
+\frac{2\sqrt{2\log(2\sfD)}}{p_{\min}\sqrt{\sfN_c}}
+3(-\log p_{\min})\bigg(\sqrt{\frac{2\log(2/\delta)}{\sfN_t}}+\sqrt{\frac{\log(2\sfN_t/\delta)}{2\sfN_c}}\bigg).
\end{equation*}
\subsection{Construction of the parameterized Hamiltonian for GRNs (Eqn.~\ref{eqn:grn_hamiltonian})}

We first write the Hamiltonian as a sum of pairwise interaction terms,
$\rmH(\bfw) = \sum_{(i,j)\in\calE} w_{ij}\,\rmH_{ij},$ 
where \( w_{ij} \) quantifies the strength and direction of regulation from $\genei$ to $\genej$ and the 
operator \( \rmH_{ij} \) represents the pairwise interaction term that captures the \textit{quantum-like} directed regulatory influence of $\genei$ on $\genej$.
The structure of the interaction terms 
$\rmH_{ij}$ is as follows. Each such term must (i) act on the gene 
$\genej$ only when gene 
$\genei$ is expressed, and (ii) given $\genei$ is expressed, it must change the state of $\genej$ between $|0\rangle_j\eqand |1\rangle_j$.
The first conditionality condition is naturally enforced by the operator \(\lvert 1 \>\!\< 1 \rvert_i\). \addon{Note that operators such as \(\lvert 1 \>\!\< 0 \rvert_i\) and \(\lvert 0 \>\!\< 1 \rvert_i\) are not suitable in this context, as they would change the state of $\genei$ instead of simply conditioning the interaction on its presence.} Therefore, $\rmH_{ij}$ takes the form $\rmH_{ij} = |1\>\<1|_i \tensor O_j$, where $O_j$ is an operator acting on $\genej$.

The second condition requires that the operator $O_j$ contains only off-diagonal elements, i.e., 
 \(\lvert 1 \>\!\< 0 \rvert_j\) (transition from unexpressed to expressed) and  
 \(\lvert 0 \>\!\< 1 \rvert_j\) (transition from expressed to unexpressed). To determine the relative signs of these transitions, recall that the system evolves according to the Schrodinger equation $\tfrac{\rmd}{\dt}|\psi(t)\> = -\imag \rmH |\psi(t)\>$. The generator $-\imag w_{ij} \rmH_{ij}$ governs the direction of regulatory influence of $\genei$ on $\genej$. To encode a meaningful distinction between activation \(|0\>_j\rightarrow|1\>_j\) and repression \(|1\>_j\rightarrow|0\>_j\), we impose a condition on the action of this generator. When $\genei$ is expressed, we require
\[
-\imag w_{ij}\rmH_{ij}\,\ket{1}_i\ket{0}_j = +\,w_{ij}\ket{1}_i\ket{1}_j,\qquad
-\imag w_{ij}\rmH_{ij}\,\ket{1}_i\ket{1}_j = -\, w_{ij}\ket{1}_i\ket{0}_j.
\]
This asymmetry condition ensures that the sign associated with the transitions between the expression states of \(\genej\) is consistent with our convention for $w_{ij}$. In particular, when \(w_{ij}>0\), corresponding to activation, the transition \(|0\rangle_j \to |1\rangle_j\) appears with a positive sign, while the reverse transition \(|1\rangle_j \to |0\rangle_j\) appears with a negative sign. Conversely, when \(w_{ij}<0\), corresponding to repression, the sign assignments are reversed: the transition \(|1\rangle_j \to |0\rangle_j\) receives a positive sign and \(|0\rangle_j \to |1\rangle_j\) receives a negative sign. 
These requirements uniquely lead to the selection of the Pauli operator $\rmY_j$. Moreover, this construction preserves Hermiticity. \addon{Note that $\rmX_j$ also contains only off-diagonal elements. However, it assigns identical signs for both $|0\>_j \rightarrow |1\>_j\eqand|1\>_j \rightarrow |0\>_j$. As a result, it does not provide a distinction between activation and repression at the level of the generator.} 
Taking into account all the requirements, the GRN Hamiltonian is therefore expressed as
\[
\rmH(\bfw)
= \sum_{\substack{(i,j) \in \calE}} 
w_{ij}\,\big( 
\ket{1}\!\bra{1}_i \otimes 
\big( i\ket{1}\!\bra{0}_j - i\ket{0}\!\bra{1}_j \big)
\big) = \sum_{\substack{(i,j) \in \calE}} w_{ij}\,\tfrac{1}{2}(\rmI-\rmZ_i)\tensor \rmY_j.
\] 

\subsection{Construction of the IC-POVM for GRNs (Eqn.~\ref{eqn:ic-povm_qgrn})}
We employ an \textit{informationally complete} measurement in QHGM, assuming that gene expression data are sufficient to reconstruct the underlying regulatory state of the system. In other words, by observing the distribution of expression levels across all genes, we can infer the latent state that encodes the regulatory logic of the network. This aligns with a common assumption in systems biology that transcriptomic data across sufficiently enough pseudotime can be used to infer gene–gene interactions \cite{marbach2012wisdom, aibar2017scenic, papili2018sincerities, sanchez2018bayesian}.

The main idea behind the construction of the IC-POVM is that its outcomes should provide a discrete representation of gene expression in a cell. Since we represent a gene as a qubit, the IC-POVM requires at least four linearly independent elements \cite{medlock2020informationally}. 
To construct these four IC-POVM elements, consider the parameterized Bloch-sphere representation of a single-qubit Hermitian operator in terms of Pauli operators \cite{nielsen2010quantum}
\[
\Lambda_m
=
\frac{1}{4}
\left(
I + \vec r_m \cdot \vec\sigma
\right),
\qquad
m\in \{0,1,2,3\},
\]
where \(\vec r_m \in \mathbb{R}^3\) are Bloch vectors and
\(\vec\sigma = (\rmX,\rmY,\rmZ)\). The POVM condition requires \(\sum_{m=0}^3 \Lambda_m = \rmI\),
which implies 
\(\sum_{m=0}^3 \vec r_m = 0\) and 
\(\|\vec r_m\|_2 = 1\) for all \(m\). Furthermore, 
the informational completeness requires the four operators
\(\{\Lambda_m\}_{m=0}^3\) to be linearly independent, i.e., the following $4\times 4$ matrix
\[
R = \begin{bmatrix}
1 & 1 & 1 & 1 \\
\vec r_0 & \vec r_1 & \vec r_2 & \vec r_3
\end{bmatrix}
\]
having full rank. The constraints imposed by normalization, positivity, and informational completeness define an underdetermined system for the Bloch vectors \(\{\vec r_m\}_{m=0}^3\). As a result, many IC-POVM constructions are possible. However, we construct a set of four elements whose measurement outcomes $m$ correspond to distinct levels of gene expression.

\noindent To this end, we introduce \emph{expression score} $\tau_m \in [0,1],$ defined as \[\tau_m := \frac{\tr(|1\>\<1|\; \Lambda_m)}{\tr(\Lambda_m)} = \frac{(1-\cos\alpha_m)}{2},\] where $\alpha_m$ measures the alignment of $\vec r_m$ with the reference state \(\ket{0}\) and $\cos\alpha_m$ corresponds to the $\rmZ$-axis component of the Bloch vector $\vec r_m$. For instance, \(\vec r_m=(0,0,\cos 0)=(0,0,1)\) yields \(\tau_m=0\), corresponding to the state \(\ket{0}\) and \(\vec r_m=(0,0,\cos \pi)=(0,0,-1)\) yields \(\tau_m=1\), corresponding to the state \(\ket{1}\).
 Note that $\tau_m$ is not a Born probability
of any measurement on the unknown state; rather, it provides a fixed, geometry-based encoding of measurement outcomes that is consistent with our 
$\rmZ$-axis interpretation of a gene being unexpressed or expressed within a cell. Guided by this interpretation, we choose angles $\{\alpha_m\}_{m=0}^3$ symmetrically as \({\pi}/{6}\), \({2\pi}/{6}\), \({4\pi}/{6}\), and \({5\pi}/{6}\), respectively. 
This choice yields an ordered set of expression scores: $\tau_0 = 0.0670, \tau_1 = 0.25, \tau_2 = 0.75, \tau_3 = 0.9330$, representing four distinct levels of gene expression ranging from lower to higher expression, respectively. The associated \(\rmZ\)-components of Bloch vectors $\{\vec r_m\}_{m=0}^3$ are \({\sqrt{3}}/{2}\), \({1}/{2}\), \(-{1}/{2}\), and \(-{\sqrt{3}}/{2}\), respectively, which remains consistent with the zero-sum constraint \(\sum_{m=0}^3 \vec r_m = 0\). The \(\rmX\) and \(\rmY\) components of Bloch vectors are then selected so that each \(\vec r_m\) is a unit vector and the matrix \(R\) is of full rank. One explicit solution satisfying these constraints is given as 
$\vec r_0 = (1/2,0,\sqrt{3}/2),\; \vec r_1 = (-1/2,-1/\sqrt{2},1/2),\; \vec r_2 = (-1/2,1/\sqrt{2},1/2),\eqand \vec r_3 = (1/2,0,-\sqrt{3}/2)$, which gives the IC-POVM given in
\eqref{eqn:ic-povm_qgrn}.

\begin{remark}
    A canonical example of an informationally complete measurement for a qubit is the symmetric informationally complete POVM (SIC-POVM) \cite{renes2004symmetric}, whose four elements correspond to Bloch vectors pointing to the vertices of a regular tetrahedron.  
While mathematically elegant, this construction does not readily admit a biologically interpretable ordering of measurement outcomes. The Bloch vectors corresponding to the SIC-POVM are given as \[
\begin{aligned}
\vec r_0 = (0,\,0,\,1),
\vec r_1= \left(\frac{2\sqrt{3}}{3},\,0,\,-\frac{1}{{3}}\right),
\vec r_2 = \left(-\frac{\sqrt{2}}{3},\,\frac{\sqrt{6}}{3},\,-\frac{1}{{3}}\right),\eqand
\vec r_3 = \left(-\frac{\sqrt{2}}{3},\,-\frac{\sqrt{6}}{3},\,-\frac{1}{{3}}\right).
\end{aligned}
\]
The resulting expression scores are as follows: $\tau_0=0, \tau_1 = \tau_2 = \tau_3 = 0.667$. Thus, three of the four SIC-POVM elements give identical $\tau_m$. This degeneracy prevents a meaningful ordering of measurement outcomes, making SIC-POVM unsuitable for the discrete representation of gene expression in a cell.
\end{remark}
{\begin{remark}[Infeasible symmetric angles]
  A natural symmetric choice of angles is
$\{0,\; {\pi}/{3},\; {2\pi}/{3},\; \pi \}.$
However, imposing these angles together with the normalization and zero-sum constraints leads to Bloch vectors of the form
$\vec r_0 = (0,\,0,\,1),\;
\vec r_1 = (a,\, \pm\sqrt{1-a^2},\, 0),\;
\vec r_2 = (-a,\, \mp\sqrt{1-a^2},\, 0),\eqand
\vec r_3 = (0,\,0,\,-1),$
for some \(a \in [0,1]\).
As a result, the matrix \(R\) constructed from these vectors is rank-deficient. This means the corresponding POVM elements are not linearly independent, making it impossible to construct an IC-POVM from this selection.
\end{remark}}
{\begin{remark}[Uniform discretization]
An alternative feasible construction can be obtained by choosing evenly spaced angles, i.e.,
\(\alpha_m = (m+1)\pi/5\), \(m=0,\ldots,3\). However, this construction yields nearly uniform bins for discretizing normalized scRNA-seq data. In contrast, our choice
\(\{\pi/6,2\pi/6,4\pi/6,5\pi/6\}\) deliberately produces non-uniform bin widths (see Fig.~\ref{fig:algorithm}B). This asymmetry is beneficial in the current gene expression context. Since near-unexpressed and strongly expressed regimes, i.e., $m=0 \eqand m=3$, are more decisive for characterizing gene regulation. Therefore, these extreme categories must be assigned finer resolution through smaller bin widths, whereas the intermediate expression regimes, i.e., $m=1 \eqand m=2$, can accommodate coarser grouping.
\end{remark}}

\vspace{5pt}
\noindent\textbf{Acknowledgments}

\noindent This research was supported in part through computational resources and services provided by Advanced Research Computing at the University of Michigan, Ann Arbor.

\vspace{5pt}
\noindent\textbf{Data Availability}

\noindent The scRNA-seq data analyzed in this study are publicly available in the core GBMap section through the CZ CELLxGENE Discover portal.\\
Link: \url{https://cellxgene.cziscience.com/collections/999f2a15-3d7e-440b-96ae-2c806799c08c}.

\vspace{5pt}
\noindent\textbf{Code Availability}

\noindent All the experimental results and source codes are available at \url{https://github.com/mdaamirQ/QHGM}.

\bibliographystyle{ieeetr}  
\bibliography{reference}  
\newpage
\appendix
\begin{bibunit}[plain]
\subsection{Numerical Experiments and Implementation Details }
\label{app:simulations_details}
\noindent \textit{Glioblastoma dataset.} We analyzed the {core GBMap dataset}, a harmonized scRNA-seq atlas of IDH-wild-type glioblastoma patients from ``Charting the Single-Cell and Spatial Landscape of IDH-Wildtype Glioblastoma with GBmap'' \cite{ruiz2025charting}. The core GBMap dataset comprises approximately 330{,}000 cells from 109 patients and spans 11 anatomical regions. The GBMap consortium performed comprehensive preprocessing, including quality control and batch correction using an scVI-based integration pipeline \cite{gayoso2022python}. For this study, we focus exclusively on cells annotated as either ``Differentiated-like'' or ``Stem-like'' in the {annotation level 2} classification.  After subsetting, the dataset contains approximately 127,000 cells, distributed across {annotation level 3} as follows: AC-like (50,847), MES-like (33,167), NPC-like (22,117), and OPC-like (21,390)  (see Fig. \ref{fig:finalgrn}A). Recent studies indicate that OPC-like cells in glioblastoma possess high proliferative capacity and tumorigenic potential, and hence, are chosen as the root node cell. They are enriched in both adult and pediatric tumors and exhibit cellular plasticity that enables transitions to other glioblastoma cell states \cite{neftel2019integrative} \cite{zamler2023primitive}. Differential expression (DE) analysis \cite{wolf2018scanpy} is performed to identify transcriptional (gene expression) differences between the OPC-like and MES-like, and the resulting ranked gene lists were subsequently mapped onto the MSigDB Hallmark gene set collection \cite{liberzon2015molecular}, enabling functional interpretation of the DE signatures. This step identifies which well-studied biological pathways and processes are enriched, making it easier to understand the key genes and functions involved. The final set of $14$ genes are the following: \textsf{BCAN}, \textsf{STMN1}, \textsf{HES6}, \textsf{ETV1}, \textsf{CADM2}, \textsf{MMP16}, \textsf{CKB}, \textsf{LIMA1}, \textsf{VCAN}, \textsf{JPT1}, \textsf{ASCL1}, \textsf{CDK4}, \textsf{TUBB2B}, and \textsf{NCAM1}. 

\vspace{5pt}
\noindent\textit{Training Details.}  
Since the released dataset does not include the learned scVI.SCANVI embeddings, we will recompute them to ensure reproducibility in downstream analyses, particularly for the calculation of PCA embeddings and pseudotime. We will follow the parameters from the original study while making necessary adjustments to accommodate updates in the newer version of scvi-tools. We use the py-VIA package \cite{stassen2021generalized} on the PCA embedding (k = 50, Jacobian-weighted edges as True) to find pseudotime values for each cell. An OPC-like cell {(barcode: 118\_1\_29)} is selected as the root, with AC-like and MES-like states defined as terminal groups. Disconnected components were excluded, and a fixed random seed (42) is used to ensure reproducibility. This assigns each cell a pseudotime value along a differentiation trajectory, providing the approximate temporal index $t$ in the subsequent quantum simulation. The scRNA-seq data is preprocessed using the Scanpy toolkit \cite{wolf2018scanpy}, and the expression values are subsequently transformed using Min-Max normalization per cell with scikit-learn \cite{pedregosa2011scikit}.
To capture approximate temporal structure, the pseudotime vector is partitioned into  $\sf N_t=50 $ equally populated bins. For each pseudotime bin and each cell within the bins, the $n$-length vector of discretized gene expression, denoted as $\bfm_{(i,k)}$, is encoded into a single integer using a base-4 representation. This is expressed by the formula:
$\sum_{j=1}^{n} m_{(i,k)}^{(j)}\,4^{(j-1)},  \text{where } m_{(i,k)}^{(j)} \in \{0,1,2,3\}.$
This encoding process converts a three-dimensional dataset $\mathcal{G}$ (pseudotime bins  \(\times\) cells \(\times\) genes) into a two-dimensional dataset (pseudotime bins  \(\times\) cells). We run 10 different simulations consisting of two complementary initial state preparation strategies. In the first approach, all parameters, including rotation angles $\theta_i$ and $\phi_i$ for each gene \eqref{eqn:initial_state}, as well as the parameters $w_{ij}$, are treated as learnable parameters. In the second approach, the initial state preparation parameters are fixed at $\theta_i = \pi/2$ and $\phi_i = 0$, while only the parameters are learnable. For each approach, we perform five random initializations of the parameters to ensure robustness and run for 3000 epochs. The model parameters $(\boldsymbol{\theta},\boldsymbol{\phi},\mathbf{w})$ are updated using the ADAM optimizer, with an adaptive learning rate 
$0.085 / \sqrt{(\mathrm{epoch}/4) + 1}$, 
providing smoother convergence during later training stages. 
Across these independent runs, we compute the median of each weight, providing a stable, representative estimate of network interactions while accounting for parameter initialization variability. This strategy allows us to capture a more comprehensive and robust understanding of the underlying regulatory network (see Fig. \ref{fig:finalgrn}B). 

\vspace{5pt}
\noindent\textit{Synthetic Data Generation.}
We first construct a ground-truth GRN consisting
of $n=12$ genes, corresponding to a $12$-qubit system. The parameters
$\bfw$ are sampled independently and uniformly from the interval $[-w_{\max},w_{\max}]$, where $w_{\max}=1$. The initial state
parameters are sampled randomly, with $\theta_i \sim \mathrm{Unif}[0,\pi]$ and
$\phi_i \sim \mathrm{Unif}[0,2\pi]$ for each gene.
Given the ground-truth parameters, gene-expression data are generated by
simulating QHGM dynamics starting from an initial product state. We sample $\sfN_t=65$ pseudotime points
uniformly from the interval $[0,1]$ and, at each time point, we generate $\sfN_c=6000$
measurement samples using the fixed IC-POVM measurement \eqref{eqn:ic-povm_qgrn}. This results in a
discretized gene-expression dataset
$\calG \in \{0,1,2,3\}^{65 \times 6000 \times 12}$.

\vspace{5pt}
\noindent\textit{Training Details.}
We implement VQ-Net using PennyLane \cite{bergholm2018pennylane} and JAX \cite{jax2018github}. Furthermore, to convert this constrained optimization into an unconstrained form, the algorithm introduces a variable $\tilde{w}_{ij}\in\mathbb{R}$ and reparameterizes the Hamiltonian parameters as 
$w_{ij} = w_{\max}\tanh(\tilde{w}_{ij})$.
This transformation guarantees that $w_{ij}$ always remains within the valid range during optimization while using optimizers over unconstrained variables $\tilde{w}_{ij}$. 
The latent interaction parameters $\tilde w_{ij}$ are initialized uniformly in $[-0.5,\,0.5]$, while the initial state parameters are initialized as
$\theta_i \in [\pi/4,\,3\pi/4]$ and $\phi_i \in [\pi/2,\,3\pi/2]$.
The parameters $\boldsymbol{\theta}$ and $\boldsymbol{\phi}$ are optimized jointly with
$\bfw$. The parameters are not explicitly constrained during training because the learning objective relies solely on the induced separable initial state, rather than on the specific values of \(\theta_i\) and \(\phi_i\). As a result, different parameter values that produce the same state (up to a global phase) are considered equivalent for the purposes of optimization.
Optimization is performed using the Adam optimizer from Optax with an empirically chosen
adaptive learning rate $0.85/\sqrt{\text{epoch}/4 + 1}$. All models are trained for
$2500$ epochs with batch size $\sfB=20$ on NVIDIA A40 GPUs.

\subsection{Converting Discretized QHGM-Generated Data to Continuous Gene Expression Levels}
\label{experimental classical}
To transform discrete labels from QHGM into continuous gene expression values, we employ a localized dequantization strategy based on Beta distributions. We set the expression scores
$\tau_0=0.067$, $\tau_1=0.25$, $\tau_2=0.75$, and $\tau_3=0.933$, as the mode of the four Beta distributions (see Fig.~\ref{fig:beta}). To generate samples within a bin, we map the interval $[b_{i-1}, b_i]$ to the standard support of the Beta distribution via a linear transformation $z = (x - b_{i-1}) / (b_i - b_{i-1})$, such that the resulting continuous value is recovered via $x = b_{i-1} + z(b_i - b_{i-1})$, where $z \sim \text{Beta}(\alpha_i, \beta_i)$.

The shape parameters $\alpha_i$ and $\beta_i$ are determined by aligning the
mode of the Beta distribution with a representative target value
$\tau_i \in [b_{i-1}, b_i]$, such that the probability mass within each bin
congregates around $\tau_i$. Defining the normalized location
$\gamma_i = {(\tau_i - b_{i-1})}/{(b_i - b_{i-1})},$
and introducing a concentration parameter $c_i = \alpha_i + \beta_i$ to control the spread, we use the mode-based relation
$\gamma_i = {(\alpha_i - 1)}{(\alpha_i + \beta_i - 2)}$
to obtain
$\alpha_i = 1 + \gamma_i (c_i - 2), \quad
\beta_i = 1 + (1 - \gamma_i)(c_i - 2).$
The concentration parameter $c_i$ is selected such that \(\int_{0.025}^{0.975} \text{Beta}(z; \alpha_i, \beta_i)\, dz = 0.99, \) ensuring that $99\%$ of the probability mass lies within the central $95\%$ of the normalized interval. 

\begin{figure}
    \centering
    \includegraphics[width=\linewidth]{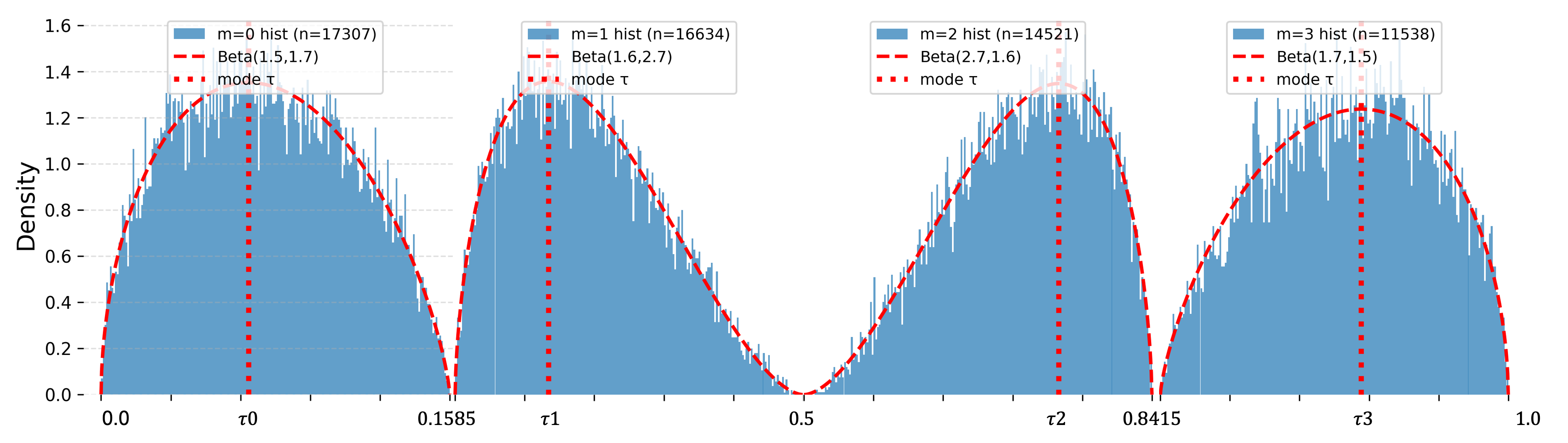}
    \caption{Visualization of continuous value gene expression data generated from discrete QHGM-generated data using localized Beta distributions.} 
    \label{fig:beta}
\end{figure}

\vspace{5pt}
\noindent\textit{Training Details.} ARACNE estimates pairwise gene dependencies by computing Spearman rank–based mutual information, followed by network reconstruction, yielding a weighted adjacency matrix of inferred interactions. In parallel, partial correlations are estimated using a shrinkage Gaussian graphical model implemented in GeneNet, and edge significance is evaluated. A symmetric adjacency matrix is constructed from the resulting partial correlation coefficients, with diagonal elements set to 0. Additionally, regulatory interactions are inferred using GENIE3, a tree-based ensemble method based on random forests, with all genes considered as candidate regulators, 500 trees per target gene, and the number of variables tried at each split set to the square root of the total number of genes. 
In addition, time-series–based regulatory interactions are inferred using SINCERITIES, which estimates signed gene–gene influences from expression dynamics. It runs in R using distance = 1 (Kolmogorov–Smirnov distributional distance between time points), method = 1 (ridge regression for regularization), noDIAG = 1 (self-regulatory edges not allowed), and SIGN = 1 (estimation of activation vs repression to obtain a signed adjacency matrix). The resulting adjacency matrix is subsequently normalized by dividing all entries by the maximum edge weight to yield a unit-scaled matrix used for downstream analyses.

For both GeneNet and SINCERITIES, given a predicted weighted adjacency matrix $\hat{\mathbf{A}}$ and the ground-truth signed adjacency matrix $\mathbf{A}$, we first discretize the predicted interactions by mapping positive weights to $+1$ (upregulation), negative weights to $-1$ (downregulation), and zero to $0$ (no interaction). Self-interactions are excluded by removing diagonal entries. Performance is then evaluated by comparing the flattened off-diagonal entries of $\hat{\mathbf{A}}$ and $\mathbf{A}$ using macro-averaged precision, recall, and F$_1$ score over the three classes $\{-1, 0, +1\}$, along with overall accuracy. This metric jointly assesses correct edge detection and regulatory sign assignment. For network edge recovery, we ignore the interaction sign and consider only the presence or absence of edges. Both the predicted and ground-truth adjacency matrices are binarized, with nonzero entries indicating an interaction and zero entries indicating no interaction. Diagonal elements are excluded. Precision, recall, F$_1$ score, and accuracy are computed by comparing the flattened off-diagonal binary matrices. This metric assesses the ability to accurately recover the network topology, regardless of the regulatory sign.

\subsection{Useful Lemmas}\label{app:primer}
\begin{lemma}[Matrix Bernstein-type inequality {\cite[Corollary~6.2.1]{tropp2015introduction}}]
\label{lem:matrix_sampling}
Let $\mu_\rmX$ be a fixed $d_1 \times d_2$ matrix. Construct a random matrix
$\rmX \in \mathbb{C}^{d_1 \times d_2}$ such that
$\EE\rmX = \mu_\rmX$ and $
\|\rmX\| \leq L.$
Let
\[
\bar{\rmX}_n := \frac{1}{n}\sum_{k=1}^{n} \rmX_k,
\]
where $\{R_k\}_{k=1}^n$ are independent copies of $R$. Then, for all $t\geq 0$,
\begin{equation}
\Pr\left\{
\|\bar{\rmX}_n - \mu_\rmX\| \geq t
\right\}
\leq
(d_1 + d_2)
\exp\!\left(
-\frac{n t^2/2}{\sigma^2 + 2Lt/3}
\right),
\end{equation}
where $\sigma^2$ is the per-sample second moment, defined as 
$\sigma^2
:= \max\big\{
\big\|\EE(\rmX\rmX^{\dagger})\big\|,
\big\|\EE(\rmX^{\dagger}\rmX)\big\|
\big\}.$
In other words, for any $\delta \in (0,1)$, the following inequality holds:
\begin{equation}
\|\bar{\rmX}_n - \mu_\rmX\|
\leq
\sqrt{
\frac{2 \sigma^2}{n}\log\left(\frac{(d_1 + d_2)}{\delta}\right)
}
+
\frac{2 L}{3n} \log\left(\frac{(d_1 + d_2)}{\delta}\right).
\end{equation}
with probability at least $(1-\delta)$.
\end{lemma}

\begin{lemma}
\label{lem:khintchine}
Consider a finite sequence $\{\Lambda_k\}$ of fixed $\sfD\times \sfD$ Hermitian matrices, and let $\{\sigma_i\}$ be i.i.d.\ Rademacher random variables.
Then the following inequality holds:
\[
\EE_\sigma\Big\|\sum_{k} \sigma_k \Lambda_k\Big\|
\;\le\;
\sqrt{2\log(2\sfD)}\;\Big\|\sum_k \Lambda_k^2\Big\|^{1/2}.\]
\end{lemma}
\begin{proof}
    The proof is provided in Appendix \ref{app:proof:lem:khintchine}.
\end{proof}
\begin{lemma}[Parametric derivatives of an evolved expectation]\label{lem:grad-hess-exp}
Consider a time-independent Hamiltonian
$\rmH(\bfw) \;=\; \sum_{k=1}^{c} w_k\, \rmH_k$.
Let $U_t(\bfw) = e^{-i t \rmH(\bfw)}$ and assume the system starts in the state $\rho_0$. For an observable $\Lambda_\bfm$ satisfying $\|\Lambda_\bfm\|\le 1$, define $p(\bfm|t,\bfw):=\tr(\Lambda_\bfm\,\rho_t(\bfw)).$
Then, for every $\bfm\in\calM^n$ and $t\in(0,t_{\max}]$, the following statements hold.
\begin{itemize}
    \item[$(i)$]The likelihood function $p(\bfm|t,\bfw)$ is twice differentiable in $\calW_B$. The first-order partial derivative with respect to the parameter $w_k$ is given by 
\begin{equation}\label{eq:first-deriv-exp}
\frac{\partial}{\partial w_k}\,\tr(\Lambda_\bfm\,\rho_t(\bfw))
\;=\;
-\imag \int_{0}^{t} \tr\Big( \Lambda_\bfm(t)\,\big[\rmH_k(s),\,\rho_0\big] \Big)\, \ds .
\end{equation}
Moreover, the second-order partial derivative with respect to parameters
\(w_j\) and \(w_k\) is given by
\begin{equation}\label{eq:second-deriv-exp}
\begin{aligned}
\frac{\partial^2}{\partial w_j\,\partial w_k}\,\tr(\Lambda_\bfm\,\rho_t(\bfw))
&=
-\int_{0}^{t}\int_{0}^{s_1}
\tr\Big(
\Lambda_\bfm(t)\,
\big[\rmH_j(s_1),\,[\rmH_k(s_2),\,\rho_0]\big]
\Big)\,\ds_2\,\ds_1 \\
&\quad
-\int_{0}^{t}\int_{0}^{s_1}
\tr\Big(
\Lambda_\bfm(t)\,
\big[\rmH_k(s_1),\,[\rmH_j(s_2),\,\rho_0]\big]
\Big)\,\ds_2\,\ds_1,
\end{aligned}
\end{equation} 
where 
$\Lambda_\bfm(t) := U_t^\dagger(\bfw)\, \Lambda_\bfm\, U_t(\bfw), 
\eqand
\rmH_k(s) := U_s^\dagger(\bfw)\, \rmH_k\, U_s(\bfw).$ 
\item[$(ii)$] Bounded gradient and hessian. 
\begin{itemize}
   \item For all $\bfw\in\calW_B$,
\(
\|\nabla_{\bfw}p(\bfm|t,\bfw)\|\le L_\phi\, \eqand \,\|\nabla_{\bfw}^2p(\bfm|t,\bfw)\|\le L_\phi^2.
\)
\end{itemize}
\item[$(iii)$] Lipschitz continuity. For all $\bfw, \bfw'\in\calW_B$,
\begin{itemize}
\item $|\phi{(\bfm|t,\bfw)} - \phi{(\bfm|t,\bfw')}|\leq L_\phi\|\bfw-\bfw'\|$
\item $\|\nabla_{\bfw}\phi{(\bfm|t,\bfw)} - \nabla_{\bfw}\phi{(\bfm|t,\bfw')}\|\leq L_\phi^2\|\bfw-\bfw'\|$
    \item   $\|\nabla_{\bfw}^2p(\bfm|t,\bfw) - \nabla_{\bfw}^2p(\bfm|t,\bfw')\| \leq 2L_\phi^3\|\bfw-\bfw'\|,$
\end{itemize}
\end{itemize}
where $L_\phi := 2t_{\max}\sqrt{c}\|\rmH\|_{\infty,\max}.$
\end{lemma}
\begin{proof}
    The proof is provided in Appendix \ref{app:proof:lem:grad-hess-exp}.
\end{proof}
\subsection{Proof of Equation \ref{eqn:bound_Fbar_ti}}
\label{app:proof:eqn:bound_Fbar_ti}
Consider the following inequalities:
\begin{align*}
    \|\nabla^2_\bfw \log(\phi(\bfm|t,\bfw))\| &= \Big\|\frac{\nabla^2_\bfw \phi(\bfm|t,\bfw)}{\phi(\bfm|t,\bfw)} - \frac{\nabla_\bfw \phi(\bfm|t,\bfw) \nabla_\bfw \phi(\bfm|t,\bfw)^{\intercal}}{\phi(\bfm|t,\bfw)^2}\Big\|\\
    &\leq \Big\|\frac{\nabla^2_\bfw \phi(\bfm|t,\bfw)}{\phi(\bfm|t,\bfw)}\Big\| + \Big\|\frac{\nabla_\bfw \phi(\bfm|t,\bfw) \nabla_\bfw \phi(\bfm|t,\bfw)^{\intercal}}{\phi(\bfm|t,\bfw)^2}\Big\|\\
     &\overset{a}{\leq} \frac{\|\nabla^2_\bfw \phi(\bfm|t,\bfw)\|}{p_{\min}} + \frac{\|\nabla_\bfw \phi(\bfm|t,\bfw)\|^2}{p_{\min}^2}\overset{b}{\leq} \frac{L_\phi^2}{p_{\min}} + \frac{L_\phi^2} {p_{\min}^2} \leq \frac{2L_\phi^2}{p_{\min}^2},
\end{align*}
where $(a)$ follows from the fact that for any vector $v$, the spectral norm of the outer product $\|vv^\intercal\|$ equals $\|v\|^2$, and $(b)$ follows from Lemma \ref{lem:grad-hess-exp}. Next, consider the following inequalities:
\begin{align*}
    \|\bar{H}_{t}(\bfw)\|
    &\overset{c}{\leq}\EE_{\bfm \sim \phi(\cdot|t,\bfw^*)}\|\nabla^2_\bfw \log(\phi(\bfm|t,\bfw))\| \leq \frac{2L_\phi^2}{p_{\min}^2},
\end{align*}
where $(c)$ follows from Jensen's inequality.

\subsection{Proof of Proposition \ref{prop:FisherBound}}
\label{app:proof:prop:FisherBound}
For a given $t_i$, the random matrices $\widehat{H}_{(i,k)}(\bfw)$ are independent satisfying
$\EE_{\phi(\cdot|t_i,\bfw^*)}\big[\widehat{H}_{(i,k)}(\bfw)\big] = \bar{H}_{t_i}(\bfw),$ and from Eqn.~\eqref{eqn:bound_Fbar_ti}, $\|\widehat{H}_{(i,k)}(\bfw)\| \leq 2(L_\phi/p_{\min})^2$. 
Therefore, using Lemma~\ref{lem:matrix_sampling}, we obtain, for all $\varepsilon_1 > 0$,
\[\Pr\Big\{\Big\|\frac{1}{\sfN_c}\sum_{k=1}^{\sfN_c}\widehat{H}_{(i,k)}(\bfw) - \bar{H}_{t_i}(\bfw)\Big\| \geq \varepsilon_1\Big\}\leq 2c \cdot\exp\left\{\frac{-\sfN_c\; \varepsilon_1^2/2}{\sigma^2_i + 4\varepsilon_1 L_\phi^2/3p_{\min}^2}\right\},\]
where $\sigma^2_i:= \big\|\EE_{\phi(\cdot|t_i,\bfw^*)}\big[(\widehat{H}_{(i,k)}(\bfw))^2\big]\big\|.$ Applying the union bound, we get 
\begin{align*}
    \Pr\Big\{\max_{1\leq i \leq \sfN_t} \|\widehat{H}_{t_i}(\bfw) - \bar{H}_{t_i}(\bfw)\| \geq \varepsilon_1\Big\} &= \Pr\Big\{\bigcup_{i=1}^{\sfN_t}\{\|\widehat{H}_{t_i}(\bfw) - \bar{H}_{t_i}(\bfw)\| \geq \varepsilon_1\}\Big\} \\
    &\leq 2c\, \sfN_t \exp\left\{\frac{-\sfN_c\; \varepsilon_1^2/2}{\sigma^2_{\max} + 4\varepsilon_1 L_\phi^2/3p_{\min}^2}\right\}. 
\end{align*}
It remains to bound $\sigma_{\max}^2:=\max_i \sigma_i^2$, which we do next using Lemma ~\ref{lem:grad-hess-exp}. Consider the following inequalities:
\begin{align*}
    \sigma^2_i &\leq \EE_{\phi(\cdot|t_i,\bfw^*)}\big\|(\nabla^2_\bfw \log(\phi(\bfm_{(i,k)}|t_i,\bfw)))^2\big\|\overset{}{=} \EE_{\phi(\cdot|t_i,\bfw^*)}\big\|\nabla^2_\bfw \log(\phi(\bfm_{(i,k)}|t_i,\bfw))\big\|^2 \leq 4({L_\phi/p_{\min}})^4,  
\end{align*}
where the equality follows from the property that for any Hermitian matrix \(\rmA\), it holds that \(\|\rmA^2\| = \|\rmA\|^2\) and the last inequality follows from Eqn.~\eqref{eqn:bound_Fbar_ti}. Hence, $\sigma_{\max}^2 \leq 4(L_\phi/p_{\min})^4$. This completes the bound on 
$\rmT_1$. We next bound $\rmT_2$, which captures the deviation of the per-time expected Hessian from the expected Hessian. Observe that the random matrices  $\bar{H}_{t_i}(\bfw)$ are independent satisfying $\EE_{t_i}[\bar{H}_{t_i}(\bfw)] = \bar{H}(\bfw)$, and using \eqref{eqn:bound_Fbar_ti}  $\|\bar{H}_{t_i}(\bfw)\| \leq 2(L_\phi/p_{\min})^2$. Therefore, applying Lemma \ref{lem:matrix_sampling}, we obtain, for all $\varepsilon_2> 0$,
\[\Pr\Big\{\Big\|\frac{1}{\sfN_t}\sum_{i=1}^{\sfN_t}\bar{H}_{t_i}(\bfw) - \bar{H}(\bfw)\Big\| \geq \varepsilon_2\Big\}\leq 2c \cdot\exp\left\{\frac{-\sfN_c\; \varepsilon_2^2/2}{\sigma^2 + 2\varepsilon_2 L_\phi^2/3p_{\min}^2}\right\},\]
where $\sigma^2:= \|\EE_{t\sim \pi}[(\bar{H}_t(\bfw))^2]\| \leq 4(L_\phi/p_{\min})^4.$ This completes the proof of Proposition \ref{prop:FisherBound}.

\subsection{Proof of Proposition \ref{prop:FisherLipschitz}}
\label{app:proof:prop:FisherLipschitz}
Define $s(\bfm,t,\bfw):=-\nabla_\bfw^2 \log(\phi(\bfm|t,\bfw)),$ which can be simplified as
\[s(\bfm,t,\bfw) = \frac{\nabla_\bfw \phi(\bfm|t,\bfw)\nabla_\bfw \phi(\bfm|t,\bfw)^\intercal}{ \phi(\bfm|t,\bfw)^2}-\frac{\nabla^2_\bfw \phi(\bfm|t,\bfw)}{\phi(\bfm|t,\bfw)}.\]
For every $\bfm\in\calM^n \eqand t\in(0,t_{\max}]$, consider the following inequalities:
\begin{align*}
    &\|s(\bfm,t,\bfw) -  s(\bfm,t,\bfw')\| \\
    &\leq \Big\|\frac{\nabla_\bfw \phi(\bfm|t,\bfw)\nabla_\bfw \phi(\bfm|t,\bfw)^\intercal}{\phi(\bfm|t,\bfw)^2} \!-\! \frac{\nabla_\bfw \phi(\bfm|t,\bfw')\nabla_\bfw \phi(\bfm|t,\bfw')^\intercal}{\phi(\bfm|t,\bfw')^2}\Big\|\!+\!\Big\|\frac{\nabla^2_\bfw \phi(\bfm|t,\bfw)}{\phi(\bfm|t,\bfw)}\!-\! \frac{\nabla^2_\bfw \phi(\bfm|t,\bfw')}{\phi(\bfm|t,\bfw')}\Big\|\\
    &\leq \frac{1}{\phi(\bfm|t,\bfw)^2}\|\nabla_\bfw \phi(\bfm|t,\bfw)\nabla_\bfw\phi(\bfm|t,\bfw)^{\intercal} -\nabla_\bfw \phi(\bfm|t,\bfw')\nabla_\bfw\phi(\bfm|t,\bfw')^{\intercal}\|\\
    &\hspace{35pt }+ \|\nabla_\bfw\phi(\bfm|t,\bfw')\|^2 \frac{|\phi(\bfm|t,\bfw)^2-\phi(\bfm|t,\bfw')^2|}{(\phi(\bfm|t,\bfw)\phi(\bfm|t,\bfw'))^2}  + \|\nabla_\bfw^2 \phi(\bfm|t,\bfw')\| \; \frac{|\phi(\bfm|t,\bfw)-\phi(\bfm|t,\bfw')|}{\phi(\bfm|t,\bfw)\phi(\bfm|t,\bfw')}  \\
    &\hspace{35pt}+ \frac{1}{\phi(\bfm|t,\bfw)}\|\nabla_\bfw^2\phi(\bfm|t,\bfw) - \nabla_\bfw^2 \phi(\bfm|t,\bfw')\| \\
    &\leq \frac{1}{p_{\min}^2}\|\nabla_\bfw \phi(\bfm|t,\bfw)\nabla_\bfw\phi(\bfm|t,\bfw)^{\intercal} -\nabla_\bfw \phi(\bfm|t,\bfw')\nabla_\bfw\phi(\bfm|t,\bfw')^{\intercal}\|  \\
    &\hspace{35pt } + 2\frac{L_\phi^3}{p_{\min}^3} \|\bfw - \bfw'\| + \frac{L_\phi^3}{p_{\min}^2}\|\bfw - \bfw'\| + 2\frac{L_\phi^3}{p_{\min}}\|\bfw - \bfw'\|  \\
    &\leq \frac{1}{p_{\min}^2}\|\nabla_\bfw \phi(\bfm|t,\bfw) -\nabla_\bfw \phi(\bfm|t,\bfw')\|(\|\nabla_\bfw\phi(\bfm|t,\bfw)\|+\|\nabla_\bfw\phi(\bfm|t,\bfw')\|) + \frac{5L_\phi^3}{p_{\min}^3} \|\bfw - \bfw'\| \\
    &\leq \frac{2L_\phi^3}{p_{\min}^2}\|\bfw - \bfw'\| + \frac{5L_\phi^3}{p_{\min}^3} \|\bfw - \bfw'\| \leq \frac{7L_\phi^3}{p_{\min}^3} \|\bfw - \bfw'\|.
\end{align*}
where the fourth inequality follows by adding and subtracting $\nabla_\bfw \phi(\bfm|t,\bfw)\nabla_\bfw\phi(\bfm|t,\bfw')^{\intercal}$ and using the fact that for any vectors $a$ and $b \in \mathbb{R}^n$, we have $\|ab^\intercal\| = \|a\|\|b\|$.
We are now equipped to bound the following:
\begin{align*}
    \|\widehat{H}(\bfw)-\widehat{H}(\bfw')\| &= \bigg\|\frac{1}{\sfN_t}\sum_{i=1}^{\sfN_t}\frac{1}{\sfN_c}\sum_{k=1}^{\sfN_c} s(\bfm_{(i,k)},t_i,\bfw) - \frac{1}{\sfN_t}\sum_{i=1}^{\sfN_t}\frac{1}{\sfN_c}\sum_{k=1}^{\sfN_c}s(\bfm_{(i,k)},t_i,\bfw')\bigg\|\\
    &\leq \frac{1}{\sfN_t}\sum_{i=1}^{\sfN_t}\frac{1}{\sfN_c}\sum_{k=1}^{\sfN_c}\|s(\bfm_{(i,k)},t_i,\bfw)-s(\bfm_{(i,k)},t_i,\bfw')\| \leq \frac{7L_\phi^3}{p_{\min}^3} \|\bfw - \bfw'\|.
\end{align*}
Next, applying the similar inequalities as above, we obtain $\|\bar{H}(\bfw)-\bar{H}(\bfw')\| \leq 7(L_\phi/p_{\min})^3\|\bfw-\bfw'\|.$

\subsection{Equivalent Conditions of Strong Convexity}
\label{app:strong_convexity}
We state a well-known result on the equivalence conditions of strongly convexity here for convenience \cite[Exercise 9.9]{wainwright2019high}. 
\begin{lemma}[Strong convexity]
Let $f:\mathbb{R}^n\to\mathbb{R}$ be twice continuously differentiable and let 
$\Omega \subset \mathbb{R}^n$ be a compact set. 
Then $f$ is $\mu$-strongly convex on $\Omega$ if and only if any of the following equivalent conditions hold:
\begin{enumerate}
\item[\textnormal{(i)}] For all $x,y\in \Omega$, $\quad f(x)\ \ge\ f(y)+\nabla f(y)^\top(x-y)+\frac{\mu}{2}\|x-y\|^2.$
\item[\textnormal{(ii)}] For all $x,y\in \Omega$, $\quad (\nabla f(x)-\nabla f(y))^\top(x-y)\ \ge\ \mu\|x-y\|^2.$
\item[\textnormal{(iii)}] For all $z\in \Omega$,$\quad \quad\;\nabla^2 f(z)\ \succeq\ \mu I.$
\end{enumerate}
\end{lemma}
\begin{proof}
We prove the equivalence by establishing the implications
$\textnormal{(i)} \Rightarrow \textnormal{(ii)} \Rightarrow \textnormal{(iii)} \Rightarrow \textnormal{(ii)} \Rightarrow \textnormal{(i)}.$

\medskip
\noindent\textbf{(i)$\Rightarrow$(ii).}
Applying $(i)$ to the ordered pairs $(x,y)$ and $(y,x)$ yields
\[
f(x)\ge f(y)+\nabla f(y)^\intercal(x-y)+\frac{\mu}{2}\|x-y\|^2,
\]
\[
f(y)\ge f(x)+\nabla f(x)^\intercal(y-x)+\frac{\mu}{2}\|x-y\|^2.
\]
Summing these two inequalities gives
$(\nabla f(x)-\nabla f(y))^\intercal(x-y)\ge \mu\|x-y\|^2$.

\medskip
\noindent\textbf{(ii)$\Rightarrow$(iii).}
Fix $z\in\sfB(x_0,r)$ and an arbitrary direction $v\in\sfB(x_0,r)$. Applying $(ii)$ with $x=z+tv$ and $y=z$ gives
\[
(\nabla f(z+tv)-\nabla f(z))^\intercal tv \ge \mu t^2\|v\|^2
\qquad \text{for all } t>0.
\]
After dividing both sides by $t^2$ gives
\[
\left(\frac{\nabla f(z+tv)-\nabla f(z)}{t}\right)^\intercal v \ge \mu \|v\|^2
\qquad \text{for all } t>0.
\]
Since $\nabla f$ is differentiable at $z$, its derivative at $z$ is given by the Hessian matrix $\nabla^2 f(z)$. In particular, for any direction $v\in\sfB(x_0,r)$, the directional derivative of the gradient satisfies
\[
\lim_{t\to 0} \frac{\nabla f(z+tv)-\nabla f(z)}{t}
= \nabla^2 f(z)\,v.
\]
Taking the limit $t\to 0$ in the preceding inequality therefore yields
$v^\intercal \nabla^2 f(z)\, v \ge \mu \|v\|^2.$
Since this inequality holds for every $v\in\sfB(x_0,r)$, it follows that the symmetric matrix $\nabla^2 f(z)-\mu I$ is positive semidefinite. Equivalently,
$\nabla^2 f(z)\ge \mu I.$
As $z$ was arbitrary, the result holds for all $z\in\sfB(x_0,r)$.

\medskip
\noindent\textbf{(iii)$\Rightarrow$(ii).}
Consider the line segment $\gamma(t)=y+t(x-y)$ for $t\in[0,1]$. By the fundamental theorem of calculus applied to the gradient,
\[
\nabla f(x)-\nabla f(y)
=\int_0^1 \nabla^2 f(\gamma(t))\,\gamma'(t)\,\dt.
\]
Taking the inner product with $\gamma'(t) = (x-y)$ and using $(iii)$, we obtain
\[
 (\nabla f(x)-\nabla f(y))^\intercal \gamma'(t)
=\int_0^1 (x-y)^\top \nabla^2 f(\gamma(t))\,(x-y)\,\dt
\ge \int_0^1 \mu\|x-y\|^2\,dt
=\mu\|x-y\|^2.
\]
\medskip
\noindent\textbf{(ii)$\Rightarrow$(i).}
Consider the line segment $\gamma(t)=y+t(x-y)$ for $t\in[0,1]$. By the fundamental theorem of calculus,
\[
f(x)-f(y)=\int_0^1 \nabla f(\gamma(t))^\intercal\gamma'(t)\,\dt
=\int_0^1 \nabla f(y+t(x-y))^\intercal (x-y)\,\dt.
\]
Using $(ii)$ with the pair $(y+t(x-y),y)$ yields
\[
(\nabla f(y+t(x-y))-\nabla f(y))^\intercal(x-y)
\ge \mu t\|x-y\|^2,
\]
for all $t\in[0,1]$. Substituting this bound into the integral representation above gives
\[
f(x)-f(y)\ge
\int_0^1 \left(\nabla f(y)^\intercal(x-y) + \mu t\|x-y\|^2\right)\,dt
=\nabla f(y)^\intercal(x-y)+\frac{\mu}{2}\|x-y\|^2.
\]
This completes the proof. 
\end{proof}

\subsection{Derivation of Equation \eqref{eqn:unif_1}}
\label{app:proof:eqn:unif_1}
We use the well-known McDiarmids inequality \cite[Corollary 2.21]{wainwright2019high} \cite[Lemma 26.4]{shalev2014understanding}.
Define the function
\[
f(\calS_i) = \sup_{\bfw \in \calW_B} | \widehat{L}_{t_i}(\bfw; \calS_i) -L_{t_i}(\bfw)|,
\]
where $\calS_i=\{\bfm_{(i,1)}$, $\cdots$ , $\bfm_{(i,\sfN_c)}\} \eqand \widehat{L}_{t_i}(\bfw; \calS_i) := \tfrac{1}{\sfN_c}\sum_{k} \ell(\phi(\bfm_{(i,k)}|t_i,\bfw))$. Let $\calS'_i=\{\bfm_{(i,1)}$, $\cdots$ , $\bfm'_{(i,j)}$, $\cdots$ , $\bfm_{(i,\sfN_c)}\}$ differ from $\calS_i$ in only one coordinate.
Then,
\[
\begin{aligned}
|f(\calS_i) - f(\calS'_i)|
&= \Big| \sup_{\bfw\in\calW_B} |\widehat{L}_{t_i}(\bfw;\calS_i)-L_{t_i}(\bfw)| - \sup_{\bfw\in\calW_B}|\widehat{L}_{t_i}(\bfw;\calS'_i)-L_{t_i}(\bfw)| \Big| \\[4pt]
&\le \sup_{\bfw\in\calW_B}\Big|\,|\widehat{L}_{t_i}(\bfw;\calS_i)-L_{t_i}(\bfw)| - |\widehat{L}_{t_i}(\bfw;\calS'_i)-L_{t_i}(\bfw)|\,\Big| \\[4pt]
&\le \sup_{\bfw\in\calW_B}\Big| \widehat{L}_{t_i}(\bfw;\calS_i) - \widehat{L}_{t_i}(\bfw;\calS'_i) \Big| 
\le \frac{1}{\sfN_c}\;\sup_{\bfw\in\calW_B}|\ell(\phi(\bfm_{(i,j)}|t_i,\bfw)) - \ell(\phi(\bfm'_{(i,j)}|t_i,\bfw))|,
\end{aligned}
\]
where the second inequality follows from the following result: for bounded functions \( g_1, g_2: \calX \rightarrow \RR \), we have 
$|\sup_{\calX} g_1 - \sup_{\calX} g_2| \leq \sup_{\calX} |g_1 - g_2|,$ and the third inequality follows the reverse triangle inequality.
Since
\(
|\log x - \log y | \le \log(1/a),
\) for all $a\leq x,y\leq 1$, we get \(
|f(\calS_i) - f(\calS'_i)| \le (-\log p_{\min})/{\sfN_c}.\) This means
 \(f \) satisfies the bounded-difference condition of McDiarmid’s inequality with  $c_j = (-\log p_{\min})/\sfN_c$ for all $k\in \{1,2,\cdots,\sfN_c\}$.
Therefore, using McDiarmid’s inequality,
for any \(\varepsilon>0\),
\[
\Pr\left\{
|f(\calS_i)- \mathbb E[f(\calS_i)]| \ge \varepsilon
\right\}
\le 2\exp\!\Big\{\frac{-2\varepsilon^2 \sfN_c}{(-\log p_{\min})^2}\Big\}.
\]
Applying the union bound gives
\[\Pr\left\{ \max_i
|f(\calS_i)- \mathbb E[f(\calS_i)]| \ge \varepsilon
\right\}
= \Pr\bigg\{ \bigcup_{i=1}^{\sfN_t}
\{|f(\calS_i)- \mathbb E[f(\calS_i)]| \ge \varepsilon\}
\bigg\} \leq 
2\sfN_t \exp\!\Big\{\frac{-2\varepsilon^2\sfN_c}{(-\log p_{\min})^2}\Big\}.\]
This yields, for any $\delta_1\in(0,1)$, with probability at least \((1-\delta_1)\),
\begin{align}
   \max_i  \sup_{\bfw\in\calW_B}\big|\widehat{L}_{t_i}(\bfw)-L_{t_i}(\bfw)\big|
&\le \max_i \EE\Big[\sup_{\bfw\in\calW_B}\big|\widehat{L}_{t_i}(\bfw)-L_{t_i}(\bfw)\big|\Big]
+ (-\log p_{\min})\,\sqrt{\frac{2\log(2\sfN_t/\delta_1)}{\sfN_c}} \nonumber \\
&\leq \max_i 2\;\EE[\hat{\mathfrak R}_{\calS_i}(\calL_{\bfm}^{(i)})] + (-\log p_{\min})\,\sqrt{\frac{\log(2\sfN_t/\delta_1)}{2\sfN_c}},\label{eqn:unif_conv_1}
\end{align}
where the last inequality follows from \cite[Lemma 26.2]{shalev2014understanding}.
Next, we apply the bounded-difference property to the empirical Rademacher complexity. With $\calS'_i$ as stated above, we have 
\[
\begin{aligned}
\big|\hat{\mathfrak R}_{\calS_i}(\calL_{\bfm}^{(i)})-\hat{\mathfrak R}_{\calS'_i}(\calL_{\bfm}^{(i)})\big|
& = \Bigg|\mathbb{E}_{\sigma} \sup_{\bfw\in\calW_B}
\frac{1}{\sfN_c}
\sum_{k=1}^{\sfN_c} \sigma_{(i,k)} \ell(\phi(\bfm_{(i,k)}|t_i,\bfw))
\\
&\hspace{50pt}-
\mathbb{E}_{\sigma} \sup_{\bfw\in\calW_B}
\frac{1}{\sfN_c}
\bigg(
\sum_{ k \ne j} \sigma_{(i,k)} \ell(\phi(\bfm_{(i,k)}|t_i,\bfw)) + \sigma_{(i,j)} \ell(\phi(\bfm'_{(i,j)}|t_i,\bfw))
\bigg)\Bigg|\\
&\le
\mathbb{E}_{\sigma} \sup_{\bfw\in\calW_B}
\frac{1}{\sfN_c} \big|\sigma_{(i,j)} \big(\ell(\phi(\bfm_{(i,j)}|t_i,\bfw)) - \ell(\phi(\bfm'_{(i,j)}|t_i,\bfw))\big)\big|\\
&=
\frac{1}{\sfN_c} \sup_{\bfw\in\calW_B} |\ell(\phi(\bfm_{(i,j)}|t_i,\bfw)) - \ell(\phi(\bfm'_{(i,j)}|t_i,\bfw))|
\le
\frac{(-\log p_{\min})}{\sfN_c}.
\end{aligned}
\]
Thus, $\hat{\mathfrak R}_{\calS_i}(\calL_{\bfm}^{(i)})$ satisfies bounded differences
with $c_j=-\log p_{\min}/\sfN_c$. By McDiarmid’s inequality, for any $\varepsilon>0$
\[
\Pr\Big\{\big|\hat{\mathfrak R}_{\calS_i}(\calL_{\bfm}^{(i)})-\EE\big[\hat{\mathfrak R}_{\calS_i}(\calL_{\bfm}^{(i)})\big]\big| \geq \varepsilon\Big\}
\le 2\exp\bigg\{\frac{-2\varepsilon^2\sfN_c}{(-\log p_{\min})^2}\bigg\}.
\]
Applying the union bound, we get, 
\begin{align*}
    \Pr\Big\{\max_i \big|\hat{\mathfrak R}_{\calS_i}(\calL_{\bfm}^{(i)})-\EE\big[\hat{\mathfrak R}_{\calS_i}(\calL_{\bfm}^{(i)})\big]\big| \geq \varepsilon\Big\} 
    &\leq 2\sfN_t\exp\bigg\{\frac{-2\varepsilon^2\sfN_c}{(-\log p_{\min})^2}\bigg\}.
\end{align*}
Therefore, for any $\delta_2 \in (0,1)$, with probability at least $(1-\delta_2),$ we have
\[\max_i\EE\big[\hat{\mathfrak R}_{\calS_i}(\calL_{\bfm}^{(i)})\big] \leq \max_i\hat{\mathfrak R}_{\calS_i}(\calL_{\bfm}^{(i)})  + (-\log p_{\min})\,\sqrt{\frac{\log(2\sfN_t/\delta_2)}{2\sfN_c}}.\]
Finally, by letting $\delta_1=\delta_2=\delta$, we obtain
$$\max_{1\le i\le \sfN_t}\;
\sup_{\bfw\in\calW_B}
\rmT_1^{(i)}(\bfw)
\le
\max_i 2\;\hat{\mathfrak R}_{S_i}(\calL_{\bfm}^{(i)})
+3(-\log p_{\min})\sqrt{\frac{\log({2\sfN_t/\delta})}{2\sfN_c}},$$
which holds with probability at least $(1-2\delta)$. 
Equivalently, we can write that for any $\varepsilon>0$ the following inequality holds: 
\subsection{Derivation of Equation ~\eqref{eqn:unif_2}}
\label{app:proof:eqn:unif_2}
We use the well-known McDiarmids inequality \cite[Corollary 2.21]{wainwright2019high} \cite[Lemma 26.4]{shalev2014understanding}.
Define the function,
\[
f(\calT) = \sup_{\bfw \in \calW_B} \Big|\frac{1}{\sfN_t}\sum_{i=1}^{\sfN_t} {L}_{t_i}(\bfw) -L(\bfw)\Big|,
\]
Let $\calT'=\{t_i,\cdots,t'_j,\cdots,t_{\sfN_c}\}$ differ from $\calT$ in only one coordinate.
Then,
\[
\begin{aligned}
|f(\calT) - f(\calT')|
&\le \frac{1}{\sfN_t}\;\sup_{\bfw\in\calW_B}\big|\EE_{\bfm\sim \phi(\cdot|t_j,\bfw^*)}
[\ell(\phi(\bfm|t_j,\bfw))] - \EE_{\bfm\sim \phi(\cdot|t'_j,\bfw^*)}
[\ell(\phi(\bfm|t'_j,\bfw))]\big|\\
&\le \frac{1}{\sfN_t}\;\sup_{\bfw\in\calW_B}\big|\EE_{\bfm\sim \phi(\cdot|t_j,\bfw^*)}
[\ell(\phi(\bfm|t_j,\bfw))]\big| + \frac{1}{\sfN_t}\;\sup_{\bfw\in\calW_B}\big|\EE_{\bfm\sim \phi(\cdot|t'_j,\bfw^*)}
[\ell(\phi(\bfm|t'_j,\bfw))]\big|\\
&\le {2(-\log p_{\min})}/{\sfN_t}.
\end{aligned}
\]
Therefore, applying McDiarmid’s inequality and \cite[Lemma 26.2]{shalev2014understanding} gives, for any $\delta_1\in(0,1)$, with probability at least \((1-\delta_1)\),
\begin{align}
\sup_{\bfw\in\calW_B}\Big|\frac{1}{\sfN_t}\sum_{i=1}^{\sfN_t}{L}_{t_i}(\bfw)-L(\bfw)\Big|
&\leq 2\;\EE[\hat{\mathfrak R}_{\calT}(\calL_{t})] + (-\log p_{\min})\,\sqrt{\frac{\log(2/\delta_1)}{2\sfN_t}},\label{eqn:unif_conv_2}
\end{align}
Next, we apply the bounded-difference property to the empirical Rademacher complexity. With $\calS'_i$ as stated above, we have 
\[
\begin{aligned}
\big|\hat{\mathfrak R}_{\calT}(\calL_t)-\hat{\mathfrak R}_{\calT'}(\calL_t)\big|
& = \Bigg|\mathbb{E}_{\sigma} \sup_{\bfw\in\calW_B}
\frac{1}{\sfN_t}
\sum_{i=1}^{\sfN_t} \sigma_{i} L_{t_i}(\bfw)
-
\mathbb{E}_{\sigma} \sup_{\bfw\in\calW_B}
\frac{1}{\sfN_t}
\bigg(
\sum_{ i \ne j} \sigma_{i} L_{t_i}(\bfw) + \sigma_{j} L_{t'_j}(\bfw)
\bigg)\Bigg|\\
&\le
\mathbb{E}_{\sigma} \sup_{\bfw\in\calW_B}
\frac{1}{\sfN_t} \big|\sigma_{j} \big(L_{t_j}(\bfw) - L_{t'_j}(\bfw)\big)\big|\le{2(-\log p_{\min})}/{\sfN_t}.
\end{aligned}
\]
Thus, $\hat{\mathfrak R}_{\calT}(\calL_t)$ satisfies bounded differences
with $c_j=-\log p_{\min}/\sfN_c$. By McDiarmid’s inequality \cite[Corollary 2.21]{wainwright2019high} \cite[Lemma 26.4]{shalev2014understanding}, we obtain, for any $\delta_2 \in (0,1)$, with probability at least $(1-\delta_2),$ we have
\[\EE\big[\hat{\mathfrak R}_{\calT}(\calL_t)\big] \leq \hat{\mathfrak R}_{\calT}(\calL_t)  + (-\log p_{\min})\,\sqrt{\frac{2\log(2/\delta_2)}{\sfN_c}}.\]
Finally, by letting $\delta_1=\delta_2=\delta$, we obtain
\[\sup_{\bfw\in\calW_B}\rmT_2(\bfw)
\le
2\,\hat{\mathfrak R}_\calT(\calL_{t})
\;+\;
3(-\log p_{\min})\sqrt{\frac{2\log(2/\delta)}{\sfN_t}},\]
which holds with probability at least $(1-2\delta)$. 

\subsection{Proof of Proposition \ref{prop:rademacher_1}.}
\label{app:proof:prop:rademacher_1}
Using the definition of empirical Rademacher complexity and the contraction lemma \cite[Lemma 26.9]{shalev2014understanding}, we get
\begin{align*}
\hat{\mathfrak R}_{\calS_i}(\calL_{\bfm}^{(i)})
&=\mathbb E_\sigma\Big[\sup_{\bfw\in\calW_B} \frac{1}{\sfN_c}\sum_{k=1}^{\sfN_c} \sigma_{(i,k)}\,\ell(\phi(\bfm_{(i,k)}|t_i,\bfw))\Big]{\leq}\frac{1}{p_{\min}}\mathbb E_\sigma\Big[\sup_{\bfw\in\calW_B}\frac{1}{\sfN_c}\sum_{k=1}^{\sfN_c} \sigma_{(i,k)}\,\tr\big(\Lambda_{\bfm_{(i,k)}}\rho_{t_i}(\bfw)\big)\Big],
\end{align*}
where the inequality follows by noting that $\ell$ is $1/p_{\min}$-Lipschitz.
Next, by grouping the terms within the trace, we get
\[
\sum_{k}\sigma_{(i,k)}\,\tr\big(\Lambda_{\bfm_{(i,k)}}\rho_{t_i}(\bfw)\big)
=\tr\Big(\Big(\sum_{k}\sigma_{(i,k)}\,\Lambda_{\bfm_{(i,k)}}\Big)\rho_{t_i}(\bfw)\Big).
\]
Applying the inequality $|\tr(\rmA^\dagger \rmB)|\le \|\rmA\|\|\rmB\|_1$ \cite[Eqn.~1.174]{watrous2018theory}, gives
\begin{align*}
    \sup_{\bfw\in\calW_B}\tr\Big(\Big(\sum_{k}\sigma_{(i,k)}\,\Lambda_{\bfm_{(i,k)}}\Big)\rho_{t_i}(\bfw)\Big) &\leq \sup_{\bfw\in\calW_B}\Big\|\sum_{k}\sigma_{(i,k)}\,\Lambda_{\bfm_{(i,k)}}\Big\|\|\rho_{t_i}(\bfw)\|_1= \frac{1}{\sfN_c}\Big\|\sum_{k}\sigma_{(i,k)}\,\Lambda_{\bfm_{(i,k)}}\Big\|.
\end{align*}
Taking expectation and using
Lemma~\ref{lem:khintchine} for each $t$, we get
\[
\hat{\mathfrak R}_{\calS_i}(\calL_{\bfm}^{(i)})
\le \frac{\sqrt{2\log(2\sfD)}}{p_{\min}\sfN_c}\Big\|\sum_{k}\Lambda_{\bfm_{(i,k)}}^2\Big\|^{1/2}.
\]
Finally, using triangle inequality, $\sqrt{\|\sum_{k}\Lambda_{\bfm_{(i,k)}}^2\|}\leq \sqrt{\sum_{k}\|\Lambda_{\bfm_{(i,k)}}^2\|}\le \sqrt{\sfN_c}$
since $0\le\Lambda_\bfm\le \rmI$. Thus,
\[
\hat{\mathfrak R}_{\calS_i}(\calL_{\bfm}^{(i)})\leq \frac{\sqrt{2\log(2\sfD)}}{p_{\min}\sfN_c}\sqrt{\sfN_c} = 
\frac{\sqrt{2\log(2\sfD)}}{p_{\min}\sqrt{\sfN_c}}.\]
This completes the proof of Proposition \ref{prop:rademacher_1}.

\subsection{Proof of Proposition \ref{prop:rademacher_2}.}
\label{app:proof:prop:rademacher_2}
Recall, for a metric space $(\calX,d)$, the covering number, denoted as $\calN(\varepsilon,\calX,d)$, is the minimum number of balls of radius $\varepsilon$ required to cover the given space $\calX$ using metric $d$. Note that for $c$-dimensional Euclidean ball of radius $B$, the covering $\calN(\eta,\calW_B, \|\cdot\|_2)$ is upper bounded as: $\calN(\eta,\calW_B, \|\cdot\|_2)\leq \big(\frac{3B}{\eta}\big)^c$. Next, on the sample $\calT=\{t_1,\dots,t_{\sfN_t}\}$, define a data-dependent pseudo-metric on function $L_{t_i}$
\[
d_\calT(\bfw,\bfw'):=\left(\frac{1}{\sfN_t}\sum_{i=1}^{\sfN_t}\big(L_{t_i}(\bfw)-L_{t_i}(\bfw')\big)^2\right)^{1/2}.
\]
Using Eqn.\ref{eqn:lipchitz_Lt}, we get $d_\calT(\bfw,\bfw')\le K\|\bfw-\bfw'\|$, where $K:=(L_\phi/p_{\min})$. Then, the covering number of the function class $\calL_{t}$ with respect to pseudo-metric $d_{\calT}$ is bounded above as 
\[
\calN(\varepsilon,\calL_{t},d_\calT)
\le
\calN(\varepsilon/K,\calW_B,\|\cdot\|_2)
\le
\left(\frac{3KB}{\varepsilon}\right)^c.
\]
Finally, we apply Dudley's theorem, which we restate here for convenience. For more details please refer to \cite[Example 5.24]{wainwright2019high}
\begin{theorem*}[Dudley’s Theorem]
Let $\mathcal F$ be a class of real-valued functions from
$S=\{z_1,\dots,z_n\}$ to $\RR$. Define the empirical pseudo-metric
$d_S(f,g)
:=\sqrt{\frac{1}{n}\sum_{i=1}^n \big(f(z_i)-g(z_i)\big)^2}.$
Then the empirical Rademacher complexity of $\mathcal F$ satisfies
\[
\hat{\mathfrak R}_S(\mathcal F)
\;\le\;
\frac{24}{\sqrt{n}}
\int_0^{\mathrm{diam}(\mathcal F,d_S)}
\sqrt{\log \mathcal N(\varepsilon,\mathcal F,d_S)}\,d\varepsilon,
\]
where
$\mathrm{diam}(\mathcal F,d_S)
:=\sup_{f,g\in\mathcal F}\; d_S(f,g),$
and $\mathcal N(\varepsilon,\mathcal F,d_S)$ denotes the $\varepsilon$-covering
number of $\mathcal F$ with respect to the metric $d_S$.
\end{theorem*}
In our setting, the function class is $\mathcal L_t$ and the empirical metric is $d_{\mathcal T}$. The radius can be bounded as $\mathrm{diam}(\calL_t,d_\calT) \leq 2(-\log p_{\min}).$ With this bound in hand, define $L_{\max}:=\min\{3KB,2(-\log p_{\min})\}$. We now derive the following sequence of inequalities.
\begin{align*}
    \hat{\mathfrak R}_\calT(\calL_t) \leq \frac{24}{\sqrt{\sfN_t}}\int_{0}^{L_{\max}} \!\!\sqrt{c\log\bigg(\frac{3KB}{\varepsilon}\bigg)} \mathrm{d} \varepsilon &= \frac{72KB\sqrt{c}}{\sqrt{\sfN_t}}\int_{\log\big(\frac{3KB}{L_{\max}}\big)}^{\infty} \sqrt{u} e^{-u} \mathrm{d}u \\
    &\leq \frac{72KB\sqrt{c}}{\sqrt{\sfN_t}}\int_{0}^{\infty} \sqrt{u} e^{-u} \mathrm{d}u 
    =  \frac{36L_{\phi}B\sqrt{\pi c}}{p_{\min}\sqrt{\sfN_t}}.
\end{align*}
This completes the proof of Proposition \ref{prop:rademacher_2}.


\subsection{Proof of Lemma \ref{lem:khintchine}.}
\label{app:proof:lem:khintchine}
Let $\Upsilon := \sum_{k} \sigma_k \Lambda_k$. We can express the spectral norm of the Hermitian matrix $\Upsilon$ as follows:
\begin{align*}
    \|\Upsilon\| = \max\{\lambda_{\max}(\Upsilon),\lambda_{\max}(-\Upsilon)\}&\overset{a}{\leq} \frac{1}{\theta} \log(\exp\{\theta \lambda_{\max}(\Upsilon)\} + \exp\{\theta\lambda_{\max}(-\Upsilon)\}) \;\text{for all} \;\theta >0\\
    &\overset{b}{\leq} \frac{1}{\theta} \log(\tr(\exp\{\theta \Upsilon\}) + \tr(\exp\{-\theta \Upsilon\})),
\end{align*}
where $\lambda_{\max}(\pm\Upsilon)$ denotes the maximum eigenvalue of $\pm\Upsilon$, $(a)$ follows from the LogSumExp inequality \cite[Eqn. 4]{nielsen2016guaranteed} stated below: 
$$\max\{x_1,x_2,\cdots,x_m\} \leq \frac{1}{\theta}\log\Big(\sum_{i=1}^m \exp{(\theta x_i)}\Big) \leq \max\{x_1,x_2,\cdots,x_m\} + \frac{\log{n}}{\theta} \text{ for all } x_i \in \RR, \theta>0,$$
and $(b)$ follows because $\exp\{\theta \lambda_{\max}(\Upsilon)\} \leq \tr(\exp\{\theta\Upsilon\})$.
Next, we take the expectation with respect to Rademacher variables on both sides.
\begin{align*}
    \EE_{\sigma}[ \|\Upsilon\|] &\overset{c}{\leq}  \frac{1}{\theta} \log(\EE[\tr\exp\{\theta \Upsilon\})] + \EE[\tr\exp\{-\theta \Upsilon\})])\\
    &\overset{d}{\leq}  \frac{1}{\theta} \log\Big(\tr\exp\Big\{\sum_k \log \EE[e^{\theta \sigma_k\Lambda_k}]\Big\} + \tr\exp\Big\{\sum_k \log \EE[e^{-\theta \sigma_k\Lambda_k}]\Big\}\Big)\\
    &\overset{e}{\leq}  \frac{1}{\theta} \log\Big(2  \tr\exp\Big\{\tfrac{1}{2}\sum_k\theta^2\Lambda_k^2\Big\}\Big)\overset{f}{\leq} \frac{1}{\theta} \log\Big(2\sfD\exp\Big\{\frac{\theta^2}{2}\Big\|\sum_k\Lambda_k^2\Big\|\Big\}\Big)=\frac{1}{\theta} \log(2\sfD) + \frac{\theta}{2}\Big\|\sum_k\Lambda_k^2\Big\|,
\end{align*}
where $(c)$ follows from Jensen's inequality, $(d)$ follows from sub-additivity of matrix cumulant generating function \cite[Lemma 3.1]{tropp2015introduction}, $(e)$ is derived using the following results:
\begin{itemize}
    \item (\textbf{Lemma 4.2} \cite{tropp2015introduction}.) 
        Suppose $A$ is a fixed Hermitian matrix and $\sigma$ is a Rademacher variable. Then, for $\theta \in \RR$, we have 
        $\EE_\sigma[e^{\theta\sigma A}] \leq \exp\{\theta^2A^2/2\} \ \text{and}\ \log(\EE_\sigma[e^{\theta\sigma A}]) \leq \theta^2A^2/2.$
        \item For Hermitian matrices $A_1,A_2,B_1,B_2$, if $A_1\leq B_1$ and $A_2\leq B_2$, then $A_1+A_2 \leq B_1+B_2$.
    \item $\tr\exp$ (trace-exponential) function is monotone with respect to the semi-definite order \cite{tropp2015introduction}.
\end{itemize}
and $(f)$ follows from the inequality $\tr(e^A) \leq d \cdot \exp\{\lambda_{\max}(A)\}$, which holds for all $d\times d$ Hermitian matrix $A$. Additionally, since $\Lambda_k^2$ is a positive semi-definite matrix for each $k$, it follows that $\lambda_{\max}(\sum_k \Lambda_k^2) = \|\sum_k \Lambda_k^2\|$. Now, by minimizing the right-hand side of the inequality with respect to $\theta>0$, we get 
$\theta^\star = \sqrt{{2\log(2\sfD)}/{\|\sum_k \Lambda_k^2\|}}.$
Finally, this leads us to conclude that \[\EE_\sigma\bigg[\Big\|\sum_k\sigma_k \Lambda_k\Big\|\bigg] \leq \sqrt{2\log(2\sfD)} \Big\|\sum_k \Lambda_k^2\Big\|^{1/2}.\]

\subsection{Proof of Lemma \ref{lem:grad-hess-exp}.}
\label{app:proof:lem:grad-hess-exp}
Following \cite[Theorem 3]{haber2018notes}, we recall the parametric derivative formula for the operator exponential: \begin{equation}\label{eqn:parametric_deriv}
    \frac{\partial}{\partial w_k} e^{-\rmA(\bfw)} = -\int_0^1 e^{-(1-s)\rmA(\bfw)} \bigg( \frac{\partial \rmA(\bfw)}{\partial w_k} \bigg)e^{-s\rmA(\bfw)} \ds,
\end{equation}
where $\rmA(\bfw)$ is a parameterized Hermitian operator. Using \eqref{eqn:parametric_deriv}, we compute the derivative of $U_t(\bfw) = e^{-it\rmH(\bfw)}$ with respect to the parameter $w_k$ as follows:
\begin{align*}
    \frac{\partial}{\partial w_k} U_t(\bfw)&= 
    -\int_0^1 e^{-(1-s)( \imag t\rmH(\bfw))} \ \frac{\partial}{\partial w_k}( \imag t\rmH(\bfw))\ e^{-s(\imag t\rmH(\bfw))} \ds
    \\
   &= 
    -\imag t \int_0^1 e^{-(1-s)(\imag t\rmH(\bfw))}\ \rmH_k\ e^{-s(\imag t\rmH(\bfw))} \ds
   =
-\imag \int_0^t U_t(\bfw)\, \rmH_k(s)\, \ds,
\end{align*}
where the final equality follows from the change of variables 
$st\rightarrow s$. Applying the adjoint operation to both sides yields
\[
\frac{\partial}{\partial w_k} U_t^\dagger(\bfw)
=
+\imag \int_0^t \rmH_k(s) U_t^\dagger(\bfw)\, \, \ds .
\]
Now, we compute the derivative of time-evolved  state $\rho_t(\bfw) = U_t \rho_0 U_t^\dagger$ as follows:
\begin{align}
    \frac{\partial \rho_t}{\partial w_k}
&=
\bigg(\frac{\partial U_t}{\partial w_k}\bigg)\rho_0 U_t^\dagger
\;+\;
U_t \rho_0 \bigg(\frac{\partial U_t^\dagger}{\partial w_k}\bigg) \nonumber\\
&= -\imag \int_0^t U_t(\bfw)\, \rmH_k(s) \rho_0 U_t^\dagger(\bfw)\,\ds +\imag \int_0^t U_t(\bfw)\rho_0 \rmH_k(s) U_t^\dagger(\bfw)\, \ds =
-\imag \int_0^t U_t \big[\rmH_k(s),\rho_0\big] U_t^\dagger\, \ds. \label{app:eqn:deriv_rhot}
\end{align}
Therefore,
\[
\frac{\partial}{\partial w_k}\tr(\Lambda_\bfm \rho_t)
=
-\imag \int_0^t \tr\big(\Lambda_\bfm\, U_t \big[\rmH_k(s),\rho_0\big] U_t^\dagger\big)\, \ds
=
-\imag \int_0^t \tr\big(\Lambda_\bfm(t)\, \big[\rmH_k(s),\rho_0\big]\big)\, \ds,
\]
where we used cyclicity of trace and $\Lambda_\bfm(t):=U_t^\dagger \Lambda_\bfm U_t$. Next, we compute the second-order partial derivative as follows:
\begin{align*}
    \frac{\partial^2}{\partial w_j\partial w_k}\tr(\Lambda_\bfm \rho_t)&= \underbrace{-\imag \int_0^t \tr\Big(\Big(\frac{\partial}{\partial w_j}\Lambda_\bfm(t)\Big)\, \big[\rmH_k(s),\rho_0\big]\Big)\, \ds}_{\rmT_1} \underbrace{-\imag \int_0^t \tr\Big(\Lambda_\bfm(t)\, \Big[\Big(\frac{\partial}{\partial w_j}\rmH_k(s)\Big),\rho_0\Big]\Big)\, \ds}_{\rmT_2}.
\end{align*}
By applying steps analogous to the derivative of the time-evolved state $\rho_t(\mathbf{w})$ \eqref{app:eqn:deriv_rhot}, we derive the following:
\begin{equation}
    \frac{\partial}{\partial w_j}\Lambda_\bfm(t) = \imag \int_0^t \big[\rmH_j(s),\Lambda_\bfm(t)\big] \, \ds\quad \eqand \quad \frac{\partial}{\partial w_j}\rmH_k(s) = \imag \int_0^s \big[\rmH_j(s_1),\rmH_k(s)\big] \, \ds_1. \label{app:eqn:deriv_Lambdam}
\end{equation}
This gives, 
\begin{align}
    \rmT_1
    &=   \int_0^t \int_0^t \tr\big(\big[\rmH_j(s_2),\Lambda_\bfm(t)\big]\, \big[\rmH_k(s_1),\rho_0\big]\big)\, \ds_2\, \ds_1\nonumber\\
    &\overset{a}{=}- \int_0^t \int_0^t \tr\big(\Lambda_\bfm(t)\big[\rmH_j(s_2),\big[\rmH_k(s_1),\rho_0\big]\big]\, \big)\, \ds_2\, \ds_1\nonumber\\
    &\overset{b}{=}- \int_0^t \int_0^{s_1} \Big(\tr\big(\Lambda_\bfm(t)\big[\rmH_j(s_1),\big[\rmH_k(s_2),\rho_0\big]\big]\,\big) + \tr\big(\Lambda_\bfm(t)\big[\rmH_j(s_2),\big[\rmH_k(s_1),\rho_0\big]\big]\,\big)\Big)\, \ds_2\, \ds_1,
    \label{app:eqn:termA_secondderiv}
\end{align}
where $(a)$ follows from the fact that the trace of the product of two commutators is given by the equation $\tr([\rmA,\rmB][\rmC,\rmD]) = -\tr(\rmB[\rmA,[\rmC,\rmD]])$ and $(b)$ is derived from the fact that for any function \( f(s_1, s_2) \), the square integral over the region \([0, t]^2\) can be expressed as follows:
\begin{align*}
    \int_0^t\int_0^t f(s_1,s_2)\, \ds_2\, \ds_1 &=\int_{0\leq s_2\leq s_1\leq t} f(s_1,s_2)\, \ds_2\, \ds_1 + \int_{0\leq s_1<s_2\leq t} f(s_1,s_2)\, \ds_2\, \ds_1 \\
    &=\int_{0\leq s_2\leq s_1\leq t} (f(s_1,s_2) + f(s_2,s_1))\, \ds_2\, \ds_1 \\
    &= \int_{0}^t  \int_0^{s_1}(f(s_1,s_2) + f(s_2,s_1))\, \ds_2\, \ds_1.
\end{align*}
\begin{align}
    \rmT_2 
    &= \int_0^t \int_0^{s_1}\tr\big(\Lambda_\bfm(t)\, \big[ \big[\rmH_j(s_2),\rmH_k(s_1)\big] \, ,\rho_0\big]\big)\, \ds_2\,\ds_1\nonumber\\
    &\overset{c}{=} \int_0^t \int_0^{s_1}\Big(\tr\big(\Lambda_\bfm(t)\, \big[  \rmH_j(s_2),\big[\rmH_k(s_1) ,\rho_0\big]\big]\big) - \tr\big(\Lambda_\bfm(t)\, \big[\rmH_k(s_1),\big[\rmH_j(s_2) ,\rho_0\big]\big]\big)\Big)\, \ds_2\,\ds_1,
    \label{app:eqn:termB_secondderiv}
\end{align}
where $(c)$ follows from the use of Jacobi identity $[[\rmA,\rmB],\rmC]=[\rmA,[\rmB,\rmC]]-[\rmB,[\rmA,\rmC]].$ Adding $\rmT_1$ \eqref{app:eqn:termA_secondderiv} and $\rmT_2$ \eqref{app:eqn:termB_secondderiv}, we get  
 \begin{align*}
     \frac{\partial^2}{\partial w_j\partial w_k}\tr(\Lambda_\bfm \rho_t) &= -\int_0^t \int_0^{s_1}\tr\big(\Lambda_\bfm(t)\big[\rmH_j(s_1),\big[\rmH_k(s_2),\rho_0\big]\big]\,\big)\, \ds_2\,\ds_1 \\
     &\hspace{50pt}-\int_0^t \int_0^{s_1} \tr\big(\Lambda_\bfm(t)\, \big[\rmH_k(s_1),\big[\rmH_j(s_2) ,\rho_0\big]\big]\big)\, \ds_2\,\ds_1.
 \end{align*}

This completes the first part of Lemma \ref{lem:grad-hess-exp}. Next, we show gradient and Hessian are bounded. Consider the following inequalities:
\begin{align*}
    \Big|\frac{\partial }{\partial {w_k}} \phi{(\bfm|t,\bfw)}\Big| 
    &\overset{}{=} \Big|-\imag\int_0^t \tr(\Lambda_{\bfm}(t)[\rmH_k(s),\rho_0])\ds\Big|\\
    &\overset{}{\leq} \int_0^t\big|\tr(\Lambda_{\bfm}(t)[\rmH_k(s),\rho_0])\big|\ds\quad \quad \quad \quad \text{using triangle inequality}\\
    &\overset{a}{\leq} \int_0^t \|\Lambda_{\bfm}(t)\|\|[\rmH_k(s),\rho_0]\|_1\ds \\
    &\overset{b}{\leq} \int_0^t \|[\rmH_k(s),\rho_0]\|_1\ds\, \overset{c}{\leq} 2\int_0^t \|\rmH_k(s)\|\ \|\rho_0\|_1\ds \leq 2t_{\max}\|\rmH_k\|,
\end{align*}
where $(a)$ follows from the inequality $|\tr(\rmA^\dagger \rmB)|\le \|\rmA\|\|\rmB\|_1$ \cite[Eqn.~1.174]{watrous2018theory}, $(b)$ is derived from the isometric invariance property of spectral norm, i.e., \( \|\Lambda_{\bfm}(t)\| = \|\Lambda_{\bfm}\| \), along with the fact that \( \|\Lambda_{\bfm}\| \leq 1 \), and $(c)$ follows from Hölder's inequality, which states that for \( \frac{1}{p} + \frac{1}{q} = 1 \), we have \( \|\rmA\rmB\|_1 \leq \|\rmA\|_p \|\rmB\|_q \) (see \cite[Eqn.~1.175]{watrous2018theory}). 
Therefore, we get,
\[\|\nabla_\bfw \phi(\bfm|t,\bfw)\|_2 \leq \Big(4t_{\max}^2 \sum_{k=1}^c \|\rmH_k\|^2\Big)^{1/2} \leq 2t_{\max}\sqrt{c}\;\|\rmH\|_{\infty,\max}.\]
Next, by applying steps analogous to those for the derivative of the likelihood function, we bound its Hessian. Consider the following inequalities:
\begin{align*}
     \Big|\frac{\partial^2 }{\partial {w_j}\partial {w_k}} \phi{(\bfm|t,\bfw)}\Big| &\overset{}{\leq} \Big|\int_0^t \int_0^{s_1}\tr\big(\Lambda_\bfm(t)\big[\rmH_j(s_1),\big[\rmH_k(s_2),\rho_0\big]\big]\,\big)\, \ds_2\,\ds_1\Big| \\
     &\hspace{50pt} + \Big|\int_0^t \int_0^{s_1} \tr\big(\Lambda_\bfm(t)\, \big[\rmH_k(s_1),\big[\rmH_j(s_2) ,\rho_0\big]\big]\big)\, \ds_2\,\ds_1\Big| \\
     &\overset{}{\leq} \int_0^t \int_0^{s_1}\Big|\tr\big(\Lambda_\bfm(t)\big[\rmH_j(s_1),\big[\rmH_k(s_2),\rho_0\big]\big]\,\big)\Big|\, \ds_2\,\ds_1 \\
     &\hspace{50pt} + \int_0^t \int_0^{s_1} \Big|\tr\big(\Lambda_\bfm(t)\, \big[\rmH_k(s_1),\big[\rmH_j(s_2) ,\rho_0\big]\big]\big)\Big|\, \ds_2\,\ds_1 \\
     &\overset{a}{\leq} \int_0^t \int_0^{s_1}\|\Lambda_\bfm(t)\|\Big\|\big[\rmH_j(s_1),\big[\rmH_k(s_2),\rho_0\big]\big]\Big\|_1\, \ds_2\,\ds_1 \\
     &\hspace{50pt} + \int_0^t \int_0^{s_1}\|\Lambda_\bfm(t)\|\Big\|\big[\rmH_k(s_1),\big[\rmH_j(s_2),\rho_0\big]\big]\Big\|_1\, \ds_2\,\ds_1\\
     &\overset{b}{\leq} 8\int_0^t \int_0^{s_1}\|\rmH_k\| \|\rmH_j\|  \ds_2\,\ds_1  \leq  4t_{\max}^2\|\rmH_k\| \|\rmH_j\|
\end{align*}
where $(a)$ follows from \cite[Eqn.~1.174]{watrous2018theory}, and $(b)$ derives from the matrix Hölder inequality applied to the commutator, leading to $\|[\rmA,\rmB]\| \leq 2\|\rmA\|\; \|\rmB\|_1$. Thus, we obtain 
\[\|\nabla^2_{\bfw} \phi(\bfm|t,\bfw)\| \leq \|\nabla^2_{\bfw} \phi(\bfm|t,\bfw)\|_\rmF= 4t_{\max}^2\Big(\sum_{j,k} \|\rmH_k\|^2\|\rmH_j\|^2\Big)^{1/2} \leq 4t_{\max}^2 c \|\rmH\|^2_{\infty,\max},\]
where $\|\cdot\|_\rmF$ is the Frobenius norm. 

Since the gradient and the Hessian of the likelihood function are bounded, the likelihood function and its gradient are Lipschitz continuous by definition. Finally, we show the Lipschitz continuity of the Hessian. To this end, define $\rmH_k'(s):= U_s^\dagger(\bfw')\, \rmH_k\, U_s(\bfw')$.
Consider the following inequalities:
\begin{align*}
    \Big|&\frac{\partial^2 }{\partial {w_j}\partial {w_k}} \phi{(\bfm|t,\bfw)}- \frac{\partial^2 }{\partial {w_j}\partial {w_k}} \phi{(\bfm|t,\bfw')}\Big| \\
    &\leq \underbrace{\int_0^t \int_0^{s_1}\Big|\tr\big\{\Lambda_\bfm\big(U_t(\bfw)\big[\rmH_k(s_1),\big[\rmH_j(s_2),\rho_0\big]\big]U_t^\dagger(\bfw)-U_t(\bfw')\big[\rmH_k'(s_1),\big[\rmH_j'(s_2),\rho_0\big]\big]U_t^\dagger(\bfw')\big)\big\}\Big| \ds_2\,\ds_1}_{\rmA_1}\\
    &+ \underbrace{\int_0^t \int_0^{s_1}\Big|\tr\big\{\Lambda_\bfm\big(U_t(\bfw)\big[\rmH_j(s_1),\big[\rmH_k(s_2),\rho_0\big]\big]U_t^\dagger(\bfw)-U_t(\bfw')\big[\rmH_j'(s_1),\big[\rmH_k'(s_2),\rho_0\big]\big]U_t^\dagger(\bfw')\big)\big\}\Big| \ds_2\,\ds_1}_{\rmA_2}.
\end{align*}
Since $\rm A_1$ and $\rm A_2$ are similar except that the indices $j$ and $k$ are reversed, if we set bounds on one, we can similarly bound the other. Without loss of generality, consider $\rmA_1$, and then apply \cite[Eqn.~1.174]{watrous2018theory}. 
\begin{align}
    \rmA_1 &\overset{}{\leq} \int_0^t \int_0^{s_1} \|U_t(\bfw)[\rmH_k(s_1),[\rmH_j(s_2),\rho_0]]U_t^\dagger(\bfw)-U_t(\bfw')[\rmH_k'(s_1),[\rmH_j'(s_2),\rho_0]]U_t^\dagger(\bfw')\|\ds_2\,\ds_1 \nonumber\\
    &\overset{a}{\leq} \int_0^t \int_0^{s_1} \|U_t(\bfw)[\rmH_k(s_1),[\rmH_j(s_2),\rho_0]]U_t^\dagger(\bfw)-U_t(\bfw')[\rmH_k(s_1),[\rmH_j(s_2),\rho_0]]U_t^\dagger(\bfw')\|\ds_2\,\ds_1 \nonumber\\
     &\hspace{25pt} + \int_0^t \int_0^{s_1} \|[\rmH_k(s_1),[\rmH_j(s_2),\rho_0]] - [\rmH_k'(s_1),[\rmH_j'(s_2),\rho_0]]\|_1\; \ds_2\,\ds_1\nonumber\\
    &\overset{b}{\leq} \int_0^t \int_0^{s_1} 2\|U_t(\bfw)-U_t(\bfw') \| \; \big\|[\rmH_k(s_1),[\rmH_j(s_2),\rho_0]]\big\|_1\; \ds_2\,\ds_1 \nonumber \\
    &\hspace{25pt} + \int_0^t \int_0^{s_1} \Big(\big\|[\rmH_k(s_1)-\rmH_k'(s_1),[\rmH_j(s_2),\rho_0]]\big\|_1 + \big\|[\rmH_k'(s_1),[\rmH_j(s_2),\rho_0]-[\rmH_j'(s_2),\rho_0]]\big\|_1 \Big)\; \ds_2\,\ds_1,\nonumber
\end{align}
where $(a)$ follows by adding and subtracting $U_t(\bfw')[\rmH_k(s_1),[\rmH_j(s_2),\rho_0]]U_t^\dagger(\bfw')$ and $(b)$ follows by first adding and subtracting $U_t(\bfw)[\rmH_k(s_1),[\rmH_j(s_2),\rho_0]]U_t^\dagger(\bfw')$ and then using \cite[Eqn. 1.175]{watrous2018theory}. Next, applying the identity $e^A-e^B = \int_{0}^1 e^{sA} (A-B) e^{(1-s)B} \ds$ \cite[Eqn.43]{haber2018notes}, we obtain 
\begin{equation}
    \|U_t(\bfw)-U_t(\bfw')\| \leq t\|\rmH(\bfw)-\rmH(\bfw')\| \leq t\|\rmH\|_{\infty,\max}\; \|\bfw-\bfw'\|_1,  \label{eqn:Ut_diff_infty}
\end{equation}
and using \cite[Eqn.1.175]{watrous2018theory} for $p=\infty$ yields
\begin{equation}
    \|\rmH_k(s)-\rmH_k'(s)\| \leq 2\|\rmH_j\| \; \|U_s(\bfw) - U_s(\bfw')\| \leq 2s\|\rmH\|_{\infty,\max}^2\; \|\bfw-\bfw'\|_1. \label{eqn:Hs_diff_infty}
\end{equation}
Using \eqref{eqn:Ut_diff_infty} and \eqref{eqn:Hs_diff_infty} along with matrix Hölder inequality applied to the commutator, we get 
\begin{align*}
    \rmA_1 & \leq \int_0^t \int_0^{s_1} 8t\|\rmH\|_{\infty,\max}^3 \|\bfw-\bfw'\|_1 \; \ds_2\,\ds_1 + \int_0^t \int_0^{s_1} 8(s_1+s_2)\|\rmH\|_{\infty,\max}^3 \|\bfw-\bfw'\|_1\; \ds_2\,\ds_1\\
    &\leq 8t_{\max}^3\|\rmH\|_{\infty,\max}^3 \|\bfw-\bfw'\|_1.
\end{align*}
Therefore, 
\begin{align*}
 \Big|\frac{\partial^2 }{\partial {w_j}\partial {w_k}} \phi{(\bfm|t,\bfw)}- \frac{\partial^2 }{\partial {w_j}\partial {w_k}} \phi{(\bfm|t,\bfw')}\Big| &\leq 16 t_{\max}^3\|\rmH\|_{\infty,\max}^3 \|\bfw-\bfw'\|_1\\
\hspace{-10pt}\implies \|\nabla^2_{\bfw} \phi(\bfm|t,\bfw) -\nabla^2_{\bfw} \phi(\bfm|t,\bfw')\| &\leq \|\nabla^2_{\bfw} \phi(\bfm|t,\bfw) -\nabla^2_{\bfw} \phi(\bfm|t,\bfw')\|_{\rmF} \\
&\leq 16 t_{\max}^3 c^2\|\rmH\|_{\infty,\max}^3 \|\bfw-\bfw'\|_1
\\
&\leq 16 t_{\max}^3 c^{3/2}\|\rmH\|_{\infty,\max}^3 \|\bfw-\bfw'\|_2 = 2L_{\phi}^3 \|\bfw-\bfw'\|_2,
\end{align*}
where the last inequality follows from the inequality $\|x\|_1 \leq \sqrt{n}\|x\|$ for $x\in\RR^n$. 
This completes the proof of Lemma \ref{lem:grad-hess-exp}.

\putbib[reference]
\end{bibunit}
\end{document}